\documentclass[10pt]{article}
\usepackage{amsmath,amssymb,amsrefs}
\usepackage{amsthm}
\usepackage{graphicx}
\usepackage{caption}
\usepackage{subcaption}
\usepackage{bm}
\usepackage{bbm}
\usepackage{hyperref} 
\usepackage{xcolor}
\usepackage{tikz}
\usetikzlibrary{shapes,arrows,positioning,calc}
\usepackage{pgfplots}
\usepgfplotslibrary{fillbetween}
\pgfmathdeclarefunction{gauss}{2}{%
\pgfmathparse{1/(#2*sqrt(2*pi))*exp(-((x-#1)^2)/(2*#2^2))}%
}
\usepackage{cite}

\hypersetup{
    colorlinks,
    linkcolor={blue},
    citecolor={green!50!black},
    urlcolor={blue!80!black}
}

\providecommand{\keywords}[1]
{
  \small	
  \textbf{\textit{Keywords---}} #1
}

\usepackage[mathscr]{euscript}
\usepackage{arydshln}
\usepackage{enumerate}
\usepackage{enumitem}
\usepackage{dsfont,textcomp}
\usepackage{float}
\usepackage{multirow}
\usepackage{pgf,tikz}
\usepackage{ulem,cancel}
\numberwithin{equation}{section}
\usepackage{csquotes}
\usepackage{stackengine,scalerel}
\usepackage{enumitem}

\newtheorem{theorem}{Theorem}[section]
\newtheorem{corollary}[theorem]{Corollary}
\newtheorem{proposition}[theorem]{Proposition}
\newtheorem{lemma}[theorem]{Lemma}

\newtheorem{definition}[theorem]{Definition}
\theoremstyle{remark}
\newtheorem{remark}[theorem]{Remark}

\usetikzlibrary{shapes}
\usetikzlibrary{plotmarks}
\usetikzlibrary{shapes.misc}
\tikzset{cross/.style={cross out, draw=black, minimum size=5*(#1-\pgflinewidth), inner sep=0pt, outer sep=0pt},
cross/.default={3pt}}
\newcommand{\ii}{\mathrm{i}}
\newcommand{\dd}{\,\mathrm{d}}
\newcommand{\pd}{\partial}

\newcommand{\mr}{\mathrm}

\renewcommand{\Re}{\mathtt{Re}}
\renewcommand{\Im}{\mathtt{Im}}

\newcommand{\sgn}{\mathrm{sgn}}

\newcommand{\email}[1]{\protect\href{mailto:#1}{#1}}

\begin{document}
\title{Stability of traveling waves in a nonlinear hyperbolic system approximating a dimer array of oscillators}
 
\author{Huaiyu Li\thanks{Department of Applied Physics and Applied Mathematics, Columbia University, New York, NY (\email{hl3002@columbia.edu}).},\ \ Andrew Hofstrand\thanks{New York Institute of Technology, New York, NY (\email{ahofstra@nyit.edu}).},\ \ Michael I. Weinstein\thanks{Department of Applied Physics and Applied Mathematics and Department of Mathematics, Columbia University, New York, NY (\email{miw2103@columbia.edu}).}}
\date{\today}

\maketitle
\begin{abstract}
We study a semilinear hyperbolic system of PDEs which arises as a continuum approximation of 
the discrete nonlinear dimer array model  introduced by Hadad, Vitelli and Alu (HVA) in \cite{HVA17}.
We classify the system's traveling waves, and study their stability properties. We focus 
on traveling pulse solutions (``solitons'') on a nontrivial background
and moving domain wall solutions (kinks); both arise as heteroclinic connections between spatially uniform equilibrium of a reduced dynamical system.  
We present analytical results on: 
nonlinear stability and spectral stability of supersonic pulses, and spectral stability of moving domain walls.
Our stability results are in terms of weighted $H^1$ norms of the perturbation, which capture the phenomenon of {\it convective stabilization}; 
as time advances, the  traveling wave ``outruns'' the \underline{growing} disturbance excited by an initial perturbation; 
the non-trivial spatially uniform equilibria are linearly exponentially unstable. We use our analytical results to interpret phenomena observed in numerical simulations.
\end{abstract}
\tableofcontents
\label{toc}
\section{Introduction}

\subsection{Background and motivation}
We study the system of semi-linear hyperbolic PDEs:
\begin{equation}
\label{eq: PDE in lab frame}
\begin{aligned}
u_t & = v_y + \mathcal N\big(u^2 + v^2\big) v \\
v_t & = u_y - \mathcal N\big(u^2 + v^2\big) u\ ,
\end{aligned}
\end{equation}
governing the time evolution of 
$ (y,t)\in\mathbb R\times\mathbb R\mapsto b(y,t) = \begin{bmatrix}
    u(y,t) & v(y,t) 
\end{bmatrix}^\mathsf T \in\mathbbm R^2$. 
The properties of the nonlinearity, $\mathcal{N}(\cdot)$ are discussed below in {\sc section} \ref{sec: nonlinearity}.
Our study is inspired by work of Hadad, Vitelli and Alu (referred to as HVA in this article) \cite{HVA17},
who introduced a nonlinear variant of the {\it discrete and linear} Su-Schrieffer-Heeger (SSH) dimer model \cite{ssh79},
which can be experimentally realized via an array of coupled \textit{nonlinear} electrical circuit elements; see \eqref{eq: discrete} below.
The SSH model is well-known to exhibit topological transitions, related to the closing 
(and formation of a linear crossing at  ``Dirac points'' in the band structure) and re-opening of a spectral gap in its band structure as the ratio 
of the intra-cell to inter-cell coupling (hopping) coefficients is varied.
HVA studied a continuum model, appropriate for wave-packet excitations centered on the Dirac point quasi-momentum. 
They derived, via phase portrait analysis and numerics, traveling pulse solutions (solitons) and traveling domain wall solutions (kinks/antikinks). They then studied, by numerical simulations,
the spatially \textit{discrete} nonlinear time-evolution for initial data given by such solitons and kinks, 
sampled on the lattice. 
These numerical simulations of the discrete model showed that the  
{\it core} of both kinks and supersonic pulses appears
to be stable against small spatially localized perturbations. 
Extensive simulations of the time-dependent nonlinear continuum model, 
\eqref{eq: PDE in lab frame} (see \cite{li2023thesis}\cite{du2023discontinuous}) 
demonstrate that this traveling core is {\it convectively stable}; the core persists although away from the core the solution tends to grow with advancing time.

{\it In this article, 
we present analytical results for the system \eqref{eq: PDE in lab frame} on nonlinear stability and spectral stability
of supersonic pulses (solitons) that asymptote to different nontrivial equilibria, and spectral stability of moving domain walls (kinks),
which contribute to an understanding of the dynamics.}

We next present a precise  mathematical formulation,  
discuss numerical results which exhibit the phenomenon of convective stabilization of pulses and solitons, 
and summarize our analytical results.

\subsection{Assumptions on the nonlinearity 
}
\label{sec: nonlinearity}
Throughout this article, the nonlinearity $s\mapsto \mathcal{N}(s)$ in \eqref{eq: PDE in lab frame}, for $0 \leq s = u^2 + v^2 < \infty$,  is assumed to be smooth and to satisfy:
\begin{itemize} 
\item[($\mathcal{N}1$)]  $\mathcal{N}'(s)<0$ for $s>0$, and   $\mathcal N(0) = 1$, $\mathcal N(1)=0$. 
\footnote{By $\mathcal N'$, $\mathcal N''$ etc., we always mean $\mathcal N'(s) := \dd \mathcal N(s) / \dd s$ etc.}
The parameter 
\begin{equation}
\label{eq: K}
K:= \big|\mathcal N'(1)\big| = - \mathcal N'(1) > 0.
\end{equation}
will play an important role.
\item[($\mathcal{N}2$)] $\lim_{s \to +\infty} \mathcal N(s) = \mathcal{N}_\infty\in [-\infty,0)$. 
\end{itemize}

A common physical assumption is that the nonlinearity be saturable.
We say that the nonlinearity $\mathcal N(s)$ 
is \textbf{saturable} 
if $(\mathcal{N}2)$ is replaced by
\medskip

\noindent $(\mathcal{N}2')$ $\lim_{s \to +\infty} \mathcal N(s) = \mathcal N_\infty \in(-\infty , 0) $ and further that
 $\mathcal N(s) -\mathcal N_\infty$,
and its derivatives $\mathcal N^{(k)} (s) \to 0$, $k = 1, 2, 
 \cdots$, decay to zero
 sufficiently rapidly as $s \to \infty$. 
\medskip

 \noindent An example of a saturable nonlinearity is $\mathcal N(s) = \frac{1-s}{1+s}$. 
An example of a general nonlinearity is $\mathcal N(s) = 1-s$.

The system \eqref{eq: PDE in lab frame} is a semilinear hyperbolic system, 
whose characteristic lines are given by solutions of 
$dy/dt = \pm 1$.
Any $\mathcal C^1$ solution satisfies the conservation law
\begin{equation}
\label{eq: conserved}
    \partial_t \big(u^2+v^2\big) + \partial_y\big(-2uv\big)=0. 
\end{equation}
The conservation law \eqref{eq: conserved} plays a role in our classification
of traveling wave solutions in {\sc section} \ref{sec: tws},
and in the proof and application of finite propagation speed 
in {\sc sections} \ref{sec:finite-prop} and {\sc section} \ref{sec:nonlin-convec}.
The system \eqref{eq: PDE in lab frame} does not appear to be of Hamiltonian type. It does have certain discrete symmetries which we summarize in the following:
\begin{proposition}
[Discrete Symmetries]
\label{prop: discrete symmetries} 
Let $ b(y,t) = [u(y,t),v(y,t)]$ denote  a solution of (\ref{eq: PDE in lab frame}). Then, 
\begin{align*}
    \big[ \mathcal P  b \big](y,t) & := [-u(y,t),-v(y,t)]\\
 \big[\mathcal T  b\big](y,t) &:= [u(y,-t),-v(y,-t)] \\
 \big[\mathcal {C}  b\big](x,t) &:= [v(-y,-t),u(-y,-t)]
\end{align*}
are also solutions. Moreover,
\begin{align*}
    \mathcal P^2 & = \mathcal T^2 = \mathcal C^2 = \mathrm{id} \\
    \mathcal {PC} &= \mathcal {CP},\ \mathcal {PT} = \mathcal {TP},\ \mathcal {TC} =\mathcal {CPT}
\end{align*}
\end{proposition}
The proof of {\sc proposition} \ref{prop: discrete symmetries} is very straightforward and we omit it.

\subsection{Phenomena motivating this work}
 System \eqref{eq: PDE in lab frame}, with nonlinearity assumptions $(\mathcal N 1)$ and $(\mathcal N 2)$, has spatially uniform equilibria:
\begin{equation}\textrm{either $b = [0\ 0]^\top$ or  $b=[u\ v]^\top$, where $|b| = \sqrt{u^2+v^2}=1$; see $(\mathcal N 1)$.
}\label{eq:spat-unif}
\end{equation}
In {\sc section} \ref{sec: tws} we classify the \textbf{traveling wave solutions} (TWS) of \eqref{eq: PDE in lab frame}, which are of the form
\begin{equation}
    b(y,t) = b_*(y,t) =  b_*(x := y-ct), \label{eq:tw-ansatz} 
\end{equation}
and tend to spatially uniform equilibria at infinity:
\[ \lim_{x \to \pm \infty } b_*(x) = b_{*,\pm}.\]
Here, $b_{*,+}$ and $b_{*,-}$ are among the spatially uniform states displayed in \eqref{eq:spat-unif}.
Traveling wave solution profiles, $b_*(x)$, are heteroclinic orbits in the phase portrait of a two-dimensional dynamical system obtained from \eqref{eq: PDE in lab frame} via the ansatz \eqref{eq:tw-ansatz}; see Section \ref{sec: tws}.

Pulses are orbits connecting distinct nontrivial equilibria satisfying $|b|=1$, 
and kinks (and antikinks) are those which connect the trivial equilibrium with a non-trivial equilibrium. 
Pulses may be supersonic ($|c|>1$) or subsonic ($|c|<1$), while  kinks and anti-kinks are all subsonic.

\subsubsection{Convective stability and weighted spaces}

Consider the case of a supersonic pulse with $c>1$. {\sc figure} \ref{fig: perturbed supersonic pulse simulation} displays snapshots of the time-evolution of a localized perturbation of $b_*$.
This initial perturbation  generates time-evolving perturbations of the pulse which, in a frame of reference moving with the traveling wave speed, $c>1$, appear
to travel leftward, away
from the traveling pulse core,
while also growing in amplitude.
In this same moving frame of reference,  
the deviation from an exact traveling wave profile, when measured within a fixed 
semi-infinite ``window" $[R,\infty)$, 
tends zero as $t$ increases because 
perturbations exit the window
at $x=R$ as $t$ increases.
We say that the supersonic pulse (or its core) is 
{\bf convectively stable}. 
The notion of convective stability has been  considered previously in, for example, \cite{pego1994asymptotic,martel2001asymptotic,pego1997convective}.

We capture this stabilization of the traveling wave core by working in {\bf weighted} function spaces. In a reference frame \textit{moving} with the traveling wave solution $b_*(x)$,
perturbations are studied as elements of $H^1\big(\mathbbm R;W(x)\dd x\big)$,
where the weight $W(x)$ is chosen to be of  \textbf{exponential type}. Specifically, $W(x)$ is monotone and of the form
\[
    W (x) = e^{w(x)},\quad w(x) \in \mathbbm R; 
\] 
see {\sc section} \ref{sec: weighted spaces}.
If, in the \textit{moving frame} (speed $c>1$), 
the perturbation travels in the direction of decrease of $w(x)$ 
(to the left),
then it is registered  as decaying with advancing time if the time rate of  amplitude growth is not too large. This intuition underlies our nonlinear and spectral stability results for supersonic pulses;
see {\sc section} \ref{sec:summary}.
\begin{remark}[Do solutions grow without bound in $L^\infty(\mathbb R)$?]\label{rem:pertgrowth}
 We note that numerical simulations of the nonlinear PDEs suggest that the perturbation of traveling wave may be growing without bound. (Note that we have no {\it a priori} $H^1$ bound on the solution; see {\sc Theorem} \ref{thm: global well-posedness perturbation H 1}.) However, as time advances, the perturbation grows in regions which are further and further away from the traveling wave core. This behavior of the perturbation is registered as time-decay with respect to the weighted norm. This 
 convective stability scenario differs from the more typical scenarios in solitary wave stability. For example, in KdV type equations, which are Hamiltonian
  and come with an {\it a priori} bound on the $H^1$ norm,  the perturbation remains bounded and is in fact comprised of small amplitude solitary waves and a radiation component which together decay in appropriate weighted or local energy norms;
   see, for example, \cite{pego1994asymptotic,martel2001asymptotic}.
\end{remark}

\begin{figure}[h!]
     \centering
     \begin{subfigure}[b]{0.85\textwidth}
     \centering
     \includegraphics[width=\textwidth]{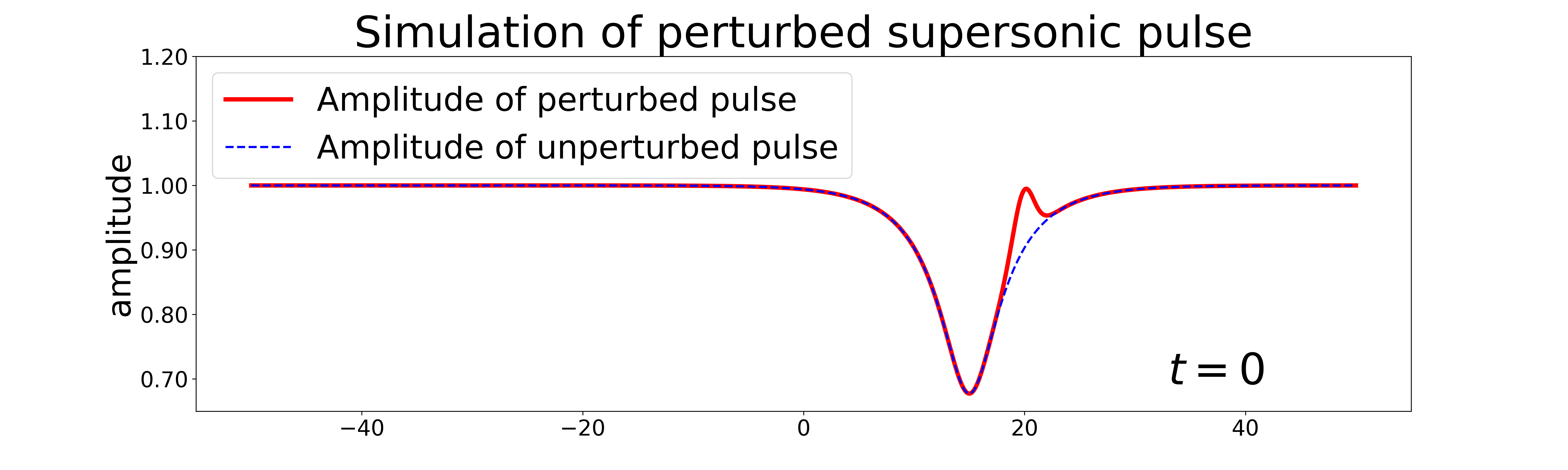}
         \label{fig: perturbed supersonic pulse initial}
     \end{subfigure}
     \begin{subfigure}[b]{0.85\textwidth}\centering\includegraphics[width=\textwidth]{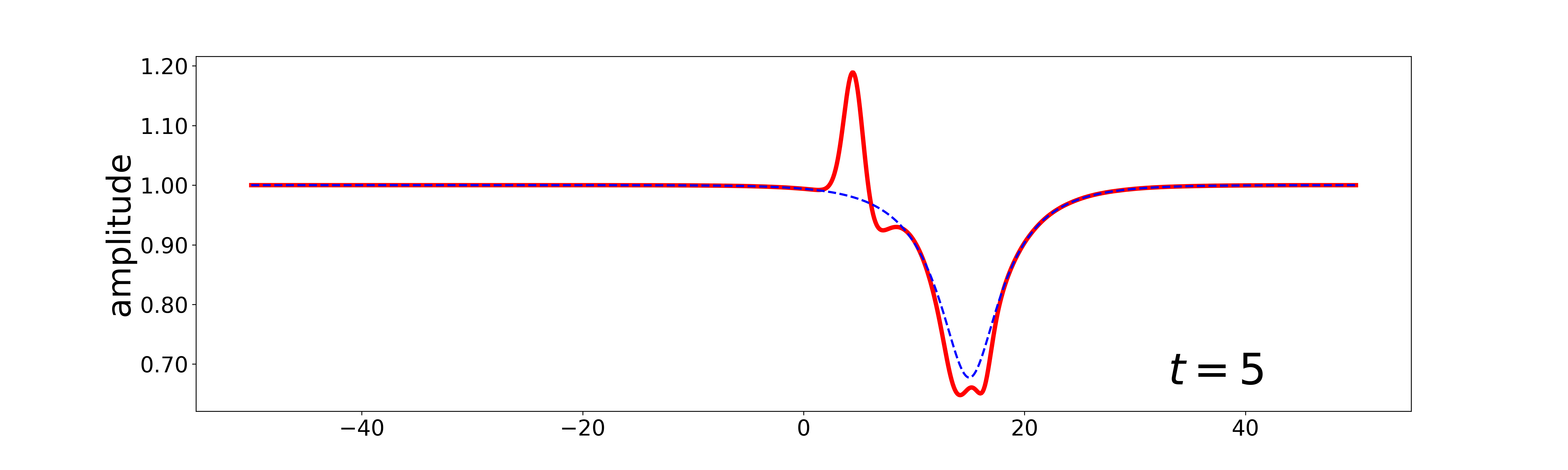}
         \label{fig: perturbed supersonic pulse 1}
     \end{subfigure}
          \begin{subfigure}[b]{0.85\textwidth}
         \centering
         \includegraphics[width=\textwidth]{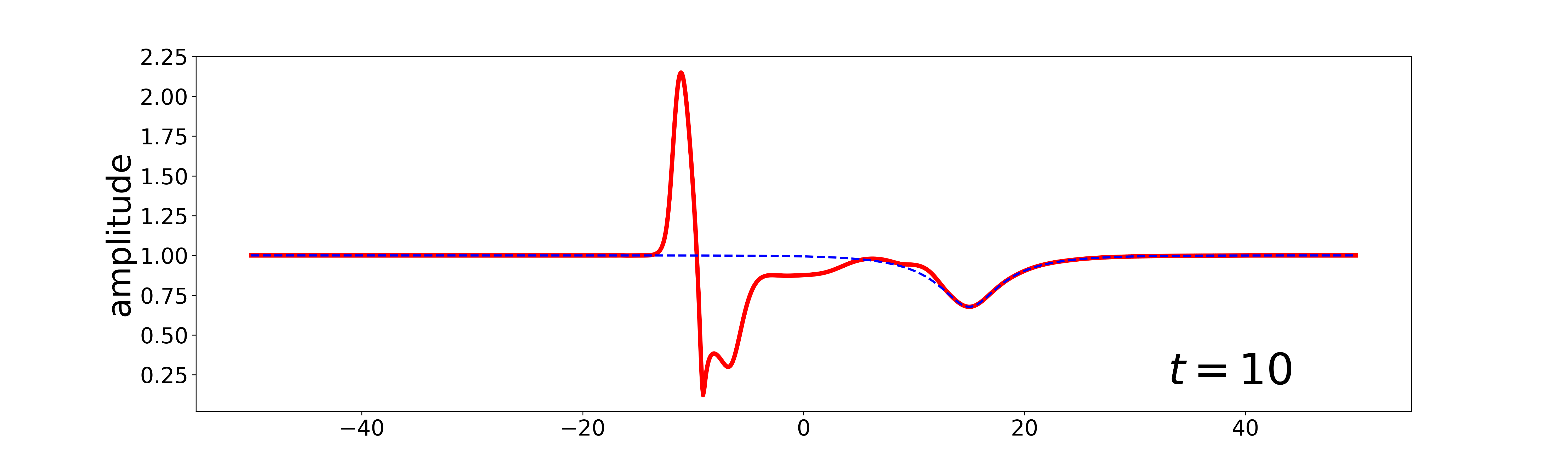}
         \label{fig: perturbed supersonic pulse 2}
     \end{subfigure}
     \begin{subfigure}[b]{0.85\textwidth}
         \centering
         \includegraphics[width=\textwidth]{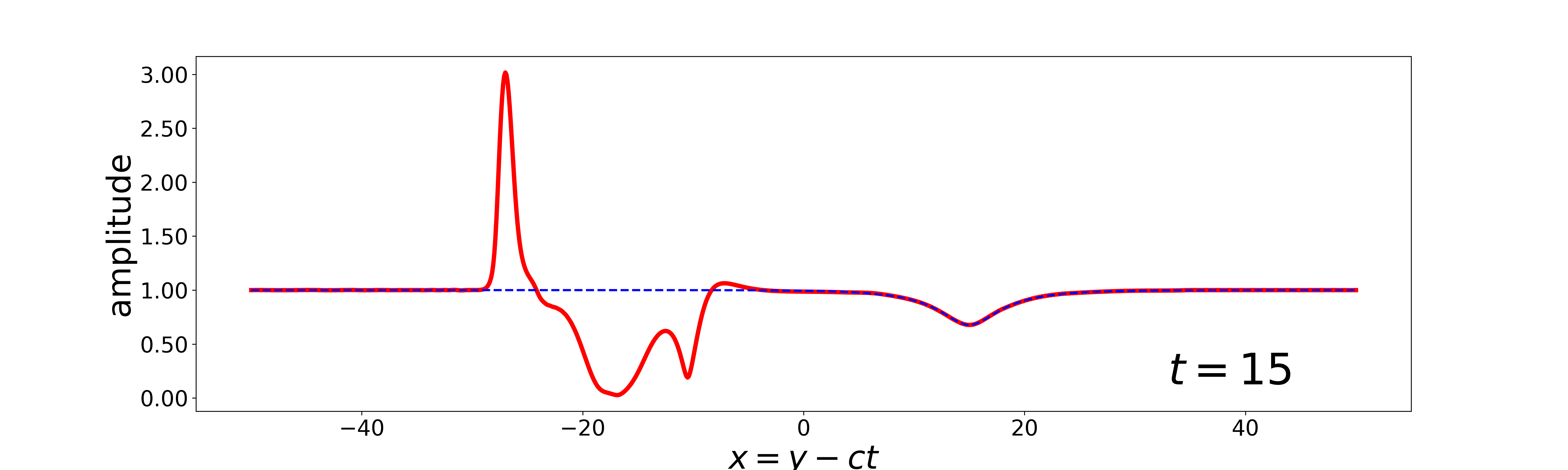}
         \label{fig: perturbed supersonic pulse 3}
     \end{subfigure}
\caption{Convective stability of supersonic pulses.
Snapshots of a perturbed supersonic ($c=2>c_0=1$) pulse of (\ref{eq: PDE in lab frame}) in a reference frame of speed $c$. 
Perturbation at $t=0$ is concentrated to the right of the core. 
Red curves indicate the amplitude of the solution at different times.
Blue dashed curves indicate the amplitude profile of the unperturbed supersonic pulse. 
The perturbation grows relative to the  unperturbed pulse as time advances; note the differing amplitude scales of the different panels. 
 The pulse core is nearly restored without any phase shift at $t = 15$.
} 
\label{fig: perturbed supersonic pulse simulation}
\end{figure}

\begin{figure}[h!]
     \centering
\includegraphics[scale=0.75]{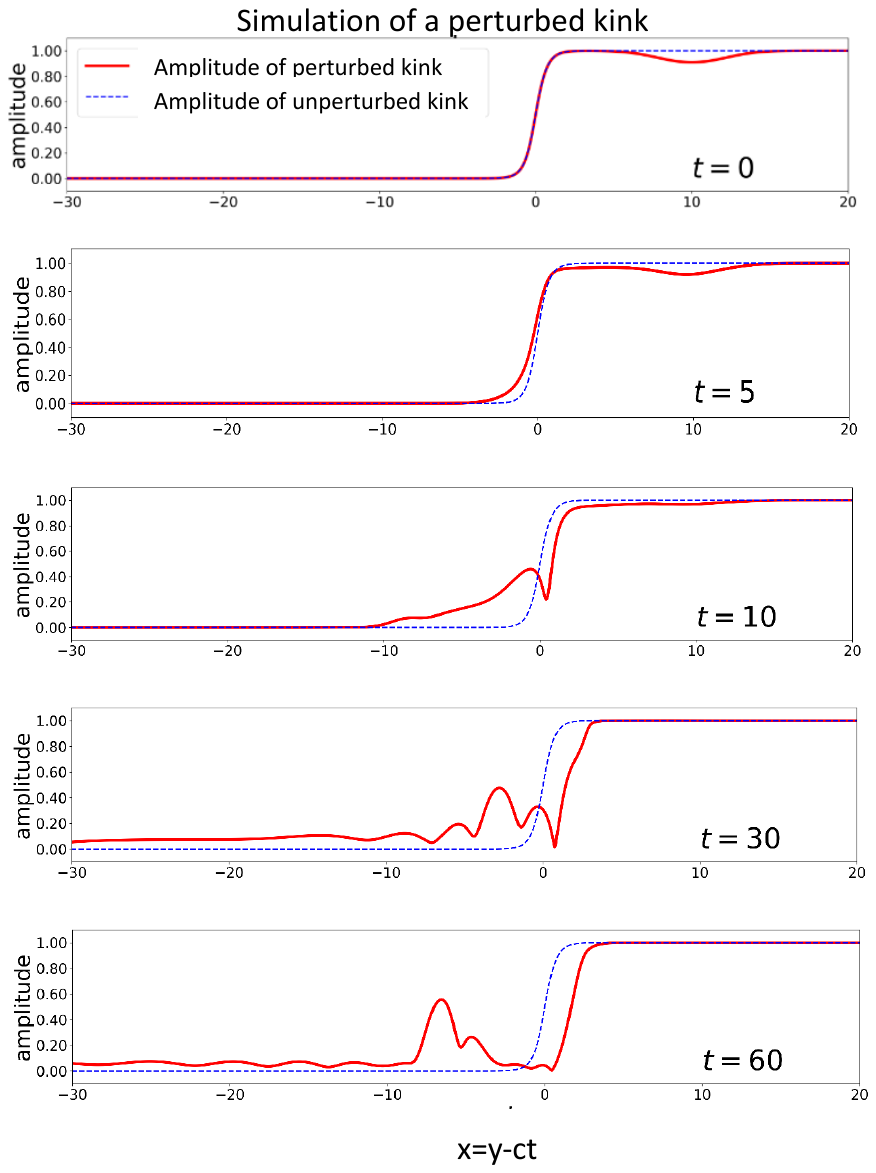}
\caption{Numerical evidence for the convective stability of kinks.
Snapshots of perturbed kink ($c=0.9$) of (\ref{eq: PDE in lab frame})) 
in a reference frame of the same speed $c$. 
Perturbation at $t=0$ is concentrated to the right of its core. 
Red curves indicate the solution amplitude at different times.
Blue dashed curves indicate the amplitude profile of the unperturbed kink.
The perturbation departs from the kink core to its left. 
The kink core is nearly completely  restored \textit{modulo a phase shift} at $t = 60$; see Section \ref{sec:nlstab-kink}.
} 
\label{fig: perturbed subsonic kink simulation}
\end{figure}

{\sc Figure} \ref{fig: perturbed subsonic kink simulation} displays numerical simulations illustrating the convective stability of kinks, which are all subsonic ($|c|<1$).
Here too, we observe in a frame of reference traveling at the same speed as does the unperturbed kink. 
The kink outruns the generated perturbations as time advances, consistent with our results on linear stability of kinks; see {\sc section} \ref{sec:summary} and the discussion of {\sc section} \ref{sec:nlstab-kink}.

\subsection{Summary of results}\label{sec:summary}
We summarize the results presented in this article.
\begin{itemize}
\item {\sc Section} \ref{sec: tws} and {\sc appendix} \ref{app: classification} contain the classification of all bounded traveling wave solutions; pulses (dark and bright), kinks and anti-kinks.
    \item {\sc Theorem} \ref{thm: convective} (Nonlinear convective stability of supersonic pulses): Consider the case of saturable nonlinearities (assumptions ($\mathcal N1$) and ($\mathcal N2'$) in
    {\sc section} \ref{sec: nonlinearity}). 
    Supersonic pulses ($|c|>1$) are nonlinearly convectively stable.    
    \item {\sc Theorem} \ref{thm: spectral stability of supersonic pulses} (Linear convective stability of supersonic pulses): For general nonlinearities (not necessarily saturable), 
    supersonic pulses are {\it spectrally stable} in suitable weighted  
 $L^2$ spaces, denoted $L^2_{w}$; see {\sc section} \ref{sec: weighted spaces}.
    \item Theorems on spectral stability of kinks ($|c|<1$): \\
    (i) {\sc Theorem} \ref{thm: kink spectral stability c = 0}:  Static / non-moving kinks ($c=0$) are spectrally stable. 
  \\  (ii) {\sc Theorem} \ref{thm: kink stability}:  Under the following additional hypothesis ($\mathcal{N}3$):
  \begin{enumerate}
\item[($\mathcal{N}3$)]   $\mathcal N'(s)\le 0$ for $s>0$ (monotonicity) and $\mathcal N''(s)\le0$ for  $s \in [0,1]$ (concavity),\\ \\
\noindent  moving kinks ($0 \leq |c|<1$) are spectrally stable.
    \end{enumerate}
\end{itemize}

\subsection{Future directions, open questions}
\label{sec:openq}
We list some possible directions for future investigation and corresponding open questions.

\subsubsection{Large time selection of supersonic pulses problems}\label{sec:selection}
Consider a supersonic pulse of speed $c>1$.
Theorem \ref{thm: spectral stability of supersonic pulses} on spectral stability
and Theorem \ref{thm: convective} on nonlinear asymptotic stability are convective stability results, which measure the initial perturbation in spatial norms with an exponential weight. The exponential rate satisfies constraints which depend on the underlying traveling wave speed, $c$,  and properties of the nonlinearity. In particular, the weighted norm imposes a minimal decay rate of the perturbation in the direction of pulse propagation.
{\it What does a supersonic pulse evolve into, 
under perturbations which violate the decay rate constraints in  {\sc Theorems} \ref{thm: convective} and \ref{thm: spectral stability of supersonic pulses} ?}

Note that any supersonic pulse (of speed $c_0>1$) whose profile connects two equilibria on the unit circle,  is embedded in a  continuous family of supersonic pulses with speeds encompassing the range $1<c<\infty$; see {\sc Figure} \ref{fig: fixed asymptotics}. Further, the profile of a supersonic pulse of speed $c$ approaches its asymptotic values at an exponential rate, 
given in \eqref{eq: supersonic pulse asymptotic rate}, which becomes smaller as $|c|$ grows; see Section \ref{sec:conv-rate}.

{\it Does a supersonic pulse traveling with speed $c_0>1$, when perturbed by a slow-decaying perturbation -- outside the validity of 
 {\sc Theorem} \ref{thm: convective} --
evolve into a supersonic pulse of some speed $c' > c_0$, with a compatible spatial decay? If so, what determines the asymptotically selected profile?}

\subsubsection{Nonlinear stability of kinks}\label{sec:nlstab-kink}
Our spectral stability analysis and numerical simulations (see {\sc figure} \ref{fig: perturbed subsonic kink simulation}) 
suggest that the {\it family of spatial translates of a kink} is nonlinearly convectively stable.
We conjecture the following: Let $b_*(x)$ denote a kink and  $B_0(x)$ a 
sufficiently rapidly decaying initial perturbation. 
Then, there exists $x_0\in\mathbbm R$, depending on $B_0$ (and $b_*(x)$), 
such that the solution $b(x,t)$ to \eqref{eq: PDE in moving frame} 
with initial data $b_*(x)+B_0(x)$, satisfies 
\begin{equation}
\label{eq:phase-shift}
    \lim_{t \to \infty} \Big\| b(x,t) - b_*(x-x_0)\Big\|_{H^1\left([R,\infty),dx\right)} = 0
\end{equation}
The phase shift in \eqref{eq:phase-shift} is related to the zero energy translation mode of the linearized operator,
see {\sc Theorem} \ref{thm: kink stability} and {\sc remark} \ref{remark: translation mode kink}.
In contrast, Theorem 
\ref{thm: convective} on nonlinear and convective (asymptotic) stability of supersonic pulses, requires no asymptotic phase adjustment.
This is corroborated by numerical studies showing no phase shift in the emerging stable supersonic pulse,
see {\sc figure} \ref{fig: perturbed supersonic pulse simulation}.
Note: although there is a state which is formally in the kernel of the linearized operator (due to translation invariance of  \eqref{eq: PDE in lab frame}),  
 this state is not in the weighted $L^2$ space with respect to which the supersonic pulse is spectrally stable; 
 see {\sc Theorem} \ref{thm: spectral stability of supersonic pulses} and the discussion following {\sc proposition} \ref{prop:no_disc-super}.
 
Another question is to clarify the scenario described in Remark \ref{rem:pertgrowth}, which is based on numerical simulations. 
And an example of further technical questions concerning kinks is whether,
for example, spectral stability can be established if the concavity assumption ($\mathcal N3$) (used in \eqref{eq: concavity}), 
on the nonlinearity, is relaxed.

\subsubsection{Alternative measures of the perturbation's spatial localization and size}
Our stability results for pulses (nonlinear and spectral stability) and kinks (spectral stability) are formulated in function spaces, requiring exponential decay of the perturbation in the direction of propagation of the traveling wave. It would be of interest to extend these stability results to spaces with weaker spatial localization requirements; 
for example,  algebraically weighted $L^2$ spaces \cite{miller1997spectral} or $H^1_{\rm loc}$  \cite{martel2001asymptotic}. 

\subsection{Linear asymptotic stability} For the case of supersonic pulses, we believe that our  results on linear spectral stability can be used to obtain exponential time-decay bounds for the 
linear semi-group $e^{L_{w,*}t}$, along the lines of the analysis of \cite{pego1994asymptotic,pego1997convective}. 
For the case of kinks, where the spectrum of $L_{w,*}$ is spectrally stable, but with part of its spectrum on the  imaginary axis, we expect the governing time decay to be dispersive type, after projecting out the zero energy mode.

\subsubsection{Periodic solutions} As discussed in detail in \cite{li2023thesis}\, 
\eqref{eq: PDE in lab frame}  has rich families of periodic solutions traveling wave solutions.
Their stability properties is an open question. 
\subsubsection{Relation between discrete and continuum models}
Finally, system \eqref{eq: PDE in lab frame} is introduced in \cite{HVA17} as a formal continuum approximation for  
a nonlinear discrete array of coupled nonlinear circuits, 
valid for excitations whose spatial scale is slow on the inter-dimer length scale.
After scaling and nondimensionalization, the discrete  system takes the form
\begin{equation}
    \label{eq: discrete}
    \begin{aligned}
        \dot u_n & = v_n - v_{n-1} + \mathcal N \big(u_n^2 + v_n^2\big) v_n \\
        \dot v_n & = u_{n+1} - u_n - \mathcal N \big(u_n^2 + v_n^2\big) u_n
    \end{aligned}
\end{equation}
As demonstrated in\cite{HVA17} there is evidence of the  pulse-like and kink-like behaviors in the discrete system \eqref{eq: discrete}.
It is of interest to understand the relation between our continuum  analytical and numerical results for \eqref{eq: PDE in lab frame} and those observed,  thus far only numerically,
in \eqref{eq: discrete}. 

\subsection{Notation and conventions}
\begin{enumerate}
\item $H^s=H^s(\mathbb R)$ denotes the Sobelev space with norm given by:
\[
    \big\| f\big\|^2_{H^s} := \int_\mathbbm{R} \big(1 + |k|^2 \big)^s \big| \hat { f} (k) \big|^2 \dd k < \infty
\]
where $\hat{f}$ denotes   the Fourier transform.
    \item {\it Weighted spaces:} We define the \textbf{weighted $L^2$ spaces}, with weight $W (x) = e^{w(x)}$ where $w(x)$ is a real-valued function on $\mathbbm R$ as
    \begin{equation}
    \label{eq: weighted spaces}
    L^2_w := L^2\big(\mathbbm R, e^{w(x)} \dd x\big),\quad 
    H^1_w:= \Big\{  f(x) \in L^1_\mr{loc}:\ e^{w(x)} f(x) \in H^1 
    \Big\} 
\end{equation}
for details of the particular weighted spaces used in this work, see {\sc Section} \ref{sec: weighted spaces}.
\item {\it Coordinates:} 
The linear stability analysis of this work is always conducted in frames of reference that travel at the speed of an underlying traveling wave. 
We denote with $y$ the spatial coordinate in the non-moving (lab) frame of reference, cf. \eqref{eq: PDE in lab frame}, and with $x = y - ct$ the spatial coordinate in the frame of reference traveling with some speed $c$; see \eqref{eq: PDE in moving frame}.
\item {\it Default branch of the square root function:} 
We define function $z \mapsto \sqrt{z}$ in such a way that its values have non-negative real parts. 
In particular, $\sqrt{1} = 1$ and $\sqrt{z}$ is conformal from the cut complex plane, $\mathbbm C \setminus (-\infty , 0]$,
to the open right-half plane $\big\{ \Re z > 0 \big\}$.
For $z \leq 0$, $\sqrt{z}$ is continued from above the cut 
and its values always have non-negative imaginary parts, e.g., $\sqrt{-1} = \ii$.
\item {\it Pauli matrices:}
We use the standard convention of defining Pauli matrices, $\sigma_0, \sigma_1, \sigma_2$ and $\sigma_3$, as a set of basis in the linear space of 2-by-2 complex matrices:
\begin{equation}
\label{eq: pauli}
    \begin{aligned}
        & \sigma_0 = \begin{bmatrix}
            1 & 0 \\ 0 & 1
        \end{bmatrix},\quad 
        & \sigma_1 = \begin{bmatrix}
            0 & 1 \\ 1 & 0
        \end{bmatrix} \\
        & \sigma_2 = \begin{bmatrix}
            0 & \ii \\ -\ii & 0
        \end{bmatrix},\quad 
        & \sigma_3 = \begin{bmatrix}
            1 & 0 \\ 0 & -1
        \end{bmatrix}
    \end{aligned}
\end{equation}
Here $\sigma_i\sigma_j=-\sigma_j\sigma_i$ for $i, j \in \{ 1,2,3 \}$ and $i\ne j$ and $\sigma_i^2=\sigma_0$.
\end{enumerate}

\subsection{Acknowledgements} 
The authors wish to thank A. Al\`{u}, Y. Hadad, Q. Du and L. Zhang for many stimulating discussions.
This research was supported in part by NSF grant DMS-1908657 (MIW, HL),  DMS-1937254 (MIW) and Simons Foundation Math + X Investigator Award \# 376319 (MIW, HL). 
AH was supported in part by the Simons Collaboration on Extreme Wave Phenomena Based on Symmetries and AFSOR Grant No. FA9950-23-1-0144.
Part of this research was completed during the 2023-24 academic year, when M.I. Weinstein was a Visiting Member in the School of Mathematics - Institute of Advanced Study, Princeton, supported by the Charles Simonyi Endowment, and a Visiting Fellow in the Department of Mathematics at Princeton University.

\section{Traveling wave solutions}
\label{sec: tws}
We express \eqref{eq: PDE in lab frame} with respect to a coordinate system traveling with speed $c$, where $|c|\ne1$. Setting  $x := y - ct$ we obtain
\begin{equation}
\label{eq: PDE in moving frame}
\begin{aligned}
    u_t & = c u_x + v_x + \mathcal N (u^2+v^2) v \\
    v_t & = u_x + c v_x - \mathcal N (u^2 +v^2 ) u
\end{aligned}
\end{equation}
For $c\ne1$, the system \eqref{eq: PDE in moving frame} has the {\bf equilibria}:
\begin{equation} b_{*0}=\begin{bmatrix}
    0 & 0
\end{bmatrix}\quad {\rm and}\quad b_{*}(\theta)=\begin{bmatrix}
    \cos\theta & \sin\theta
\end{bmatrix}, \quad \theta \in ( -\pi , \pi].
\label{eq:equil}\end{equation}
The profile of a \textbf{traveling wave solution} (TWS) profile, $b(x) = \begin{bmatrix} u(x) & v(x)\end{bmatrix}$,
of speed $c \neq \pm 1$ is an orbit of the dynamical system 
\begin{equation}
\label{eq: TWS profile}
    \begin{aligned}
    u' &= \frac{\mathcal N (u^2+v^2)  }{1 - c^2} \big(u + c v\big)
    \\
    v' &= \frac{\mathcal N(u^2+v^2)}{1 - c^2} \big( -c u - v\big)
    \end{aligned}
\end{equation} 
If  $|c| > 1$, we say the TWS is  {\it supersonic} and if  $|c|<1$ say that it is {\it subsonic}. 

Evaluating the conservation law 
\eqref{eq: conserved} on a TWS $(u,v)(y-ct)$, we conclude that along its phase plane trajectory, the   ``energy'' $E_c[u,v]$ of $x\mapsto (u(x),v(x))$:
\begin{equation}
\label{eq: E c}
    E_c[b] = E_c[u,v] \equiv c (u^2+v^2) + 2 u v= \frac{-1 + c}{2} (u-v)^2 
    + \frac{1+c}{2} (u+v)^2
\end{equation}
is independent of $x$. 
Thus, traveling wave profiles correspond to connected subsets of level sets in $\mathbb R^2$
 of $(u,v)\mapsto E_c[u,v]$:
\[ E_c^{-1}(E)=\Big\{ (u,v) \in \mathbbm R^2:\  E_c(u,v)=E\Big\}.\]


\subsection{Level sets  of $E_c(u,v)$}
\label{sec:Lsets}
\begin{itemize}
\item For $c>1$, $E_c(u,v)$ is positive definite. 
Hence, the level sets $E_c = E$ are ellipses
parametrized by $E>0$:
\begin{equation} (c+1) \left(\frac{u+v}{\sqrt2}\right)^2+(c-1)
\left(\frac{u-v}{\sqrt2}\right)^2 = E.\label{eq:ellipse+1}
\end{equation}
This family of ellipses degenerates to the origin, $(0,0)$ as $E\downarrow0$.
\\
\item For $c<-1$, $E_c(u,v)$ is negative definite. 
Hence, the level sets are ellipses
parametrized by $E<0$:
\begin{equation} \big(|c|-1 \big) \left(\frac{u+v}{\sqrt2}\right)^2+ \big(|c|+1 \big)
\left(\frac{u-v}{\sqrt2}\right)^2 = -E=|E|.\label{eq:ellipse-1}
\end{equation}
This family of ellipses also degenerates to the origin, $\begin{bmatrix} 0& 0\end{bmatrix}$ as $E\uparrow0$.
\\
\item For $|c|<1$, $E_c(u,v)$ is indefinite.
The level sets are hyperbolas with two branches,
with one orientation for $E>0$ and another orientation for $E<0$. 
As $E\downarrow0$, and as $E\uparrow0$, these level sets degenerate to a pair of lines which intersect at the origin.
\end{itemize}

\subsection{Bounded heteroclinic traveling wave solutions}
\label{sec:bddTWS}
We denote a bounded traveling wave solution (TWS) profile with speed $c$ for the parameter $E$ by $b_{c,E}(x)$. 
A bounded TWS, $b_{c,E}(x)$, corresponds to a bounded  heteroclinic orbit  of \eqref{eq: TWS profile} which connect,  as $x$ varies from $-\infty$ to $+\infty$, 
 distinct equilibria which lie in the set:
 \begin{equation}
 \big\{\begin{bmatrix}
     u & v
 \end{bmatrix}:u^2+v^2=1\big\}\cup \big\{\begin{bmatrix} 0 & 0\end{bmatrix}\big\}.\label{eq:equilib}\end{equation}
{\it Hence, heteroclinic orbits are determined by the bounded connected subsets of level sets of $E_c(u,v)$, whose boundary points lie in the set of equilibria \eqref{eq:equilib}.}
See Figures \ref{fig: supersonic}-\ref{fig: kinks} 
\footnote{The reduced dynamical system \eqref{eq: TWS profile}  also has periodic orbits for $|c|>1$; we do not study these solutions in the present work.}.

The following proposition, displays relations among traveling wave orbits, which are implied by the  symmetries of \eqref{eq: PDE in lab frame}. 

\begin{proposition}[Discrete symmetries of the family of traveling wave solutions]
\label{prop: discrete symmetries TWS}
Let $ b(x)$ be the \textbf{profile} of a TWS with speed $c$. 
The corresponding solution $ b(y-ct)$ of \eqref{eq: PDE in lab frame} can be transformed into other TWSs under \textbf{discrete} transformations 
$\mathcal P$, $\mathcal T$ and $\mathcal C$ given in {\sc proposition} \ref{prop: discrete symmetries}. 
In particular, the profiles of and corresponding conserved quantity $E_c$ of the transformed traveling wave solution is listed below.
\begin{enumerate}[label=(\roman*)]
    \item 
    \label{prop: TWS P symmetry parity}
    $\mathcal{P}  b(x)$ is a TWS with speed $c$ whose profile is $- b(x)$ and $E_c[\mathcal{P}  b] = E_c[ b(x)]$.
    \item 
    \label{prop: TWS T symmetry time reversal}
    $\mathcal{T}  b(x)$ is a TWS with speed $-c$ whose profile is $ \sigma_3 b(x)$ and $E_{-c}[\mathcal{T}  b] = -E_c[ b(x)]$.
    \item 
    \label{prop: TWS C symmetry conjugacy}
    $\mathcal{C} b(x)$ is a TWS with speed $c$ whose profile is $\sigma_1  b(-x)$ and $E_c[\mathcal{C}  b] = E_c[ b(x)]$.
\end{enumerate}
\end{proposition}
\begin{remark}\label{rem:NB} {\bf N.B.}
In view of Proposition \ref{prop: discrete symmetries TWS} we shall focus, particularly in our stability analyses, on the case of right-moving pulses and kinks, $c\ge0$. 
The linearized spectra of TWSs indeed respect these symmetries; see {\sc Theorem} \ref{thm: discrete symmetry of spectra}.
A complete classification of all traveling wave solutions and their relation through discrete transformations is given in {\sc appendix} \ref{app: classification}.
\end{remark}

\subsection{Traveling Pulses and Domain Walls (Kinks and Anti-kinks)}
\label{sec:PKAK}

The set of equilibria (fixed points) of the system \eqref{eq: TWS profile} consists of 
the origin together with the unit circle of the $(u,v)-$ plane; 
see \eqref{eq:equilib}. 
There are no non-trivial homoclinic orbits -- all homoclinic orbits are fixed points/equilibria.
However, there are many heteroclinic orbits connecting either distinct points
on the unit circle, or some point on the unit circle and the origin. 
A heteroclinic orbit connecting a pair of fixed points on the unit circle is called a {\it pulse};
its amplitude 
$r^2(x)=u^2(x)+v^2(x)\to 1$ as $x\to\pm\infty$. 
A heteroclinic orbit connecting the origin and the unit circle is 
either a {\it kink} or {\it antikink} (examples of moving {\it domain walls}); 
$r(x)\to0$ tends to $0$ as $x\to-\infty$  and $r(x)\to1$ as $x\to+\infty$ (kink) or $r(x)\to0$ tends to $0$ as $x\to\infty$  and $r(x)\to1$ as $x\to-\infty$ (antikink).
\medskip

\noindent
\subsubsection{\it Pulses: Supersonic and Subsonic} \label{sec:theta}
Consider a traveling pulse solution (with speed $c$ and energy parameter $E$) given by a heteroclinic orbit connecting distinct equilibria on the unit circle.
If this orbit asymptotes to the equilibrium 
$\begin{bmatrix}
    \cos \theta & \sin \theta
\end{bmatrix}$, 
then since $E_c[u,v] = c (u^2+v^2) + 2 u v=E$ is constant, we have 
\begin{equation}\textrm{$E=c+\sin(2\theta)$}\label{eq:Eeqc+}
\end{equation}
 and hence $-1\leq E - c \leq 1$.
It is easy to see that if $-1\leq E - c \leq 1$, then there are four distinct values of $\theta$ in the interval $(-\pi,\pi]$
 satisfying \eqref{eq:Eeqc+}. 
Fix $\theta = \theta_{c,E}\in (-\pi,\pi]$ to be the  solution of $\sin(2\theta)=E-c$ of smallest absolute value.
If $|c|>1$, then the ellipse $E_c[u,v]=c+\sin(2\theta)$ intersects the unit circle at the four points
\begin{subequations}\label{eq:4pts}
\begin{align}
A:& \begin{bmatrix}\cos\theta_{c,E} & \sin\theta_{c,E}\end{bmatrix}\\
B:&\begin{bmatrix}\cos(\frac{\pi}2-\theta_{c,E}) &\sin(\frac{\pi}2-\theta_{c,E})\end{bmatrix} = \begin{bmatrix}\sin\theta_{c,E} & \cos\theta_{c,E}\end{bmatrix}\\
C:&\begin{bmatrix}\cos(-\frac{\pi}2-\theta_{c,E}) &\sin(-\frac{\pi}2-\theta_{c,E})\end{bmatrix}= -\begin{bmatrix}\sin\theta_{c,E} & \cos\theta_{c,E}\end{bmatrix}
\\
D:&\begin{bmatrix}\cos(\pi + \theta_{c,E}) & \sin(\pi + \theta_{c,E})\end{bmatrix}=-\begin{bmatrix}\cos\theta_{c,E}& \sin\theta_{c,E}\end{bmatrix}
\end{align}
\end{subequations}
The point $A$ is a reflection of the point $B$, and the $C$ is a reflection of the point $D$, both with respect to the line $v=u$.
They are distinct points provided $\theta_{c,E}\ne\pi/4$ and form a pair of double points on this line for $\theta=\pi/4$. Thus, any $E$ and $c$ (here with $0<E-c<1$) gives rise to an ellipse with the four intersection points \eqref{eq:4pts}.
Note: if $-1<E-c<0$, then the four points come in pairs which are reflections about the line $v=-u$.

Conversely, given any $\theta\in(-\pi,\pi]$ and $c>1$, we define $E$ via \eqref{eq:Eeqc+} and find that  \begin{align}
    &\textrm{expressions \eqref{eq:4pts}, with $\theta_{c,E}$ replaced by $\theta$,
are the $4$ intersection points of the ellipse}\nonumber\\ 
&\qquad\qquad \c (u^2+v^2) + 2 u v = c +\sin(2\theta). \label{eq:theta-ellipse}
\end{align}

For $|c|<1$,  the level sets $E_c$ are hyperbolae with two branches in the $(u,v)-$ plane. The branches are both symmetric about $v=u$ or about $v=-u$, and each branch intersects the unit circle at two distinct points, for a total (again) of four intersections.


\bigskip

\noindent{\it Supersonic traveling pulses, $|c|>1$:} For $|c|>1$ the level set 
\begin{equation}\label{eq:Ectheta}
    E_c[u,v]= c + \sin 2 \theta
\end{equation} 
is an ellipse which passes through the four points \eqref{eq:4pts} on the unit circle.
The portion of this ellipse which is exterior to the unit circle is called a  {\it bright soliton} pulse,
 and the portion of this ellipse 
 which is interior to the unit circle is called a {\it dark soliton} pulse; see {\sc Figure} \ref{fig: supersonic}.
 For given $c > 0$ and $E > 0$, the orbits corresponding to supersonic pulses with speeds $\pm c$ and phase portrait energy parameters $\pm E$ can be related to one another via the discrete symmetries displayed in {\sc proposition} \ref{prop: discrete symmetries TWS}; see {\sc figure} \ref{fig: supersonic}. 
 \begin{figure}[h!]
     \includegraphics[width = \textwidth]{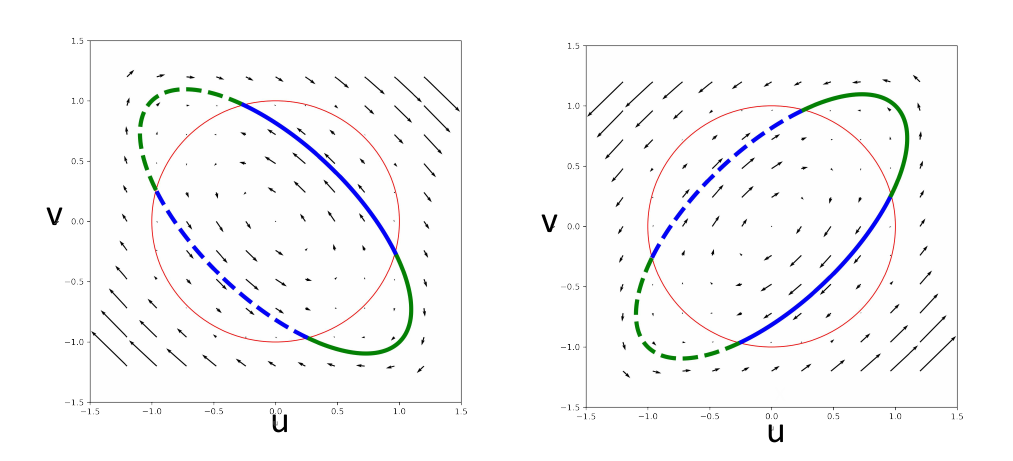}
\caption{
 Supersonic (dark and bright) pulses and their relations via discrete symmetries; Proposition \ref{prop: discrete symmetries}.
Left panel: fixed $c > 1$. Right panel: fixed  $c < -1$.
Dark solid / dashed ellipses correspond to a particular $E$ such that $-1<E-c<1$.
For example, referring to the left panel:
suppose the solid blue curve is denoted $b_{c,E}$. Then, the dashed blue curve is $\mathcal P b_{c,E}=-b_{c,E}$. Further, the
solid green curve is $\mathcal T b_{-c,-E}$, which travels with $c > 1$, and the dashed green curve is $\mathcal{PT} b_{-c,-E}$. 
The curves in the right panel arise by applying the transformation $\mathcal T$ to  curves on the left panel, plotted with the same line colors and styles.}
\label{fig: supersonic}
\end{figure}
\begin{figure}[h!]
\includegraphics[width=\textwidth]{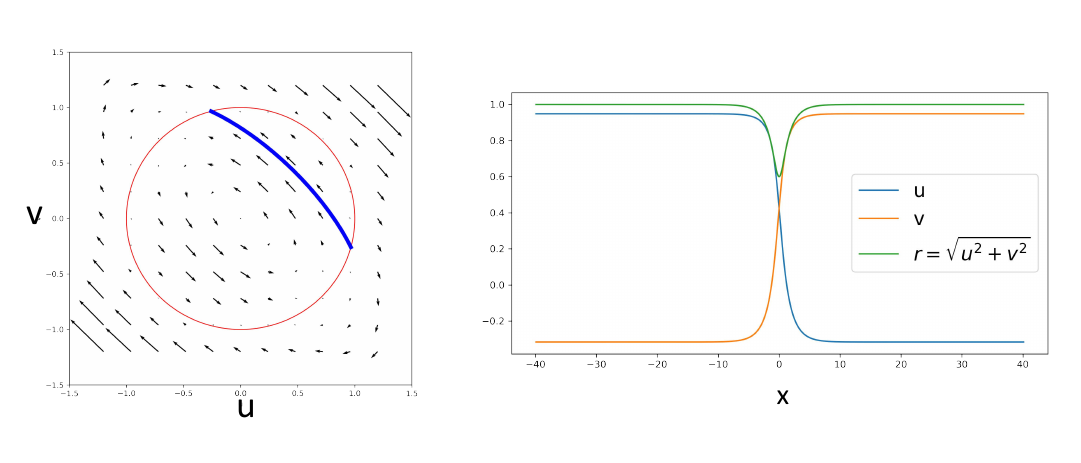}
\caption{Left panel: Orbit (solid curve),  $b_{c,E}(x)=(u(x),v(x))$,  in a phase portrait corresponding to a  typical (dark) supersonic pulse;
    $c>1$ and $-1+c < E < 1+c$. 
     Right panel: plot of components $u(x)$, 
     $v(x)$ and its amplitude $r(x)=\sqrt{u^2(x)+v^2(x)}$.}
\label{fig: pulses}
\end{figure}
 
\noindent{\it Subsonic traveling pulses,  $|c|<1$: } 
From the above discussion, for $|c|<1$ the level set \eqref{eq:Ectheta} is a hyperbola with two branch curves. Each branch curve intersects the unit circle at two points of the points in \eqref{eq:4pts}. The part of a branch curve contained inside the unit disc
is a subsonic dark soliton pulse.
The parts of branch curves which lie outside unit disc are unbounded and correspond to spatially unbounded traveling waves; we do not consider these.
\medskip

\subsubsection{\it Kinks and 
Antikinks as limits of subsonic pulses} \label{sec:subs-to-kink}
For the case of subsonic ($|c|<1$) pulses, the hyperbolic level sets, which determine subsonic (dark) pulses, degenerate, as $|E|\to0$, into two straight lines. For example, if $0<c<1$, then by the expression in \eqref{eq:ellipse+1}, then the limiting two lines are given by:
\begin{equation} u+v = \pm
\sqrt{\frac{1-c}{1+c}}\ (u-v).
\label{eq:4lines}\end{equation}
 These lines determine four line segments in the phase portrait which connect the origin to a distinct point on the unit circle;  see {\sc figure} \ref{fig:kink-phaseport}. 
 \begin{figure}[h!]
 \centering
 \begin{subfigure}[b]{0.49\textwidth}
     \centering
     \includegraphics[width=\textwidth]{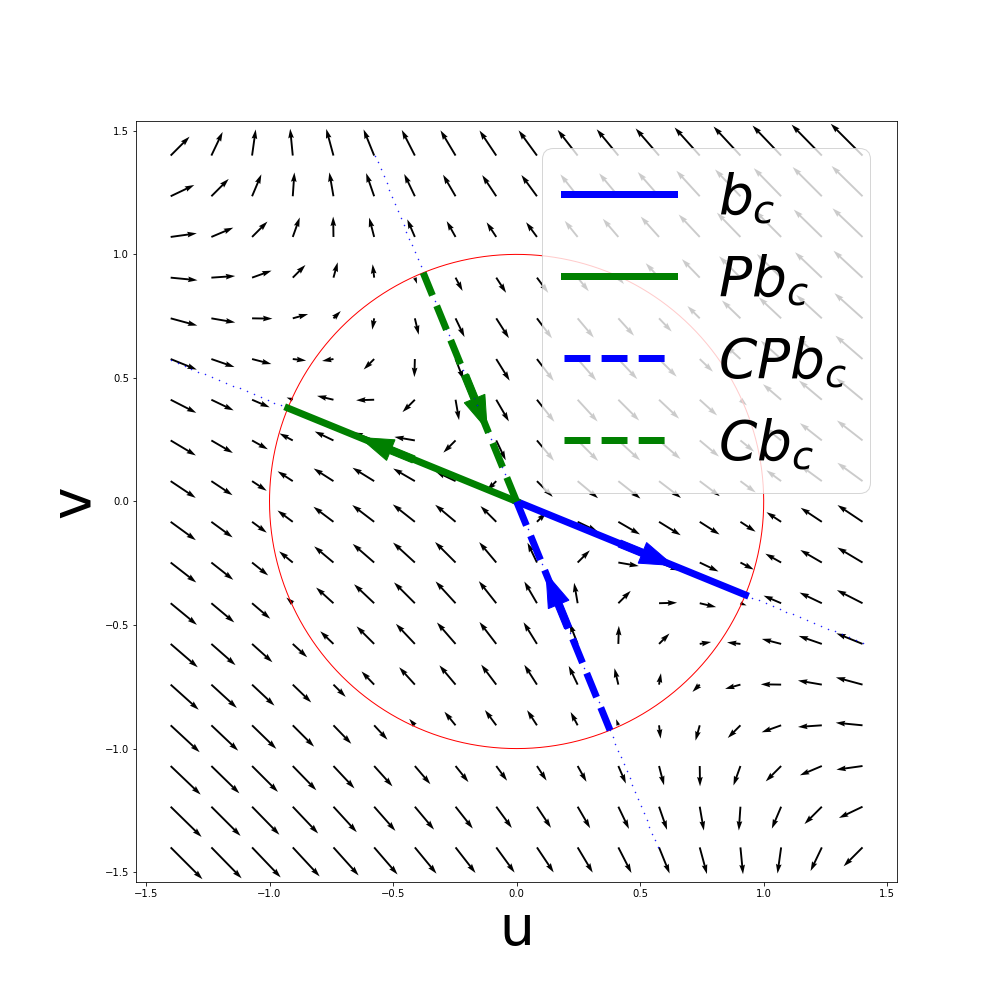}
     \caption{}
     \label{fig: subsonic kink c > 0}
 \end{subfigure}
 \begin{subfigure}[b]{0.49\textwidth}
     \centering
     \includegraphics[width=\textwidth]{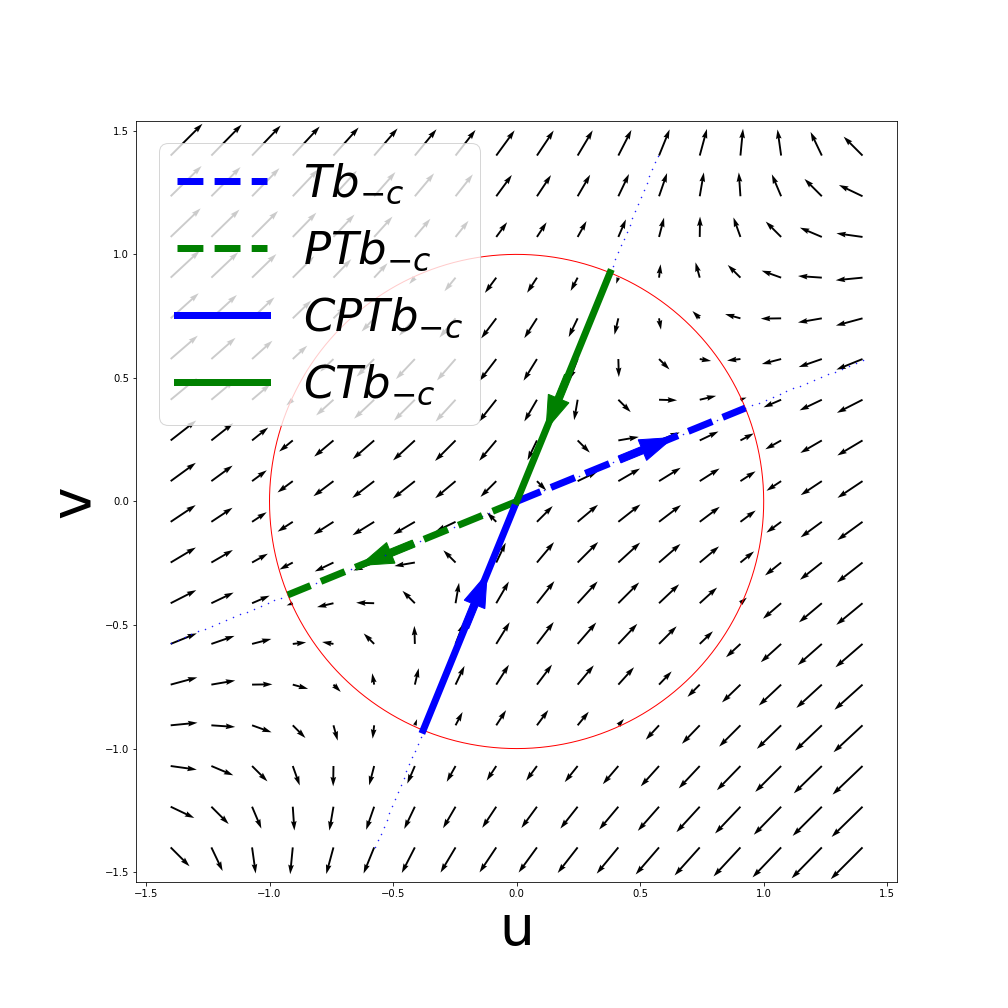}
     \caption{}
     \label{fig: subsonic kink c < 0}
 \end{subfigure}
    \caption{Kinks (solid), antikinks (dashed) and their relation through discrete symmetries. (a) $c \in [0,1)$ (b) $c \in (-1,0]$. 
    In both of the plots, solid lines stand for kinks whose amplitudes $|b|$ increase in the same direction of their speed $c$; dashed lines are antikinks, whose amplitudes decrease \textit{against} the direction of their speed $c$. 
    Heteroclinic connections,
    which are interior to the unit circle and which connect points on the unit circle, 
    which are not highlighted with solid curves, 
    correspond to spatially bounded \textit{subsonic} pulses. 
    These are not stable. 
    See {\sc section} \ref{sec: instability}}
     \label{fig:kink-phaseport}
\end{figure}
 The corresponding four solutions consist of two kinks and two antikinks; those profiles for which the phase portrait trajectory radius $r(x)$ approaches the unit circle in the direction of transport ($\sgn(c)$) are called {\bf kinks}, 
and those for which the approach to the unit circle is in the direction  which is opposite to the direction of propagation, are called {\bf antikinks}. 
See {\sc figure} \ref{fig: kinks} for an example. {\sc proposition} \ref{prop: discrete symmetries TWS}  discusses relations among kink and antikink trajectories under discrete symmetry transformations.
Along these four kink and antikink trajectories the dynamics \eqref{eq: TWS profile} reduces, thanks to \eqref{eq:4lines}, 
to a one dimensional dynamical system
of the general form 
\footnote{Let $b_{c,0}(x)$ ($0<c<1$),  denote the orbit connecting the origin to the unit circle at $(\cos \theta_{c,0}, \sin \theta_{c,0})$, 
with $E[b_{c,0}]=0$.
Its amplitude \textit{increases} to $1$ as $x \to \infty$:
\begin{equation}
    \label{eq: kink profile ODE}
    \begin{aligned}
    u'(x) &= \frac{u(x)\mathcal N(r(x)^2) }{\sqrt{1-c^2}},\ r(0) = \frac{1}{2}
    \\
    v(x) &= - u(x)\tan\left( \frac{1}{2}\arcsin c \right)
\end{aligned}
\end{equation} where $r(x)^2 = u(x)^2 + v(x)^2$. 
}
\[ u'= \tilde{\mathcal{N}}(u^2)\ u,\ v(x)=\textrm{constant}\times u(x).
\]

\begin{figure}[h!]
 \centering
\includegraphics[width=\textwidth]{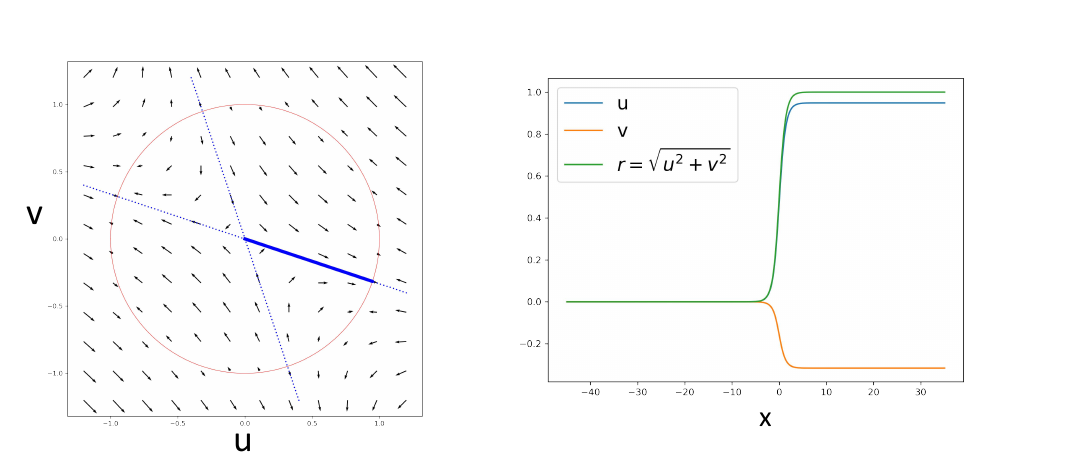}
 \caption{Left panel: Phase portrait with trajectory of kink (solid blue), $b_{c,0}(x)$ which tends to  $[0\ 0]$ as $x\to-\infty$ and to $\begin{bmatrix}
         \cos \theta_{c,0} & \sin \theta_{c,0}
     \end{bmatrix}$ as $x \to +\infty$.
     Right panel: Graphs of components of $b_{c,0}(x)$, $u(x)$ and $v(x)$,  and amplitude $r(x)=\sqrt{u^2(x)+v^2(x)}$.}
\label{fig: kinks}
\end{figure}


\subsection{Convergence rate of heteroclinic orbits to asymptotic equilibria}\label{sec:conv-rate}
Along a trajectory in the phase portrait, 
 as $(u(x),v(x))$ approaches its asymptotic state $\begin{bmatrix}
    u_0 & v_0
\end{bmatrix} ^\mathsf T= \begin{bmatrix}
    \cos\theta & \sin \theta 
\end{bmatrix}^\mathsf T$ 
on the circle, 
along its stable manifold.
The behavior, near an equilibrium $(u_0,v_0)$,
is characterized by  the constant coefficient linearization of \eqref{eq: TWS profile}:
\begin{equation}\label{eq:asy-lin}
    \frac{\dd}{\dd x} 
    \begin{bmatrix}
        \delta u \\ \delta v 
    \end{bmatrix}
    = -\frac{2K}{1-c^2}
    \begin{bmatrix}
        u_0 + c v_0 \\ - c u_0 - v_0
    \end{bmatrix} \begin{bmatrix}
        u_0 & v_0
    \end{bmatrix}
    \begin{bmatrix}
        \delta u \\ \delta v 
    \end{bmatrix}.
\end{equation}
The matrix in \eqref{eq:asy-lin} has eigenpairs:
\begin{align} \mu=0,\quad \Vec{v}_0= \begin{bmatrix} -v_0 \\ u_0 \end{bmatrix},\\
\mu_{\theta} = -\frac{2K}{1-c^2} \cos 2\theta,\quad \Vec{v}_\theta= \begin{bmatrix} u_0 + c v_0 \\ -c u_0 - v_0 \end{bmatrix}.
\end{align}
We may express the nontrivial eigenvalue of the linearization, $\mu_\theta$ in terms of the parameters $E$ and $c$. Along a pulse or a kink solution we must have $E=c+\sin2\theta$. Hence, 
$\sin2\theta=E-c\in (-1,1)$ and therefore $\cos2\theta=\pm\sqrt{1-(E-c)^2}$.

 For a kink, $b_{c,0}$, from $E=c+\sin2\theta_{c,0}=0$.  Therefore, we have for the asymptotic behavior of kink $b_{c,0}$ as $x \to  \infty$:
\begin{equation}
    \begin{bmatrix}
        \delta u (x) \\
        \delta v(x)
    \end{bmatrix} 
    \sim 
    \begin{bmatrix}
        \cos \theta_{c,0} \\
        \sin \theta_{c,0}
    \end{bmatrix} 
    \exp \left(- \frac{2K}{\sqrt{1-c^2}} x \right)
\end{equation}
and as $|x| \to \infty$, for pulses 
\begin{equation}
    \begin{bmatrix}
        \delta u (x) \\ \delta v(x)
    \end{bmatrix}
    \sim 
    \begin{bmatrix}
        u_{\pm\infty} + c v_{\pm \infty}
        \\ 
        -c u_{\pm \infty} - v_{\pm \infty}
    \end{bmatrix}
    \exp \left( - \frac{2K \sqrt{1-(E-c)^2}}{c^2-1} \big|x\big| 
    \right)
\end{equation}
The translation modes $\pd_x b_*(x)$ has the same exponential decaying behavior,  
for kinks:
\begin{equation}
\label{eq: kink translation mode decay}
    \pd_x \begin{bmatrix}
        u (x) \\
        v(x)
    \end{bmatrix} 
    \sim 
    \exp \left(- \frac{2K}{\sqrt{1-c^2}} x \right),
     \quad \textrm{as $x\to \infty$},
\end{equation}
and for supersonic pulses
\begin{equation}
\label{eq: supersonic translation mode decay}
    \pd_x \begin{bmatrix}
        u (x) \\ v(x)
    \end{bmatrix}
    \sim 
    \exp \left( - \frac{2K\sqrt{1-(E-c)^2}}{c^2-1} \big|x\big| 
    \right),\quad \textrm{as $|x| \to \infty$}. 
\end{equation}
Note that the asymptotic behavior does not depend on which of the  four non-trivial equilibria (see \eqref{eq:4pts}) are approached as $x$ tends to infinity.

\subsection{Family of supersonic pulses with fixed asymptotics}
\label{sec: fixed asymptotics}
Consider supersonic pulses with speed $c>1$, obtained from the ellipse
\[ c(u^2+v^2) + 2uv = c+\sin(2\theta).\]
By the observation \eqref{eq:theta-ellipse}, for \underline{fixed} $\theta\ne\pi/4$, the family of ellipses has the same four intersection points with the unit circle for all $c>1$. 
In polar coordinates, $(u,v)=(\rho\cos\phi,\rho\sin\phi)$, these curves are given by 
\[ \rho^2(\phi;\theta)  = \frac{c+\sin(2\theta)}{c+\sin(2\phi)}.\]
The minimum radius $\rho_{min}$, is attained at $\phi=\pi/4$:
\begin{equation} \rho^2_{min}= \rho^2(\pi/4;\theta) = \frac{c+\sin(2\theta)}{c+1}<1.\label{eq:rhomin}\end{equation}

For each $c> 1$ there is a traveling wave trajectory, $b_{c, c + \sin 2\theta}$, that connects $(\cos \theta, \sin \theta)$ to $(\sin \theta, \cos \theta)$; see Figure \ref{fig: fixed asymptotics}.
\begin{figure}
    \centering
    \includegraphics[scale=0.75]{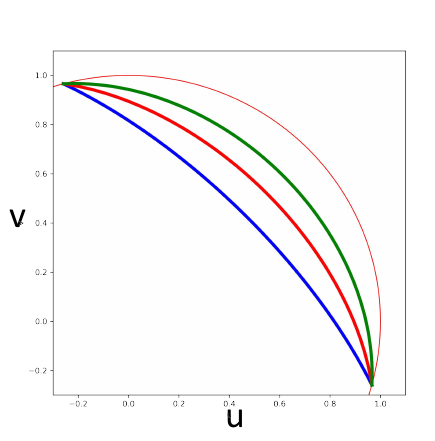}
    \caption{Phase portrait trajectories of (dark) supersonic pulse $b_{c,c+\sin(2\theta)}$, for fixed $\theta$ and several distinct values of $c>1$, with the same set of asymptotics on the unit circle (thin red).
    The pulse whose orbit is in blue travels the slowest, with its profile converging the fastest as $x \to \pm \infty$; 
    the pulse whose orbit is in green travels the fastest, with its profile converging the slowest as $x \to \pm \infty$.
    For the pulse in red, the speed and profile convergence rate to equilibria as $x\to\pm\infty$ are situated between those of the blue and green pulses.}
    \label{fig: fixed asymptotics}
\end{figure}
These pulses  have different rates of approach to their asymptotic equilibria (see \eqref{eq: supersonic translation mode decay}) :

\begin{equation} b_{c,c+\sin2\theta} \sim e^{-\gamma |x|},\quad  \gamma = \frac{2K \cos 2 \theta}{c^2 - 1} = \frac{2K \sqrt{1-(E-c)^2}}{c^2 - 1}.
\label{eq: supersonic pulse asymptotic rate}
\end{equation}
From \eqref{eq:rhomin} and \eqref{eq: supersonic pulse asymptotic rate} we see that supersonic pulses corresponding to slower speeds (say $c\approx 1$, but $c>1$) have a deeper dip in their intensity profile ($z\mapsto u^2(z)+v^2(z)$) and approach their spatial asymptotes rapidly, while supersonic pulses corresponding to faster speeds (say $c\gg1$) have a shallower dip in their intensity profile and approach their spatial asymptotes slowly.

\section{Nonlinear dynamics around traveling wave solutions}
\label{sec: nonlinear results}
Throughout this section, we require the nonlinearity to be \textbf{saturable}; see {\sc section} \ref{sec: nonlinearity}. 
We derive the hyperbolic system of PDEs  \eqref{eq: pert} for 
\begin{equation}
 B(x,t) = \begin{bmatrix} U(x,t) \\ V(x,t) \end{bmatrix},\label{eq:Bpert}
 \end{equation}
the perturbation about a fixed traveling wave solution $b_*$. We then
establish well-posedness of the initial value problem for \eqref{eq: pert}; local well-posedness using standard fixed-point techniques and  global well-posedness via Gr\"{o}nwall's inequality applied to integral inequalities for the norm of the solution. 
We then prove finite propagation speed of information via energy-type arguments. 
Then, we provide a bound on the growth rate of the perturbations, for an arbitrary initial perturbation. We use this bound, together with our finite propagation speed result, to prove nonlinear convective and asymptotic stability of supersonic pulses, for initial pertubations,  $B_0(x)$, which  decay rapidly enough as $x \to +\infty$.

\subsection{Local and global well-posedness}\label{sec:well-posed}
Consider (\ref{eq: PDE in moving frame}) with initial data $b_0$ which is a perturbation of a TWS: $b_0(x) = b_*(x) +  B_0(x)$. 
Writing $  b(x,t) =  b_*(x) +  B(x,t)$,  we obtain that the perturbation  $ B(x,t)$ (see \eqref{eq:Bpert})
solves the initial value problem:
\begin{equation}
\label{eq: pert}
\begin{aligned}
    \begin{bmatrix}
    U \\ V
    \end{bmatrix}_t & = \Sigma \begin{bmatrix}
    U \\ V
    \end{bmatrix}_x +  N_* \big(x; U(x,t),V(x,t)\big) \\
    &:= \begin{bmatrix}
    c & 1\\1 & c
    \end{bmatrix}\begin{bmatrix}
    U \\ V 
    \end{bmatrix}_x + 
    \begin{bmatrix} 
    \mathcal N\big(r_*(x)^2\big) V \\
    - \mathcal N\big(r_*(x)^2\big) U
    \end{bmatrix} +
    \begin{bmatrix}
    \delta \mathcal N_*(x; U,V) v_*(x) \\
    -\delta \mathcal N_*(x; U,V) u_*(x) 
    \end{bmatrix}
\end{aligned}
\end{equation}
where $ r_*( x ) \equiv |b_*(x)| = \sqrt{u_*^2(x) + v_*^2(x)}$ and
\begin{equation}
    \label{eq: delta N *}
    \delta \mathcal N_*(x;U,V):= \mathcal N\Big( \big( u_*(x) + U \big)^2 + \big(v_*(x) + V\big)^2 \Big) - \mathcal N\big( u_*(x)^2 + v_*(x)^2 \big) 
\end{equation}
with initial data $B(x,0) =  B_0(x)$. 
A formally equivalent integral formulation of the IVP is  
\begin{equation}
\label{eq: pert integral}
     B(x,t) = T_c(t)  B_0(x) + \int_0^t  T_c(t-t')
     N_*\big(x, B(x,t')\big) \dd t'
\end{equation}
in terms of the semi-group
\[
     T_c(t) = \exp \left(  t \begin{bmatrix}
    c & 1 \\ 1 &c 
    \end{bmatrix} \pd_x\right) 
\]
A solution $ B(x,t)$ of (\ref{eq: pert integral}) is called a {\bf mild solution} of the initial value problem for \eqref{eq: pert}. 
For $ f \in H^s$, Fourier transform with respect to $x$ of $ T_c(t)  f(x)$ gives
\[
    \big(  T_c(t)  f\big)^\wedge (k) = \exp \left( \ii \begin{bmatrix}
    c & 1 \\ 1 & c
    \end{bmatrix} kt \right) \hat { f}(k)
\]
Therefore $ T_c(t)$ is a unitary group on $H^s$ for all $s\in \mathbbm R$ by Plancherel's identity.

To prove our well-posedness results, we require properties of the nonlinearity $N_*$ (see \eqref{eq: pert}), which we
summarize in the following proposition.
The proof is given in {\sc appendix} \ref{app: nonlinearity}.
\begin{proposition}
\label{prop: nonlinearity} Assume that $\mathcal{N}$ in \eqref{eq: PDE in moving frame} is a saturable nonlinearity.  
    Consider the nonlinear mapping $N_*(B) = N_*\big(x ; U(x), V(x)\big)$ given  in \eqref{eq: pert}. 
    $N_*$ as a mapping only depends on the nonlinearity profile $\mathcal N_*(r^2)$ and the TWS $b_*$. 
      For $N_*$, the following hold:
    \begin{enumerate}
        \item $N_*$ is \underline{globally} Lipschitz on $L^2$. Namely, for any $B,\tilde B\in L^2$, there is a constant independent $B,\tilde B$, such that
        \begin{equation}
            \big\| N_*(B) - N_*(\tilde B) \big\|_{L^2}
            \leq 
            C\big\| B - \tilde B \big\|_{L^2}
        \end{equation}
        \item $N_*$ is \underline{locally} Lipschitz on $H^1$. 
        In particular, for any $B, \tilde B \in H^1$, 
        there is a constant $C$ independent of $B, \tilde B$ such that 
        \begin{equation}
            \label{eq: N * locally lipschitz H 1}
            \big\| N_*(B) - N_*(\tilde B) \big\|_{H^1}
            \leq 
            C\Big( 1 + \min \big\{ \big\|B\big\|_{H^1}, \big\|\tilde B\big\|_{H^1}  \big\} \Big) \big\| B - \tilde B \big\|_{H^1}
        \end{equation}
        As a special case, for any $B \in H^1$, 
        \begin{equation}
            \label{eq: N * linear growth H 1}
            \big\| N_*(B)\big\|_{H^1}
            \leq 
            C \big\| B \big\|_{H^1}
        \end{equation}
    \end{enumerate}
\end{proposition}
The following $H^1$ global well-posedness result holds:
\begin{theorem}
[Global well-posedness of the mild solution of (\ref{eq: pert})]
\label{thm: global well-posedness perturbation H 1}
For $B_0 \in H^1$, \eqref{eq: pert}
has a unique global mild solution $B(x,t) \in \mathcal C^0 \Big( [0,\infty)_t, H^1 \big( \mathbbm R_x \big) \Big)\cap C^1 \Big( [0,\infty)_t, L^2\big( \mathbbm R_x \big) \Big)$, which satisfies  the following exponential bound:
\begin{equation}
\label{eq: H 1 norm exponential growth}
    \big\|B\big\|_{H^1} \leq C e^{Ct} \big\|B_0\big\|_{H^1}.
\end{equation} 
The constant $C$ depends on $\mathcal N$ and $b_*$ but does not depend on the initial data, $B_0$.
\end{theorem}
\begin{proof}[Proof of {\sc Theorem} \ref{thm: global well-posedness perturbation H 1}]
The proof is standard so we only remark briefly on it. Local well-posedness follows from a  standard  application of contraction mapping principle, whose
 hypotheses on the nonlinear term, $N_*$, are verified in {\sc proposition} \ref{prop: nonlinearity}. 
The key to global well-posedness is the bound \eqref{eq: H 1 norm exponential growth} which we now prove.
From \eqref{eq: pert integral} and estimate
\eqref{eq: N * linear growth H 1} we have:
\begin{equation*}
\begin{aligned}
    & \big\|  B \big\|_{H^1} \leq  \big\| B_0\big\|_{H^1} + \int_0^t \Big\| N_* \big( x;  B(x,t') \big) \Big\|_{H^1} \dd t' 
    \leq  \big\| B_0\big\|_{H^1} + \int_0^t  C \big\|  B(x,t') \big\|_{H^1} \dd t'
\end{aligned}
\end{equation*}
Therefore Gr\"{o}nwall's integral inequality\cite{HuNa01} yields \eqref{eq: H 1 norm exponential growth}.
Following standard arguments (see, for example,  \cite[Theorem]{Reed76}), we conclude the global existence of mild solutions to (\ref{eq: pert integral}) in $H^1$.
\end{proof}

\subsection{Finite propagation speed}
\label{sec:finite-prop}
Substitute $u(y,t) = u_*(y,t) + U(y,t)$ and $v(y,t) = v_*(y,t) + V(y,t)$ in \eqref{eq: PDE in lab frame}
 (\textit{non-moving} frame) and we have
\begin{equation}
\label{eq: pert lab}
\begin{aligned}
U_t = &  V_y +\mathcal N_* \big( U,V \big) V + \Big [ \mathcal N_* \big( U,V \big) - \mathcal N_* \big( 0,0 \big) \Big] v_* \\
V_t = & U_y - \mathcal N_* \big(U,V \big) U -  \Big [ \mathcal N_* \big(U,V \big) - \mathcal N_* \big( 0,0 \big) \Big] u_*
\end{aligned}
\end{equation}
which governs the perturbation $B(y,t) = \begin{bmatrix}
    U(y,t) & V(y,t)
\end{bmatrix}^\mathsf T$ 
in the non-moving reference frame. 
For simplicity, we have introduced the notation
\begin{equation}
\label{eq: nonlinearity shorthand}
\begin{aligned}
   \mathcal N_*\big(U,V\big) & := \mathcal N \Big(r_*^2 + 2 u_* U + 2 v_* V + U^2 + V^2 \Big),\quad \textrm{with $u_* = u_*(y-ct)$ etc.}
\end{aligned}
\end{equation}

The following result says the speed of propagation of data for the system \eqref{eq: pert} is at most $|c_0| = 1$:
\begin{proposition}
\label{prop: finite propagation speed}
Consider the initial value problem for the system  (\ref{eq: pert lab})  with initial data $B_0$. Fix $y_0 \in \mathbbm R$ and $ t_0 > 0$.
\begin{enumerate}
    \item 
    \label{prop: finite propagation 1}
    Assume $B_0 \in H^1$. 
    If $B_0(y') = 0$ for all  $y'\in [y_0 - t_0, y_0 + t_0]$, then 
     $B(y',t') = 0$ for all $(y',t')\in \Delta (y_0,t_0)$
     where 
    \begin{equation}
    \label{eq: domain of dependence}
      \Delta (y_0,t_0) := \big\{ (y',t') \in \mathbbm R^2 |\ |y'-y_0|\le t_0-t',\ 0\le t'\le t_0 \big\}  
    \end{equation}
    $\Delta (y_0,t_0)$ is called the domain of dependence of the space-time point $(y_0,t_0)$.
    \item 
    \label{prop: finite propagation 2}
    {Assume 
    $A_0=(W_0,Z_0), B_0=(U_0,V_0) \in H^1$ are such that $B_0(y') = A_0(y')$ for $y' \in [y_0 - t_0, y_0 + t_0]$. Then, 
    $B(y',t') = A(y',t')$ on $\Delta (y_0,t_0)$.}
    \item 
    \label{prop: finite propagation 3}
    Suppose that for some $\ell \in \mathbbm R$, $B_0(y') = A_0(y')$ for all $ y' \in [ \ell, \infty)$.   
    Then, for all  $(y',t')$ satisfying $y' \geq \ell + t$
   we have $B(y',t') = A(y',t')$. 
    In the moving frame with speed $c$ where the spatial coordinate is $x = y - ct$, we have $B(x,t) = A(x,t)$ for all $x \geq \ell - (c-1)t$.
\end{enumerate}
\end{proposition}
It suffices to prove only Part \ref{prop: finite propagation 2} and Part \ref{prop: finite propagation 3} of this proposition. 
Part \ref{prop: finite propagation 1} follows from the case $A_0 = 0$. 
\begin{proof}
Introduce characteristic variables 
\begin{equation}
\label{eq: characteristic variables}
\tilde U = \frac{U+V}{\sqrt2},\qquad \tilde V = \frac{-U+V}{\sqrt2}
\end{equation}
and
\[
\tilde u_* = \frac{u_*+v_*}{\sqrt2},\qquad \tilde v_* = \frac{-u_*+v_*}{\sqrt2}.
\]
Therefore, 
\begin{equation}
\label{eq: pert characteristic}
\begin{aligned}
\big( \pd_t - \pd_y \big) \tilde U & = \mathcal N_*\big( U, V\big)  \tilde V + \Big[\mathcal N_*\big( U, V\big) - \mathcal N_*\big(0,0\big)\Big] \tilde v_* \\ 
\big( \pd_t + \pd_y \big) \tilde V & = - \mathcal N_*\big(U, V\big)  \tilde U - \Big[\mathcal N_*\big(U, V\big) - \mathcal N_*\big(0,0\big)\Big] \tilde u_*
\end{aligned}
\end{equation}
{Denote $A = \begin{bmatrix} W & Z \end{bmatrix}^\mathsf T$  a second solution which satisfies \eqref{eq: pert lab}. Analogously we define, via \eqref{eq: characteristic variables}, 
$\tilde A = \begin{bmatrix} \tilde W & \tilde Z \end{bmatrix}^\mathsf T$.}
Taking the difference, we obtain coupled equations of $\tilde U -\tilde W$ and $\tilde V-\tilde Z$:
\begin{equation}\label{eq: pert characteristic 1}
\begin{aligned}
\big( \pd_t - \pd_y \big) \big(\tilde U -\tilde W \big) & = \mathcal N_*\big(U, V\big)  \tilde V - \mathcal N_*\big(W, Z\big)  \tilde Z + \Big[\mathcal N_*\big( U, V\big) - \mathcal N_*\big(W,Z\big)\Big] \tilde v_* \\ 
\big( \pd_t + \pd_y \big) \big(\tilde V-\tilde Z\big) & = - \mathcal N_*\big(U, V\big)  \tilde U + \mathcal N_*\big(W, Z\big)  \tilde W  - \Big[\mathcal N_*\big(U, V\big) - \mathcal N_*\big(W,Z\big)\Big] \tilde u_*
\end{aligned}
\end{equation}
We next derive an {\it energy inequality} from \eqref{eq: pert characteristic 1}.
Multiply both sides of the first equation of \eqref{eq: pert characteristic 1} with $2 \big( \tilde U - \tilde W\big) $, 
and similarly the second equation in \eqref{eq: pert characteristic 1} with $2\big(\tilde V - \tilde Z\big)$. 
Adding the results gives
\begin{equation}
    \label{eq: energy estimate identity}
\begin{aligned}
    & \big(\pd_t - \pd_y \big) \big(\tilde U - \tilde W \big)^2 + \big(\pd_t + \pd_y\big) \big( \tilde V - \tilde Z \big)^2 
    \\
    =& 2\Big[\mathcal N_*\big( U, V\big) - \mathcal N_* \big( W, Z\big)\Big] 
    \Big[ \big(\tilde U - \tilde W \big) \tilde v_* 
    - \big(\tilde V - \tilde Z \big) \tilde u_* + \tilde U\tilde Z - \tilde V\tilde W 
    \Big]
\end{aligned}
\end{equation}


 Fix $\ell>0$.  Assume $H^1$ initial data $ B_0(y) = A_0(y)$ 
on $y \in [-\ell,\ell]$, 
and consider the  closed trapezoidal region on the $(y,t)$-plane:
\begin{equation}
\label{eq: Omega t}
    \Omega_t = \bigcup_{0\leq s\leq t}\big\{ (y,s)\ :\  -\ell + s \leq y \leq \ell - s \big\},
\end{equation}
where $0 < t < \ell$.
The region $\Omega_t$ is bounded by the four line segments $\Gamma_{1,2,3,4}$ 
and is shown in {\sc figure} \ref{fig: Omega t}.
It is easy to check 
$
\Omega_\ell = \Delta (0,\ell)=[-\ell,+\ell]$; see \eqref{eq: domain of dependence}.

Let \begin{equation}
\label{eq: energy estimate vector field}
\Phi = \Big[ - \big(\tilde U - \tilde W \big)^2  + \big( \tilde V - \tilde Z \big)^2 , \big(\tilde U - \tilde W \big)^2 + \big( \tilde V - \tilde Z \big)^2  \Big].
\end{equation}
Then,  energy identity   \eqref{eq: energy estimate identity} is equivalent to 
\begin{equation}
\label{eq: div Phi}
\mathrm{div}_{y,t} \Phi\  =\  \textrm{RHS of \eqref{eq: energy estimate identity}}
\end{equation}
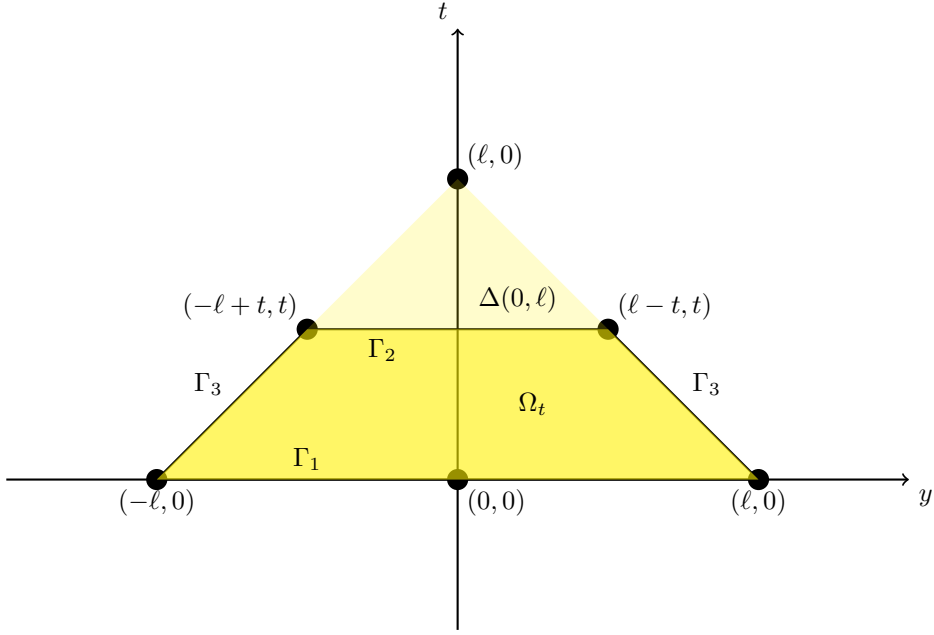
\begin{figure}
    \centering
    \begin{tikzpicture}
[scale=2]
  \draw[thick,->] (-3,0) -- (3,0) node[anchor=north west] {$y$};
  \draw[thick,->] (0,-1) -- (0,3) node[anchor=south east] {$t$};
  \draw[thick] (-2,0) -- (2,0) -- (1,1) -- (-1,1) -- cycle;
  \fill (-2,0) circle (2pt) node[below] {$(-\ell,0)$};
  \fill (2,0) circle (2pt) node[below] {$(\ell,0)$};
  \fill (-1,1) circle (2pt) node[anchor = south east] {$(-\ell + t,t)$};
  \fill (1,1) circle (2pt) node[anchor = south west] {$(\ell - t,t)$};
  \fill (0,0) circle (2pt) node[anchor = north west] {$(0,0)$};
  \fill (0,2) circle (2pt) node[anchor = south west] {$(\ell,0)$};
\fill[fill=yellow,opacity=0.5] (-2,0) -- (2,0) -- (1,1) -- (-1,1) -- cycle;
  \fill[fill=yellow,opacity=0.2] (-2,0) -- (2,0) -- (0,2) -- cycle;
  \node at (0.5,0.5) {$\Omega_t$};
  \node at (0.4,1.2) {$\Delta(0,\ell)$};
  \node[anchor = south] at (-1,0){$\Gamma_1$};
  \node[anchor = north] at (-0.5,1){$\Gamma_2$};
  \node[anchor = south east] at (-1.5,0.5){$\Gamma_3$};
  \node[anchor = south west] at (1.5,0.5){$\Gamma_3$};
\end{tikzpicture}
\caption{Domain of dependence: The closed trapezoidal region $\Omega_t$ used in the proof of {\sc proposition} \ref{prop: finite propagation speed}, shaded with darker yellow;
the domain of dependence of point $(0,\ell)$ is $\Omega_t$ along with the top region shaded with lighter yellow.}
    \label{fig: Omega t}
\end{figure}
{
Integrating \eqref{eq: div Phi} over $\Omega_t$ and applying Gauss's divergence theorem we obtain:
\begin{equation}
\label{eq: gauss div}
\begin{aligned}
 & \int_{\Omega_t}\ \mathrm{div}_{y,t}\ \Phi \dd y \dd t =  \int_{\Gamma_1} \Phi \hat n \cdot \dd l + \int_{\Gamma_2} \Phi \hat n \cdot \dd l + \int_{\Gamma_3} \Phi \hat n \cdot \dd l + \int_{\Gamma_4} \Phi \hat n \cdot \dd l
\\
= & - \int_{-\ell}^\ell \big(\tilde U_0(y) - \tilde W_0(y) \big)^2 + \big(\tilde V_0(y) - \tilde Z_0(y) \big)^2 \dd y \ \text{(this term vanishes)} \\
    & + \int_{-\ell + t}^{\ell - t} \big(\tilde U(y,t) - \tilde W(y,t) \big)^2 + \big(\tilde V(y,t) - \tilde Z(y,t) \big)^2 \dd y \\
     & + \int_0^t 2 \big(\tilde U(-\ell + y,y) - \tilde W(-\ell + y,y) \big)^2 \dd y \ \text{(this term $\geq 0$ )} \\
     & +\int_0^t 2 \big( \tilde V(\ell - y, y) - \tilde Z(\ell - y, y) \big) ^2 \dd y 
     \ \text{(this term $\geq 0$ )} \\
     = & \int_{\Omega_t}  2 \Big[ \mathcal N_* \big(U, V\big) - \mathcal N_* \big(W,Z\big)\Big] 
     \Big[  \big(\tilde U - \tilde W \big) \tilde v_* 
    - \big(\tilde V - \tilde Z \big) \tilde u_* + \tilde U\tilde Z - \tilde V\tilde W 
    \Big] 
     \dd y \dd t
\end{aligned}
\end{equation}
Now we bound pointwise the absolute value of the integrand on the last line. 
Note that $u_*$, $v_*$ are bounded since the TWSs we work with in this paper are all bounded.
The derivative of the nonlinearity, $\mathcal N'$, is also bounded since it is continuous, and its arguments are all bounded.
By the growth rate bound \eqref{eq: H 1 norm exponential growth} of {\sc Theorem} \ref{thm: global well-posedness perturbation H 1}, we have that

 on the fixed time interval $0\leq t\leq \ell$, the $H^1$ thus the $L^\infty$ norm of the perturbation is bounded. Moreover, 
\[
\begin{aligned}
\Big| \mathcal N_* \big(U, V\big) - \mathcal N_* \big(W,Z\big) \Big|
    & \leq C \Big[ |\tilde U - \tilde W| + |\tilde V - \tilde Z|\Big]
\\
\Big| \big( \tilde U - \tilde W\big)  \tilde v_* - \big( \tilde V - \tilde Z \big) 
     \tilde u_* \Big| 
     & \leq C \Big[ \big| \tilde U - \tilde W\big|+\big| \tilde V - \tilde Z \big|\Big]
\\
\big | \tilde U \tilde Z - \tilde V \tilde W \big| \leq \big|\tilde U \big| \big| \tilde W - \tilde Z\big| + \big|\tilde U - \tilde V\big| \big|\tilde W \big| &
\leq C \Big[ \big| \tilde U - \tilde W\big|+\big| \tilde V - \tilde Z \big|\Big],
\end{aligned}
\]
where the constants $C$ depends on the nonlinearity, TWS profile $b_*$ as well as $\ell$. 
These estimates imply that the absolute value of the expression on the last line of \eqref{eq: gauss div} has the upper bound:
\[
\begin{aligned}&
\int_{\Omega_t} 2 \Big|\mathcal N_* \big(\tilde U, \tilde V\big) - \mathcal N_* \big(W,Z\big)\Big|
\Big| \big(\tilde U - \tilde W \big) \tilde v_* 
- \big(\tilde V - \tilde Z \big) \tilde u_* + \tilde U\tilde Z - \tilde V\tilde W \Big|
\dd y \dd t'\\
\leq & C \int_0^t \int_{-\ell + t'}^{\ell - t'} \Big[ \big| \tilde U (y,t') - \tilde W (y,t')\big|^2 +\big| \tilde V (y,t')- \tilde Z (y,t') \big|^2 \Big]  \dd y \dd t'
\end{aligned}
\]
Let $I(t)$ be the nonnegative function of $t$ defined by the third line of \eqref{eq: gauss div}:
\[
    I(t) := \int_{-\ell + t}^{\ell - t} \big|\tilde U(y,t) - \tilde W(y,t) \big|^2 + \big|\tilde V(y,t) - \tilde Z(y,t) \big|^2 \dd y
    \geq 0
\]
Then, 
\[
    I(t) \leq  C \int_0^t I(t') \dd t'
\]
where we used that $I(0) = 0$ by the assumption $B_0(y) = A_0(y)$ on $[-\ell, \ell]$. By Gr\"onwall's inequality
$I(t) = 0$ for  $0 \leq t \leq \ell$. 
This proves Part \ref{prop: finite propagation 2} of  Proposition
\ref{prop: finite propagation speed}.

Part \ref{prop: finite propagation 3} of {\sc proposition}
\ref{prop: finite propagation speed} is an immediate consequence of Part \ref{prop: finite propagation 2}.
To see this, suppose $B_0 (y) = A_0(y)$ on $(\ell, \infty)$ for some $\ell$.
Now if $y > \ell + t$, then 
$[ y - t , y + t] \cap (-\infty, \ell] = \emptyset $. 
Hence, for all such $(y,t)$, we have $B(y',t') = A(y',t')$ for $(y',t')\in \Delta(y,t)$. 
Thus, $B(y,t)=A(y,t)$ for all $(y,t)$ such that  $y > \ell + t$. Equivalently,  in a frame of reference moving with speed $c$: $x = y - ct > \ell - ct + t$. 
See {\sc figure} \ref{fig: lab2move} for illustration of this part of the proof.
Part \ref{prop: finite propagation 3} is thus proved, and the proof of {\sc proposition} \ref{prop: finite propagation speed} is now complete.
}
\begin{figure}
    \centering
    \begin{tikzpicture}
[scale=0.8]
  \draw[thick,->] (-10,0) -- (-2,0) node[anchor=north west] {\footnotesize $y$};
  \draw[thick,->] (-10,1) -- (-10,4) node[anchor=north west] {\footnotesize $t$};
  \draw[thick, color=blue] (-7,0) -- (-4,3);
  \draw[thin, color=green] (-7,0) -- (-3,3);
  
  \node[anchor = east] at (-6,1) {\textcolor{blue}{\footnotesize $y = \ell + t$}};
  \fill (-7,0) circle (2pt) node[anchor = north] {\footnotesize $(\ell,0)$};
  \fill[fill=yellow,opacity=0.2] (-7,0) -- (-2,0) -- (-2,3) -- (-4,3) -- cycle;
  \fill[fill=yellow,opacity=0.5] (-6.2,0) -- (-3.2,0) -- (-4.7,1.5) -- cycle;
  \fill (-4.7,1.5) circle (2pt) node[anchor = west] {\footnotesize $(y,t)$};
  \node[anchor = north] at (-4.7,0.9) {\footnotesize $\Delta(y,t)$};
  \node[anchor = west] at (-4.5,2.3) {\footnotesize $B(\cdot,t) = A(\cdot,t)$};
  \node[anchor = east] at (-5,2.3) {\footnotesize $B(\cdot,t) \neq A(\cdot,t)$};
  \node[anchor = south] at (-5,3.5) {Non-moving frame};
  
  \draw[thick,->] (1,0) -- (8,0) node[anchor=north west] {\footnotesize $x=y-ct$};
  \draw[thick, color = blue] (4,0) -- (3,3);
  \draw[thin, color = green] (4,0) -- (4,3);
  \fill (4,0) circle (2pt) node[anchor = north] {\footnotesize $(\ell,0)$};
  \node[anchor = east] at (3+2/3,1) {\textcolor{blue}{\footnotesize $x = \ell -ct + t$}};
  \fill[fill=yellow,opacity=0.2] (4,0) -- (8,0) -- (8,3) -- (3,3) -- cycle;
  \fill[fill=yellow,opacity=0.5] (4.8,0) -- (7.8,0) -- (4.3,1.5) -- cycle;
  \fill (4.3,1.5) circle (2pt) node[anchor = west] {\footnotesize $(x=y-ct,t)$};
  \node[anchor = north] at (5.5,0.9) {\footnotesize $\tilde\Delta(x,t)$};
  \node[anchor = west] at (3.2,2.3) {\footnotesize $B(\cdot,t) = A(\cdot,t)$};
  \node[anchor = east] at (3.2,2.3) {\footnotesize $B(\cdot,t) \neq A(\cdot,t)$};
  \node[anchor = south] at (3,3.5) {Moving frame};
\end{tikzpicture}
\caption{Finite propagation speed of the perturbation and separated localizations of the traveling pulse core and its perturbations:
 Illustration of the proof of Part \ref{prop: finite propagation 3} of {\sc proposition}
\ref{prop: finite propagation speed}.
Here, at time $t=0$, in the lab frame, $B_0(y)  =A_0(y)$ and in the moving frame where $x = y-ct$, 
$B_0(x)= A_0(x)$. The thin green lines are the spacetime trajectory of a point on the unperturbed supersonic pulse whose coordinate is $x = y = \ell$ at $t = 0$.
The yellow regions in both frames are the spacetime region on on $(y,t)$ and $(x = y-ct,t)$ planes respectively, 
where $B(\cdot,t) = A(\cdot, t)$.
The darker yellow triangles in the left panel stands for $\Delta(y,t)$, 
the domain of dependence of $(y,t)$; whereas in the right panel the \textit{same} domain of dependence of $(x= y-ct,t)$, 
observed in the moving frame.
The blue lines are borders 
of the regions $B(\cdot,t) \neq A(\cdot, t)$ and $B(\cdot,t) = A(\cdot, t)$, in either reference frame respectively.}
    \label{fig: lab2move}
\end{figure}
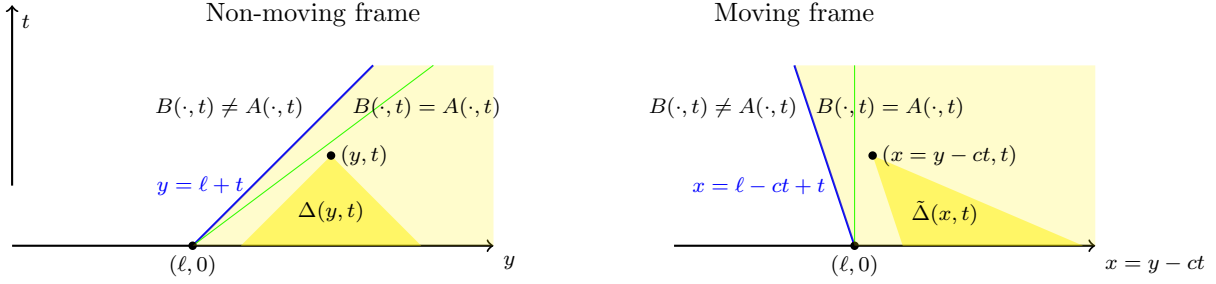
\end{proof}

\section{Nonlinear convective stability of supersonic pulses}\label{sec:nonlin-convec} 

Consider a supersonic pulse, traveling with speed $c>1$.
The following theorem is stated in the comoving frame, i.e., for an ``observer'' that travels with the pulse at speed $c$. 
We continue with the assumption that the nonlinearity is saturable as we do for the last section;
see {\sc section} \ref{sec: nonlinearity}.
\begin{theorem}[Supersonic pulses are nonlinearly convectively stable]
\label{thm: convective}
   Assume a saturable nonlinearity. Let $b_*$ be a supersonic pulse that travels with speed  $c>1$. 
    \begin{enumerate}
    \item
    \label{thm: convective lab} 
    Consider equation \eqref{eq: pert lab} for the perturbation in a \underline{non-moving} frame of reference,   with initial data $B(y,0) = B_0(y)$. 
    Assume that for some $\gamma_0  > 0$, if   
    $\gamma > \frac{\gamma_0}{c-1}$, then
    \begin{equation}
    \label{eq: convective condition}
    \big\| B_0 \big\|_{H^1 \big( [\ell, \infty ) \big) } = o \big(e^{-\gamma \ell} \big)\quad \textrm{as}\quad \ell \to \infty.
    \end{equation} 
    Then, for any $R \in \mathbbm R$, we have exponential time-decay: 
    \begin{equation}
    \label{eq: convective stability lab} \big\| B(y,t) \big\|_{H^1\big( [R + ct, \infty)_y \big) } = \mathcal{O} \Big( e^{- \big ( \gamma(c-1) - \gamma_0 \big) t} \Big).
    \end{equation}
    \item 
    \label{thm: convective moving} 
   Equivalently, consider equation \eqref{eq: pert} for the perturbation of $b_*$ in the \underline{comoving} frame of reference ($x = y-ct$), traveling with the same speed $c$, with initial data $B(x,0) = B_0(x)$ satisfying \eqref{eq: convective condition}. Then, for any $R \in \mathbbm R$, 
    \begin{equation}
    \label{eq: convective stability comoving}
    \big\| B(x,t) \big\|_{H^1\big( [R, \infty)_x \big) } = \mathcal{O} \Big( e^{- \big ( \gamma(c-1) - \gamma_0 \big) t} \Big).
    \end{equation}
    \end{enumerate}
\end{theorem}

In particular, any Gaussian perturbation, $B_0$, centered at an arbitrary point, satisfies \eqref{eq: convective condition}. Numerical simulations of the evolution, which are consistent with {\sc theorem} \ref{thm: convective} are presented  in {\sc figure} \ref{fig: perturbed supersonic pulse simulation}.
{\sc Theorem} \ref{thm: convective} states that as long as the initial perturbation decays fast enough as $x \to \infty$, then
 in any finite window \textit{in the comoving frame with speed $c$}, the solution profile will eventually converge toward the unperturbed profile. 
The details of the perturbation  outside this window are inaccessible to this approach. In fact, in view of the linear exponential instability of non-trivial equilibria, we believe that the perturbation grows outside the window. The numerical simulations of {\sc figure} \ref{fig: perturbed supersonic pulse simulation} support this.
\begin{proof}[Proof of {\sc Theorem} \ref{thm: convective}]
We prove Part \ref{thm: convective moving}; Part \ref{thm: convective lab} is equivalent.
From {\sc Theorem} \ref{thm: global well-posedness perturbation H 1},
there is a constant $\gamma_0$, 
depending only on the nonlinearity and the traveling wave solution $b_*$, such that for all $B_0\in H^1(\mathbb R)$:
\begin{equation}
\label{eq: H 1 exponential growth 1}
\|B(\cdot, t)\|_{H^1(\mathbb{R})} \leq \gamma_0 e^{\gamma_0 t} \|B_0\|_{H^1(\mathbb{R})}
\end{equation}
Now we fix $B_0\in H^1$, and require that it satisfies \eqref{eq: convective condition}, and $B(x,t)$ will denote the solution to \eqref{eq: pert} with initial data $B_0$ for the rest of this proof.

Next, we define a family of initial data given this \textit{fixed} $B_0$. 
Let $\ell\in \mathbbm R$. 
Let $A_0(\ell;y)$ be defined on all $\mathbbm R$ as the function obtained by reflecting the tail of $B_0$ to the right of $y=\ell$ about $y=\ell$:
\begin{equation}
\label{eq: A 0}
    A_0(\ell; y) = \begin{cases}
        B_0(y) & \text{for $y \geq \ell$} \\
        B_0(2\ell - y) & \text{for $y < \ell$}.
    \end{cases}
\end{equation}
A graphical illustration  of 
$\mathcal A(\ell;y)$ is given in {\sc figure} \ref{fig: A 0 ell}.
\begin{figure}
    \centering
    \begin{tikzpicture}
    \begin{axis}[
        no markers, 
        domain=-4:4, 
        samples=100,
        ymin=0,
        axis lines*=left, 
        height=3cm, 
        width=12cm,
        xtick=\empty, 
        ytick=\empty,
        enlargelimits=false, 
        clip=false, 
        axis on top,
        grid = major,
        hide y axis
    ]
    \addplot [very thick,blue] {gauss(0,1)};
    \addplot [very thick,red, domain=2:4] {gauss(0,1)};
    \addplot [very thick,red, domain=-4:2] {gauss(4,1)};
    \node[anchor=south, color = blue] at (axis cs:0,0.4) {$B_0$};
    \node[anchor=south west, color = red] at (axis cs:2.5,0.02) {$A_0(\ell;\cdot)$};
    \draw[dashed] (axis cs:2,0) -- (axis cs:2,0.4);
    \node[anchor = west] at (axis cs:2,0.3) {$y = \ell$};
    \end{axis}
\end{tikzpicture}
    \caption{Illustration
    of the construction 
    of the family of initial data $A_0(\ell;y)$ used in the proof of {\sc Theorem} \ref{thm: convective}.
    This is only a \textit{schematic} plot, since $B_0$ and $A_0$ are vector-valued.}
    \label{fig: A 0 ell}
\end{figure}
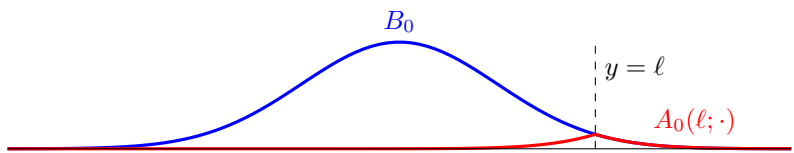
Further, we introduce the corresponding one-parameter family, $A(\ell; x,t)$, of solutions of the IVP \eqref{eq: pert} (posed in the moving frame) with initial data $A(\ell;x,t=0)=  A_0(\ell; x)$. 
 Therefore, by \eqref{eq: H 1 exponential growth 1} we have
\begin{equation}
 \label{eq: H 1 exponential growth 2}
    \big\| A(\ell; \cdot, t ) \big\|_{H^1(\mathbb{R})} \leq \gamma_0 e^{\gamma_0 t} \big\|A_0(\ell; \cdot ) \big \|_{H^1(\mathbb{R})}, \quad \textrm{for all}\ t\ge0\  \textrm{and}\ \ell \in \mathbbm R.
\end{equation}

Now $B(x,t)$ and $A(\ell;x,t)$ are defined for each $x,\ell \in \mathbbm R$ and $t \geq 0$.
Note that $B_0(y) = A_0(\ell; y)$ for $y \geq \ell$, since for $t = 0$, $x = y - ct = y$. 
Then by Part \ref{prop: finite propagation 3} of {\sc proposition} \ref{prop: finite propagation speed}, 
for any $x, \ell \in \mathbbm R$ and any $t \geq 0$:
\begin{equation}
\label{eq: tail identity}
    B(x,t) = A(\ell; x,t),\quad \textrm{if $x + (c-1) t \geq \ell$}
\end{equation}
{\it In the following, we shall use 
as a family of functions $A(\ell; x,t)$, which depends on $\ell,x,t$, and which satisfies the $\ell-$ parameterized family of identities  
\eqref{eq: tail identity} to  bound the solution $B(x,t)$.}
Fix any $R \in \mathbbm R$. Then, for any $t \geq 0$ and $x \geq R$, we have
\[
x + (c-1)t \geq R + (c-1)t = \ell(t).
\] Therefore, by \eqref{eq: tail identity}, let $\ell = R + (c-1)t$, we have, 
\begin{equation}
\label{eq: tail identity 1}
    B(x,t) = A(R + (c-1)t; x,t),\quad \textrm{for all  $x \geq R$ and 
    $t \geq 0$}.
\end{equation} 
See {\sc figure} \ref{fig: tail identity} for illustration as well as the idea of competing growth/decay rates used in the current proof below. 
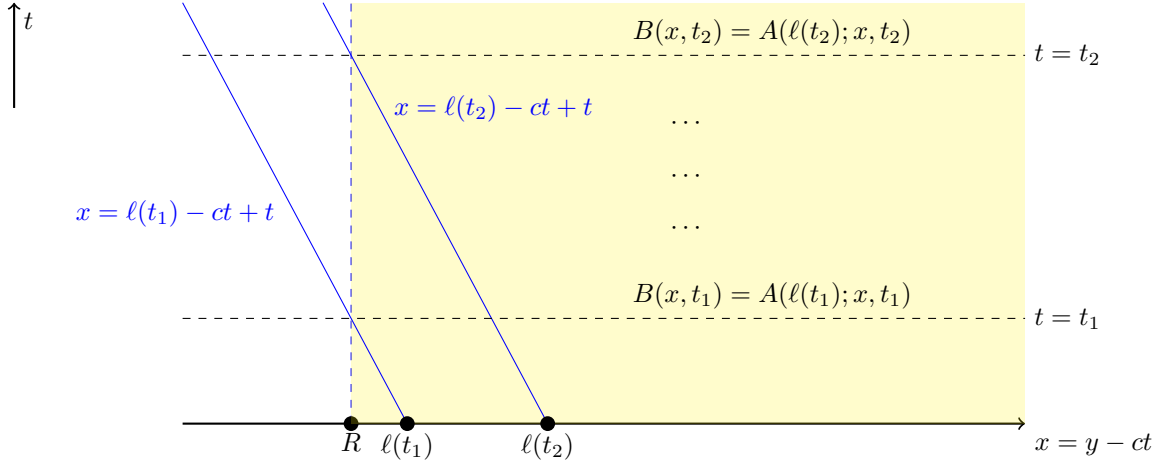
\begin{figure}
    \centering
    \begin{tikzpicture}
    [scale = 1.4, x=1.6cm]
    \draw[thick,->] (-1,0) -- (4,0) node[anchor=north west] {$x=y-ct$};
    \draw[thick,->] (-2,3) -- (-2,4) node[anchor=north west] {$t$};
  \draw[dashed, color = blue] (0,0) -- (0,4);
  \fill (0,0) circle (2pt) node[anchor = north] {$R$};
  \fill[color = yellow, opacity = 0.2] (0,0) -- (0,4) -- (4,4) -- (4,0) -- cycle;
  \draw[dashed] (-1,3.5) -- (4,3.5) node[right] {$t = t_2$};
  \node[anchor = south] at (2.5,3.5) {$B(x,t_2) = A(\ell(t_2); x,t_2)$};
  \draw[blue] (-1/6,4) -- (3.5/3,0);
  \node[anchor = west] at (0.2,3) {\textcolor{blue}{$x = \ell(t_2) - ct + t$}};
  \fill (3.5/3,0) circle (2pt) node[anchor = north] {$\ell(t_2)$};
  \draw[dashed] (-1,1) -- (4,1) node[right] {$t = t_1$};
  \node[anchor = south] at (2.5,1) {$B(x,t_1) = A(\ell(t_1); x,t_1)$};
  \draw[blue] (-1,4) -- (1/3,0);
  \node[anchor = east] at (-0.4,2) {\textcolor{blue}{$x = \ell(t_1) - ct + t$}};
  \fill (1/3,0) circle (2pt) node[anchor = north] {$\ell(t_1)$};
  \node[anchor = north] at (2,3) {$\cdots$};
  \node[anchor = north] at (2,2.5) {$\cdots$};
  \node[anchor = north] at (2,2) {$\cdots$};
  \end{tikzpicture}
    \caption{Illustration of \eqref{eq: tail identity 1} used in the proof of {\sc Theorem} \ref{thm: convective}.
    For all $t \geq 0$, here illustrated by $t = t_1$ and $t = t_2$, the identity \eqref{eq: tail identity 1} holds for all $x \geq R$, 
    colored yellow.
    Compare with the right panel with {\sc figure} \ref{fig: lab2move}.}
    \label{fig: tail identity}
\end{figure}
We note that as a function of $t$, 
the RHS of \eqref{eq: tail identity 1} is not a solution to \eqref{eq: pert}, in general.
It follows that for for all $R\in \mathbbm R$, all $t \geq 0$, we have
\begin{equation}
\label{eq: convective estimate 1}
\begin{aligned}
    &\big\| B(\cdot, t) \big\|_{H^1\big( [R, \infty) \big) } = \big\| A(R + (c-1) t ;\cdot,t) \big\|_{H^1\big( [R, \infty) \big) } 
    \leq  \big\| A(R+(c-1)t;\cdot,t) \big\|_{H^1}
\end{aligned}
\end{equation}
By bound \eqref{eq: H 1 exponential growth 2}
\begin{equation}
\label{eq: convective estimate 2}
    \big\| A(R+(c-1)t;\cdot,t) \big\|_{H^1} \leq \gamma_0 e^{\gamma_0 t} \big\| A_0 (R + (c-1) t ;\cdot ) \big\|_{H^1}
\end{equation}
Combining \eqref{eq: convective estimate 1} and \eqref{eq: convective estimate 2}, 
\begin{equation}
    \label{eq: convective estimate 4}
    \big\| B(\cdot, t)\big\|_{H^1\big([R, \infty) \big) } \leq \gamma_0 e^{\gamma_0t}
    \big\| A_0 (R + (c-1) t ;\cdot ) \big\|_{H^1}
\end{equation}
where the RHS depends only on data $B_0$ through $A_0$ defined by \eqref{eq: A 0}. 
Now we proceed to bound the RHS of \eqref{eq: convective estimate 4}.

By assumption \eqref{eq: convective condition},
we can choose $M \in \mathbbm R$ large enough,
such that  
\[
\big\| B_0 \big\|^2_{H^1 \big( [z, \infty ) \big) } \leq \frac12\left(M e^{-\gamma z}\right)^2,\quad z\in\mathbb R.
\]
By the definition \eqref{eq: A 0} of $A_0$ we have
\begin{equation}
\label{eq: tail H 1 norm}
\big\|A_0(z;\cdot) \big\|^2_{H^1}
 = 2\big\| B_0 \big\|^2_{H^1 \big( [z, \infty ) \big) } \leq M^2 e^{-2\gamma z}
\end{equation}  
Now we apply \eqref{eq: tail H 1 norm} and set $z = R + (c-1) t $, 
\begin{equation}
\label{eq: convective estimate 3}
\big\| A_0 (R + (c-1) t ;\cdot ) \big\|_{H^1} \leq M e^{-\gamma R } e^{-\gamma (c-1) t}
\end{equation}
So by \eqref{eq: convective estimate 4} and \eqref{eq: convective estimate 3}, 
\begin{equation}
    \label{eq: convective estimate}
    \big\| B(\cdot, t) \big\|_{H^1\big( [R, \infty) \big) } \leq 
     M \gamma e^{-\gamma R} e^{ - \big( \gamma(c-1) - \gamma_0 \big) t}.
\end{equation}
Finally, since by hypothesis $\gamma > \frac{\gamma_0}{c-1}$, where $c > 1$,  we have $\gamma(c-1) - \gamma_0 > 0$,  and hence \eqref{eq: convective estimate} implies exponential decay in time as asserted in  \eqref{eq: convective stability comoving}.
\end{proof}


\section{Linearized stability analysis}

\noindent{\bf Motivation:} The nonlinear stability analysis of supersonic pulses of the previous section relies on the assumption that nonlinearity is saturable and, 
in particular, both $\mathcal N(s)$ and its derivative $\mathcal N'(s)$ are bounded. In fact, if the Lipschitz constant of $\mathcal N(s)$ grows with $s$, 
then it may well be that solutions of the IVP do not, in general, exist globally in time; they may, for example, blow up in $H^1$ in finite time. 
Hence, the stability arguments which apply to supersonic pulses with saturated nonlinearities do not apply when considering non-saturable nonlinearities.
%
%
Further, our stability analysis of supersonic pulses does \textit{not} apply to kink-type traveling wave solutions; see Figure \ref{fig: kinks}. 
Thus, we are motivated to study the  linearized stability / instability properties of kinks and pulses.

We next give a general discussion of the linear spectral analysis of (heteroclinic) traveling waves, and then present results 
on the spectral stability properties of equilibria,
which arise as the $x\to\pm\infty$ limits of traveling wave solutions. 
In {\sc section} \ref{sec: pulse stability} we turn to the spectral stability of supersonic pulses, 
and then in {\sc section} \ref{sec: kink stability} to the spectral stability of kink traveling wave solutions.

\subsection{Linearized spectral analysis of pulses and kinks; setup}



Theorem \ref{thm: convective} (for saturable nonlinearities), as well as  numerical simulations described in the introduction and in \cite{du2023discontinuous},  indicate the {\it convective stability} of supersonic pulses and kinks (which are all subsonic). We now study  this effect via a linear spectral analysis, by working in spatially weighted $L^2-$ spaces, which register perturbations which move away from the core of the traveling wave, as decaying with time.

\subsubsection{Exponentially weighted spaces} 
\label{sec: weighted spaces}

We work in exponentially weighted function spaces;
 see, for example, \cite{pego1994asymptotic}, \cite{pego1997convective}, \cite{miller1996asymptotic}.
In particular, we introduce  weights $W(x)$, 
 with differing exponential rates as $x$ tends to plus or minus infinity:
\begin{equation}
\label{eq: weight general}
W(x) \equiv e^{w(x)}= 
\begin{cases}
e^{a_- x} \text{ for $x\leq -1$} \\
        e^{a_+ x} \text{ for $x\geq 1$}
    \end{cases}.
\end{equation}
Here, $a_\pm \in \mathbbm R$, and  $W(x)$, is chosen to be monotone, smooth and to interpolate between $e^{a_-x}$ and $e^{a_+x}$.
We shall work with the weighted Lebesgue  and Sobolev spaces: 
\begin{equation*}
    L^2_w := L^2\big(\mathbbm R, e^{w(x)} \dd x\big),\quad 
    H^1_w:= \Big\{  f(x) \in L^1_\mr{loc}:\ e^{w(x)} f(x) \in H^1 \Big\} 
\end{equation*}
In  our analysis of pulses and kinks,  we'll make use of the following special cases:
\begin{enumerate}
    \item $a_- = a_+ = a$. In this case we take $W(x) = e^{ax}$, and for simplicity we will write
    \begin{equation*}
        L^2_a:= L^2\big(\mathbbm R, e^{ax}\dd x\big),\quad 
        H^1_a:= \Big\{  f(x) \in L^1_\mr{loc}:\ e^{ax} f(x) \in H^1 \Big\}
    \end{equation*}
    \item $a_- = 0$, $a_+ = a \neq 0$. Thus, $W(x) = 1$ for $x<-1$ and $W(x)=e^{ax}$ for $x>1$.
    \item $a_+ = 0$, $a_- = a \neq 0$. 
\end{enumerate}
As we shall see, these choices are determined by the spectral properties of the spatially uniform states to which our heteroclinic traveling waves converge as $x\to\pm\infty$.

\subsubsection{Linearized perturbation equation in a moving frame with speed $c\ge0$}

Let $b_*(x)=\begin{bmatrix} u_*(x) & v_*(x)\end{bmatrix}$ denote a traveling wave solution, which in a frame of reference moving with speed $c$ is a static (time-independent) solution. Define $r_*^2(x)= u_*^2(x)+v_*^2(x)$. In a frame of reference, moving with speed $c$, 
 the perturbation, $B=\begin{bmatrix} U & V \end{bmatrix}^\mathsf T$,
 is governed by equation \eqref{eq: pert}. 
 Keeping only linear terms in 
 \eqref{eq: pert}, we obtain the \textit{linearized perturbation equation}, governing infinitesimal perturbations:
\begin{equation}
\label{eq: linearized operator}
 B_t = L_*  B,\quad    L_* \equiv \Sigma \pd_x + A_*
\end{equation}
with
\begin{equation}
\label{eq: A *}
   A_*(x) \equiv
    \begin{bmatrix}
    {2 \mathcal N'\big(r_*^2 \big)} u_*v_* & \mathcal N_* \big(r_*^2 \big) +  {2 \mathcal N' \big(r_*^2 \big)} v_*^2 \\
     -\mathcal N \big(r_*^2 \big)- {2\mathcal N' \big(r_*^2 \big)}u_*^2 & -{2\mathcal N' \big(r_*^2 \big)}u_*v_*
    \end{bmatrix}
\end{equation} and 
\begin{equation}
\label{eq: formula Sigma}
    \Sigma\equiv  \begin{bmatrix}
    c & 1 \\ 1 & c
    \end{bmatrix}
\end{equation}
We refer to $L_*$ as the \textit{linearized operator about $b_*$}.

Suppose $\lambda\in\mathbb{C}$ 
with $\Re \lambda > 0$, and   $0\neq B_0 (x)\in L^2$ are such that $L_* B_0= \lambda  B_0$. 
Then, $B(x,t)=e^{\lambda t}B_0(x)$ is a solution of \eqref{eq: linearized operator},
such that $\|B(\cdot,t)\|_2$  grows exponentially as $t\to \infty$. 
However, in appropriately weighted spaces  $B(x,t)$ needs not always grow.
Indeed, the weighted perturbation $e^{w(x)}B(x,t)$ satisfies 
\[ 
\partial_t
\left(e^{w(x)}B(x,t)\right) = L_{*,w}\left(e^{w(x)}B(x,t)
\right),  \]
where $L_{*,w}$ is related to $L_*$ by conjugation:
\begin{equation}
\label{eq: conjugate operator}
    L_{*,w} := e^{w(x)} L_* \big( e^{-w(x)} \cdot \big) = \Sigma \big(\pd_x - w'(x) \big) + A_*(x)
\end{equation}
The study of the weighted perturbation $e^{w(x)}B(x,t)$ in $L^2$ or $H^1$ 
is equivalent to the study of $B(x,t)$ in the corresponding weighted spaces $L^2_w$ or $H^1_w$.
\begin{definition}
[Spectral stability]
\label{def: spectral stability}
Let $ b_*(x)$ denote a heteroclinic traveling wave solution of speed $c$, whose profile satisfies (\ref{eq: TWS profile}). 
Let $L_*$ denote the linearized operator of $b_*$. We say that $b_*$ is spectrally stable if
 \[ \textrm{  $L_w^2(\mathbb{R})-$ spectrum of $L_{*}\ \subset \{z: \Re z\le0\}$},\]
 or equivalently 
\[ \textrm{  $L^2(\mathbb{R})-$ spectrum of $L_{*,w}\ \subset \{z: \Re z\le0\}$}. \]
\end{definition}
\medskip

\noindent{\bf Terminology:} Suppose that for a choice of weight $W(x)=e^{w(x)}$ of the exponential type \eqref{eq: weight general} the traveling wave solution $b_*$ is spectrally stable in the sense of Definition \ref{def: spectral stability}. Then, if the context is clear, we shall refer to $b_*$ as being spectrally stable without explicit reference to the particular weight $W(x)=e^{w(x)}$.

\subsection{Spectral stability of equilibria}
\label{sec: equilibria}
We study the spectral stability of equilibria 
\[ 
\text{$ b_*=b_O = \begin{bmatrix} 0 & 0\end{bmatrix} $ and $ b_*(\theta) = \begin{bmatrix} \cos\theta &  \sin \theta\end{bmatrix} $,\quad $\theta\in(-\pi,\pi]$}
\] 
in  $L^2_a = L^2_w$, where  $W(x) = e^{ax}$.
\subsubsection{The trivial equilibrium}
\label{sec: trivial equilibrium}

Consider the trivial equilibrium, $ b_* =  b_O:= [0,0]$, viewed in a frame of reference moving with speed $c$. {\it 
    The relevance of considering the stability of equilbria in different reference frames lies in their determining the essential spectrum of the linearized operator, $L_*$, for heteroclinic traveling waves of speed $c$.}

Let $a\in\mathbb{R}$ be fixed. The $L^2_a$-spectral stability properties are determined by the $L^2$-spectrum of the operator:
\begin{equation}
    \label{eq: L O a}
    L_{O,a} = \begin{bmatrix}
    c & 1 \\ 1 & c 
    \end{bmatrix} \big( \pd_x - a \big) + \begin{bmatrix}
    0 & 1\\
     -1 & 0 
    \end{bmatrix} = \Sigma (\pd_x - a) + A_O
\end{equation}
The spectrum of $L_{O,a}$ is determined \cite{KaPr13} by the frequency of non-trivial plane wave solutions of wave numbers $k\in\mathbbm{R}$:
$B = e^{ikx} B_0$, 
$B_0 \in\mathbbm C^2$ and $L_O B_0 =\lambda  B_0$, where $\lambda = \lambda(k)$ satisfies:
\begin{equation*}
    \det \left( \big( \ii k - a \big) \begin{bmatrix}
    c & 1 \\ 1 & c 
    \end{bmatrix}  + \begin{bmatrix}
    0 & 1\\
     -1 & 0 
    \end{bmatrix} - \begin{bmatrix}
    \lambda & 0 \\ 0 & \lambda
    \end{bmatrix} \right) = 0,
\end{equation*}
yielding two branches (dispersion relations), depending on $a$, 
the union of whose images is exactly the $L^2$-spectrum of $L_{O,a}$, 
equivalently the $L^2_a$-spectrum of $L_O$.
It can be seen that the essential spectrum is stable (does not intersect the open right half plane) if and only if $a=0$.
The two branches of $L^2=L^2_0$ essential spectrum are swept out by the dispersion relations:
\begin{equation}
\label{eq: lambda O}
    \lambda(k) = \lambda_{O}^\pm(k) = \ii \Big(  k c \pm \sqrt{1+k^2} \Big),\quad k\in\mathbbm{R}.
\end{equation}
\begin{proposition}\label{prop:stab-equil}
    The trivial equilibrium, ($b_*=(0,0)$, is spectral stable in $L^2(\mathbb R)$.
\end{proposition}
There are two qualitatively distinct cases: 
\begin{itemize}
    \item 
$|c|<1$ (subsonic frame of reference), 
\[    \sigma(L_O) = \ii \mathbbm R \setminus \ii (-\sqrt{1-c^2}, \sqrt{1-c^2})
\]
{\it i.e.} the spectrum is a subset of the imaginary axis and has a gap, which is symmetric about the origin, and 
\item $|c| >1$ (supersonic frame of reference)
\[
    \sigma(L_O) = \ii \mathbbm R
\]  
\end{itemize}

\subsubsection{Nontrivial equilibria}
\label{sec: nontrivial equilibria}

Nontrivial equilibria are of the form $b_* = b_\theta:= \begin{bmatrix} 
\cos \theta & \sin \theta
\end{bmatrix}$; see \eqref{eq:equil}.
In a frame of reference with speed $c$, the linearized operator $L_\theta = L_*$ is given by \eqref{eq: linearized operator}. The weight-conjugated operator (see \eqref{eq: conjugate operator}) is:
\begin{equation}
\label{eq: L theta a}
    L_{\theta, a} =\begin{bmatrix}
    c & 1 \\ 1 & c 
    \end{bmatrix}
    \big( \pd_x - a \big) 
    -K \begin{bmatrix}
    \sin 2 \theta   & 1 - \cos 2\theta  \\
     -1 - \cos 2\theta & - \sin 2 \theta 
    \end{bmatrix};
\end{equation} 
recall that $K=-\mathcal{N}^\prime(1)>0$ by assumption \ref{eq: K}. The essential spectrum of $L_{\theta,a}$ is characterized by its (bounded)  plane wave solutions $e^{i(kx-\lambda)t)}\xi$, with $\xi\in\mathbb C^2\ne0$. 
Thus, we obtain dispersion curves 
\begin{equation}
\label{eq: lambda theta a}
    \lambda_{\theta, a}^\pm (k) = \big( \ii k - a \big ) c \pm 
    \sqrt{\big (\ii k - a\big ) \big (\ii k - a +2 K\cos 2 \theta \big)}
\end{equation} 
and 
\begin{equation}
    \sigma(L_{\theta, a}) = \bigcup_{\beta=\pm}\{\lambda^\beta_{\theta,a}(k): k\in\mathbb{R}\}. \label{eq:specLtheta-a}
    \end{equation}
    We seek conditions on $a$ guaranteeing that
     $\Re\sigma(L_{\theta,a})\le0$. Note from \eqref{eq: lambda theta a} that for $k\in\mathbb{R}$ 
     \begin{equation}  \Re\lambda_{\theta, a}^\pm (k) = -ac\pm \Re\sqrt{\big (\ii k - a\big ) \big (\ii k - a +2 K\cos 2 \theta \big)} \label{eq:Relampm}\end{equation}
     
Using \eqref{eq:Relampm} we obtain the following bounds on $ \Re \sigma(L_{\theta,a})$:
\begin{proposition}\label{prop:Re_sig}
\begin{align} \inf \Re \sigma(L_{\theta,a}) &= -ac - \big| a -K \cos 2 \theta \big| \label{eq: inf equil spectrum}\\
\label{eq: sup equil spectrum}
    \sup \Re \sigma(L_{\theta,a}) &=  -ac + \big| a - K \cos 2 \theta \big|
\end{align} 
Moreover, these extrema are not achieved if and only if $a - K \cos 2\theta \neq 0$;
otherwise, 
if and only if $a - K \cos2\theta = 0$, $\sigma\big(L_{\theta,a}\big) \subset \ii \mathbbm{R}$. 
\end{proposition}
We shall use the following technical lemma proved in {\sc appendix} \ref{app: proof of lemma: real part bounded}:
\begin{lemma}
\label{lemma: real part bounded} Let $\alpha, \beta\in\mathbb{R}$ be fixed, and consider the mapping $f:k\in\mathbb{R}\mapsto \mathbb{C}$ given by:
\[
f(k) = \sqrt{\big(\ii k - \alpha \big) \big(\ii k - \beta \big) },\quad k\in\mathbb{R}.
\]
where the square-root function is defined 
on the cut complex plane $\mathbbm C \setminus (-\infty,0]$ to have a positive real part; 
on the cut it is taken to have nonnegative imaginary part. 
Then, for all $k\in\mathbb{R}$, we have 
\[
 \Re f(k) \le \Big| \frac{\alpha + \beta} 2 \Big| 
\]
In particular, if $\alpha + \beta = 0$, then $\Re f(k)=0$ for all $k\in\mathbb R$. And if $\alpha + \beta \ne 0$, then $\sup \Re f(k)$ is attained only in the limit $k\to\pm\infty$. 
\end{lemma}
\begin{proof}[Proof of Proposition \ref{prop:Re_sig}]
Applying the Lemma \ref{lemma: real part bounded} with $\alpha=a$ and $\beta=a-2K\cos2\theta$. 
   Therefore \[-ac - \big|a - K \cos 2\theta|
\leq \Re \lambda_{\theta, a}^\pm(k) \leq -ac + \big|a - K \cos 2\theta|\]
Moreover, the lower and upper bounds above are optimal;
if $a \neq K \cos 2\theta$, then the $\sup$ in \ref{eq: sup equil spectrum} and $\inf$ in \ref{eq: inf equil spectrum} are only attained in the limit $k\to \infty$. 
Otherwise 
$\Re\lambda_{\theta,a}^\pm(k)= -ac$ for all $k\in\mathbbm R$.
\end{proof}
By definition, $b_\theta$ is $L^2_a$-spectrally stable  if and only if  
$\sup \Re \sigma(L_{\theta,a}) \leq 0$,
which by \eqref{eq: sup equil spectrum} is the condition
\begin{equation}
    \label{eq: nontrivial equil. stability}
     \big| a -K \cos 2 \theta \big| \leq ac
     \quad {\rm or}\quad -a(c+1)\le-K\cos2\theta\le a(c-1).
\end{equation}
For $c>1$, condition \eqref{eq: nontrivial equil. stability} is equivalent to:
\begin{equation}
\label{eq: supersonic a condition}
    a \geq \max \left\{- \frac{K \cos 2\theta} { c-1} ,  \frac{K\cos 2\theta}{c+1}\right\}\equiv a_{c>1}(\theta,c)
\end{equation}
For $c<-1$, 
\begin{equation}
\label{eq: supersonic a condition 2}
    a \leq \min \left\{ - \frac{K \cos 2\theta}{ c-1} , \frac{K \cos 2\theta}{c+1}\right\}
\end{equation}
For $-1<c<1$, 
\begin{equation}
\label{eq: subsonic equilibria weight condition}
\frac{K \cos 2\theta}{1+c} \leq a \leq \frac{K \cos 2\theta}{1-c}.
\end{equation}
For parameters $a$ which satisfy \eqref{eq: subsonic equilibria weight condition} to exist, it is necessary that the indicated $a-$interval be non-empty. Thus we require:
\begin{equation}
\label{eq: c cos 2 theta geq 0}
c \cos 2\theta \geq 0.
\end{equation} 
We summarize the preceding discussion in:
\begin{proposition}[Spectral Stability of Equilibria in $L^2_a=L^2(e^{ax}dx)$]
    \label{prop: equilibria stability}
Fix a spatially uniform equilibrium, $b_0$ or $b_*(\theta)$, and consider, in a frame of reference moving with speed $c \neq \pm 1$, the linearized operator, $L_{O}$ (trivial equilibrium) or $L_{\theta,a}$ (non-trivial equilibria). 
\begin{enumerate}
    \item Trivial equilibrium: The trivial equilibrium is spectrally stable if and only if $a = 0$. In this case, $L_{O}$ is a subset of the imaginary axis.
\item Non-trivial equilibria, $\begin{bmatrix}
    \cos \theta & \sin \theta
\end{bmatrix}$: The spectrum  $L_{\theta, a}$ 
 is contained in the closed left-half plane if and only if \eqref{eq: nontrivial equil. stability} holds. The condition 
\eqref{eq: nontrivial equil. stability} is equivalent to \eqref{eq: supersonic a condition}
if $c>1$, or \eqref{eq: supersonic a condition 2} if $c<-1$ and \eqref{eq: subsonic equilibria weight condition} if $-1<c<1$.
\end{enumerate}
\end{proposition}

\section{Strategy for studying TWS stability}
The spectrum of a closed operator $L$ on a
Banach space can be uniquely decomposed into two disjoint subsets of $\mathbbm{C} $\cite{KaPr13}: 
the \textbf{essential spectrum} $\sigma_\mr{e}(L)$ and the \textbf{discrete spectrum} $\sigma_\mr{d}(L)$:
$\sigma(L) = \sigma_\mr{e}(L) \cup \sigma_\mr{d}(L)$. 
In particular, for $L_{w,*}$, the linearization about a traveling wave solution, $b_*(x)$, we have
\[
\sigma(L_{w,*}) = \sigma_\mr{e}(L_{w,*}) \cup \sigma_\mr{d}(L_{w,*}).\]
So $\sigma(L_{w,*})$ is contained in the closed left-half plane (and hence  $b_*(x)$ is spectrally stable)
if and only if both $\sigma_\mr{e}(L_{w,*})$ and $\sigma_\mr{d}(L_{w,*})$ are both contained in the closed left-half plane.

For heteroclinic traveling wave solutions, where
$b_*(x)$  approaches  equilibria as $x \to \pm \infty$, the corresponding  {\bf right-} and {\bf left-asymptotic linearized operators} are, 
formally, given by:
\begin{equation}
\label{eq: L pm}
    L_\pm = \lim_{x\to\pm\infty}L_{*,w},
\end{equation} 
$L_+\ne L_-$ due to the heteroclinic nature of our traveling waves. 

Introduce  the  piecewise constant-coefficient {\bf asymptotic operator} which transitions between $L_-$ and $L_+$ across $x=0$:
\begin{equation}
\label{eq: L infty}
    L_\infty =  \mathbbm 1_{x\leq 0}L_- +  \mathbbm 1_{x >0}L_+.
\end{equation}
Since $L_{*,w}-L_\infty$ is spatially well-localized, by Weyl's theorem on the invariance of essential spectrum under relatively compact perturbations, we have 
\[ \sigma_\mr{e} \big(L_{*,w}\big) = \sigma_\mr{e} \big( L_\infty \big);\] see {\sc proposition} \ref{prop: essential spectra identical}.
Therefore
\begin{equation}
\label{eq: decomposition}
    \sigma\big(L_{*,w} \big) = \sigma_\mr{e}\big(L_\infty\big) \cup \sigma_\mr{d}\big(L_{*,w} \big)
\end{equation} 
The following result is used to located the maximum real part over the essential spectrum of $L_{*,w} $ in terms of the operators $L_+$ and $L_-$.
\begin{proposition}
\label{prop: essential spectrum qualitative}
Let $b_*(x)$ denote any TWS with speed $c$, 
which is asymptotic to spatially equilibria as $x\rightarrow \pm \infty$. 
Denote by $L_{*,w}$, the operator obtained by conjugating the linearized operator with the weight $W(x)=e^{w(x)}$ of the exponential type;
see section \ref{sec: weighted spaces} and \eqref{eq: conjugate operator}.  
Finally, denote by $L_\pm$ (where we suppress the dependence on $w(x)$), the constant coefficient asymptotic operators; see \eqref{eq: L pm}.
Then we have
\begin{equation}
\label{eq: L pm spectra in essential spectrum}
    \sigma\big(L_+\big) \cup \sigma \big(L_- \big) \subset \sigma_\mr{e}(L_{*,w})
\end{equation}
and 
\begin{equation}
\label{eq: sup essential spectrum L * w}
    \sup \Re \sigma_\mr{e}(L_{*,w}) = \max\big\{ \sup \Re \sigma(L_+), \sup \Re \sigma(L_-)\big\}
\end{equation}
\end{proposition}
We sketch the proof of Proposition \ref{prop: essential spectrum qualitative} in {\sc appendix} \ref{app: proof of prop: essential spectrum qualitative}. It is a consequence of 
Proposition \ref{prop: essential spectrum qualitative} and a direct application of the theory on the essential spectra of asymptotically constant differential operators \cite[Chapter 3]{KaPr13}.

Before proceeding with a detailed stability analysis,
we note that the linearized spectra of pairs of heteroclinic traveling waves which are related by discrete symmetries,
in {\sc propositions} \ref{prop: discrete symmetries} and \ref{prop: discrete symmetries TWS},
also have simple relations. 
\begin{theorem}
[Discrete symmetry of linearized spectra] 
\label{thm: discrete symmetry of spectra}
Let $ b_*$ be a TWS of speed $c$ and conserved quantity $E_c[b] = E$, see \eqref{eq: E c}.
Let $L_{*,w}$ be the weight-conjugated linearized operator of $b_*$ with weight $W(x) = e^{w(x)}$. 
\begin{enumerate}[label=(\roman*)]
    \item Let
    $b_\bullet = \mathcal P b_*$.  Then, $b_\bullet$ is the profile of a TWS of the same speed $c$ and conserved quantity $E_c = E$.
    Moreover, 
    $L_{\bullet,w}= L_{*,w}$ and hence $\sigma(L_{\bullet,w})=\sigma(L_{*,w})$
    \item Let $b_\bullet = \mathcal {TC} b_* = \begin{bmatrix} v_*(-x) & -u_*(-x) \end{bmatrix}$.
    Then, $b_\bullet$ is the profile of a TWS of speed $-c$ with conserved quantity $E_c = -E$.
    Let $\tilde w(x) = w(-x)$. Then, 
 $L_{\bullet,\tilde w} = L_{*,w}$, and hence
  the $\sigma(L_{\bullet,\tilde w})=\sigma(L_{*,w})$.
\end{enumerate}
\end{theorem}
\begin{remark}\label{rem:pulse2pulse}
{\sc Theorem} \ref{thm: discrete symmetry of spectra} is convenient since it reduces checking the spectral stability of TWSs to checking that of representative ones.
In particular, for supersonic pulses with speed $\pm c$ we only need to work with $b_{c,E}$ and $\mathcal T b_{-c,-E}$,
corresponding to the solid blue line and the solid green line in {\sc figure} \ref{fig: supersonic}.
Other supersonic pulses with speed 
$\pm c$ schematically shown in {\sc figure} \ref{fig: supersonic} can be obtained by acting $\mathcal P$ and $\mathcal T C$ on these two representative solutions. 
Note that pulse solutions are all invariant under $\mathcal C$.
On the other hand, for kinks, we only need to work with $b_{c,0}$ kink with $c\geq0$, represented by the solid blue line in 
{\sc figure} \ref{fig: subsonic kink c > 0}.
For details of how heteroclinic traveling waves transform under discrete symmetries, 
see {\sc appendix} \ref{app: classification}.
\end{remark}

\section{Supersonic pulses}
\label{sec: pulse stability} 
By Remark \ref{rem:pulse2pulse}, we need  only study the two supersonic pulses corresponding to the solid blue and green trajectories in {\sc figure}  \ref{fig: supersonic}. Recall $K = - \mathcal N'(1)>0$; see \eqref{eq: K}.
\begin{theorem}
[Spectral stability for supersonic pulses]
\label{thm: spectral stability of supersonic pulses}
Let $ b_*(x)$ be a supersonic pulse of speed $c > 1$ and corresponding to a trajectory of \eqref{eq: TWS profile} with phase portrait energy $E_c = E$ for $c>1$, 
marked with \underline{solid} blue or \underline{solid} green lines in {\sc figure} \ref{fig: supersonic}.
Then,
\begin{enumerate}
    \item 
$b_*$ is spectrally stable in  $L^2_a$, 
{\it i.e.} $\sup\Re\sigma(L_{*,a})\le0$, 
if and only if 
\begin{equation}
\label{eq: supersonic a condition 1}
    a\geq \frac{K}{c - 1} \sqrt{1 - \big(E-c\big)^2} = \frac{c+1}{2}\gamma>\gamma; 
\end{equation}
see \eqref{eq: supersonic translation mode decay}. 
In fact, 
\[
    \sup \Re \sigma\big( L_{*,a} \big) = - a \big( c - 1 \big) + K \sqrt{1 - \big(E-c\big)^2} \leq 0
\]
Hence, if strict inequality holds in \eqref{eq: supersonic a condition 1}, then $\sup\Re\sigma(L_{*,a})<0$.
\item Assume $E\ne c\pm 1$, or $E = c \pm 1$ and $a \neq 0$, then, $\sigma(L_{*,a})$ is always in the \underline{open} left-half plane;
the supremum is not attained. 
If $E = c \pm 1$ and $a = 0$, 
$\sigma(L_{*,a})$ is in the \textit{closed} left-half plane and not in the open left-half plane.
\item $0\notin\sigma(L_{*,a})$. In particular, the translation mode: $e^{ax}\ \partial_xb_*$, which satisfies $L_{*,a} \left( e^{ax}\ \partial_xb_*\right)=0$, is not an $L^2(\mathbb{R})$ solution of  $L_{*,a}Y=0$.
\item For supersonic pulses with $c<-1$ we have similar results by {\sc Theorem} \ref{thm: discrete symmetry of spectra} and {\sc remark} \ref{rem:pulse2pulse}.
\end{enumerate}
\end{theorem} 
We now proceed with the proof of Theorem \ref{thm: spectral stability of supersonic pulses}. 
We have the decomposition $\sigma(L_{*,a}) = \sigma_\mr{e}(L_{*,a}) \cup \sigma_\mr{d}(L_{*,a})$.
We will prove, for an appropriate choice of   weight $e^{ax}$, that both the {\it essential spectrum} and  the {\it discrete spectrum} 
of operator $L_{*,a}$
are contained in the \textit{open} left-half plane,  except when $E = c \pm 1$ and $a = 0$, when $\sigma(L_{*,a})$ is contained in the closed left-half plane.

\subsection{Essential spectrum for supersonic pulses}\label{sec:ess-spec-super}
The essential spectrum is determined by the operator $L_{*,a}$ evaluated on its asymptotic equilibria; in particular, 
we have the expression on the supremum of its essential spectrum; 
see {\sc proposition} \ref{prop: essential spectrum qualitative}.
For pulses, $b_*(x) = \big[u_*(x),v_*(x) \big] = b_{c,E}$ or $\mathcal T b_{-c,-E}$, are the representative profiles, see 
{\sc remark} \ref{rem:pulse2pulse}.
$b_* = b_{c,E}$ asymptotics to 
\begin{equation}
\label{eq: asymptotic equilibria supersonic pulse}
\begin{aligned}
    \big[
    u_*(-\infty), v_*(-\infty)\big] & = \begin{bmatrix} 
    \cos \theta & \sin \theta \end{bmatrix}\\
    \begin{bmatrix}
    u_*(\infty) & v_*(\infty)
    \end{bmatrix} &= \begin{bmatrix}
    \sin \theta & \cos \theta \end{bmatrix}
    = \begin{bmatrix}
    \cos \Big(\frac\pi2 - \theta\Big) & \sin \Big(\frac\pi2 - \theta\Big) \end{bmatrix}
\end{aligned}
\end{equation} 
while $b_* = \mathcal T b_{-c,-E}$ asymptotics to 
\begin{equation}
\label{eq: asymptotic equilibria supersonic pulse 1}
\begin{aligned}
    \big[
    u_*(-\infty), v_*(-\infty)\big] & = \begin{bmatrix} 
    \cos \theta & \sin \theta \end{bmatrix}\\
    \begin{bmatrix}
    u_*(\infty) & v_*(\infty)
    \end{bmatrix} &= \begin{bmatrix}
    - \sin \theta & - \cos \theta \end{bmatrix}
    = \begin{bmatrix}
    \cos \Big(- \frac\pi2 - \theta\Big) & \sin \Big(- \frac\pi2 - \theta\Big) \end{bmatrix}
\end{aligned}
\end{equation} 
where $\theta=\theta_{c,E} = \frac 12 \arcsin (E-c)$, consistent with definition of $\theta_{c,E}$ in {\sc section} \ref{sec:theta}.

By Proposition \ref{prop: essential spectrum qualitative} we have 
\[
\sup \Re \sigma_\mr{e}(L_{*,a}) = 
\max \big\{\sup \Re \sigma(L_-), \sup \Re \sigma(L_+)\big\}, 
\]
where $L_- = L_{\theta,a}$ and $L_+ = L_{\pm \frac{\pi}{2} - \theta,a}$ whose expressions are given in \eqref{eq: L theta a}, which we shall prove to be non-positive
for $a$ satisfying \eqref{eq: supersonic a condition 1}.
Note that the two choices of $L_+ = L_{\pm \frac\pi2 - \theta, a}$ are identical operators, since by their definition \eqref{eq: L theta a},
$L_{\pm \frac\pi2 - \theta, a}$ only depend on $\pm \frac\pi2 - \theta$ through cosine or sine of $ 2\Big(\pm \frac\pi2 - \theta\Big)$, which give the same values.

Now we apply {\sc proposition} \ref{prop: essential spectrum qualitative} to the weighted operator $L_{*,a}$ of supersonic pulse $b_*(x)$
to find the range of $a$ for both the spectra of $L_\pm$ to be contained in the \textit{open} left-half plane,
or to be on the imaginary axis.
By {\sc proposition} \ref{prop: equilibria stability}, $\sigma(L_-)$ is in the closed left-half plane if \eqref{eq: supersonic a condition} is satisfied:
\[
    a \geq \max \left\{- \frac{K \cos 2\theta} { c-1} ,  \frac{K\cos 2\theta}{c+1}\right\} = \frac{K \cos2\theta}{c+1}
\]
Note that $\cos2\theta \geq 0$. 
Similarly for $\sigma(L_+) = \sigma(L_{\pm\frac\pi2 - \theta, a})$, 
since $\cos 2 (\pm \pi/2 - \theta) = - \cos 2\theta \leq 0$, 
$\sup \Re \sigma(L_+) \leq 0$ if and only if 
\[
    a \geq \max \left\{- \frac{- K \cos 2\theta} { c-1} ,  \frac{- K\cos 2\theta}{c+1}\right\}= \frac{K \cos 2\theta}{c-1}
\]
again as a result of \eqref{eq: supersonic a condition}.
So 
\[
\sup \Re \sigma_\mr{e}(L_{*,a}) = 
\max \big\{\sup \Re \sigma(L_-), \sup \Re \sigma(L_+)\big\} \leq 0
\]
is equivalent to 
\[
a \geq \max \left\{ \frac{K \cos 2\theta}{c-1}, \frac{K \cos 2\theta}{c+1}\right\} 
 = \frac{K \cos2\theta}{c-1} = \frac{K}{c-1} \sqrt{1-(E-c)^2}
\]
since $\theta=\theta_{c,E} = \frac 12 \arcsin (E-c)$ and $\cos 2\theta = \sqrt{1-(E-c)^2}$.
This is exactly \eqref{eq: supersonic a condition 1}.
Moreover, the suprema of $\Re \sigma(L_-) \equiv \Re \sigma(L_{\theta, a})$ and $\Re \sigma(L_+) \equiv \Re \sigma(L_{\pm \frac\pi2 + \theta, a})$ satisfy the following inequality
\[
\sup \Re \sigma(L_{\theta, a}) = - a c - \big| a - K \cos 2\theta \big| 
\leq - a c + \big| a + K \cos 2\theta \big| \leq \sup \Re \sigma(L_{\pm \frac\pi2 + \theta, a})
\]
since $\cos2\theta = \sqrt{1- (E-c)^2} \geq 0$, and there is
\begin{equation}
\label{eq: ess spec of supersonic pulse is sup Re sigma L +}
\sup \Re\sigma_\mathrm{e} (L_{*,a}) = \Re \sigma(L_+).
\end{equation}
\underline{If $E\neq c \pm 1$},
the supremum is never achieved is a consequence of Proposition \ref{prop:Re_sig}, {\it i.e.}
 $\sup \sigma_\mathrm{e} (L_{*,a}) \subset \big\{ \Re z < 0\big\}$.
\underline{If $E = c \pm 1$}, then \eqref{eq: supersonic a condition 1} becomes $a \geq 0$. 
If $a = 0$, then from \eqref{eq: L theta a}, both $\sigma_\mathrm{e}\big(L_\pm\big)$ are both subsets the imaginary axis. So we have proved
\begin{proposition}
    \label{prop: supersonic ess spec}
    \begin{enumerate}
        \item 
   The essential spectrum of $L_{*,a}$, $\sigma_\mr{e}(L_{*,a})$, is contained in the closed left half plane, if and only if condition \eqref{eq: supersonic a condition 1} on the weight parameter $a$ is satisfied.
   \item Assume \eqref{eq: supersonic a condition 1}.
   $\sigma_\mr{e} (L_{*,a})$ is in the open left-half plane, if and only if additionally we have that either $E \neq c \pm 1$ or $a \neq 0$. 
    \item Assume \eqref{eq: supersonic a condition 1}. Then, $\sigma_\mr{e} (L_{*,a})$ is on the imaginary axis if and only if $E = c \pm 1$ and $a = 0$.
     \end{enumerate}
\end{proposition}
\begin{figure}[h!]
    \centering
    \textbf{Spectrum of linearized operator, $L_{*,a}$, for a supersonic pulse}\par\medskip
 \includegraphics[width= \textwidth]{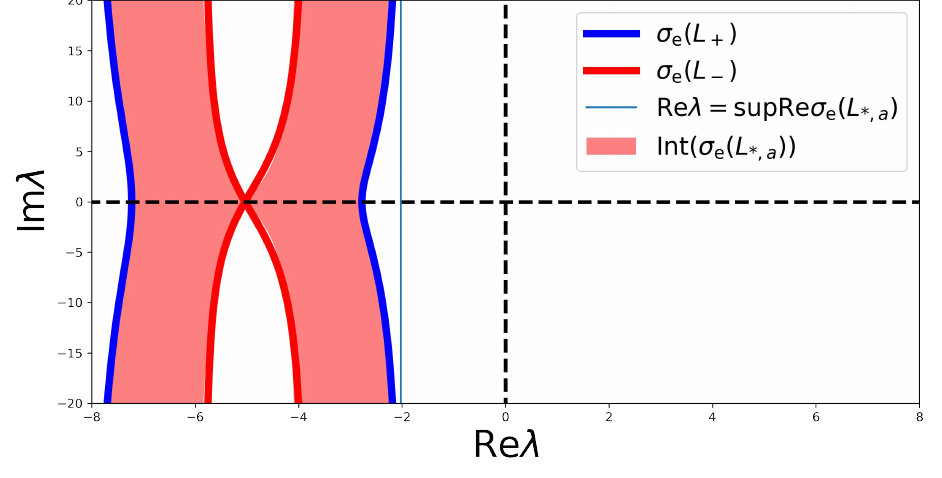}
 \caption{{\it Spectral stability of supersonic pulses via Theorem \ref{thm: spectral stability of supersonic pulses}:} Essential spectrum, $\sigma_{\rm e}(L_{*,a})$ 
  of weight-conjugated linearized operator $L_{*,a}$ for 
 a supersonic pulse $b_* = b_{c,E}$. Parameters:   $K = - \mathcal N'(1) = 2$,
 speed $c = 5$ 
 and phase portrait energy $E_c = 4.84$.
Weight parameter $a = 1$ ($W(x)=e^{x}$) satisfies $a \geq 0.494$, the lower bound in
 \eqref{eq: supersonic a condition 1}. Shown are  spectra of $L_\pm$ (dark blue and red curves)
 which enclose $\sigma_{\rm e}(L_{*,a})$ (shaded light red).
Vertical light blue line: $\Re\lambda=\sup \Re\sigma_\mathrm{e}(L_{*,a})<0$.
There is no discrete spectrum.
 }
 \label{fig: essential supersonic pulse}
\end{figure}
We have plotted the essential spectrum of a supersonic pulse in an appropriate weighted space $L^2_a$,
to illustrate {\sc proposition} \ref{prop: supersonic ess spec}, in {\sc figure} \ref{fig: essential supersonic pulse}.

\subsection{Discrete spectrum for supersonic pulses}
In this section we prove that for \textit{any} $a\in\mathbb R$, 
$L_{*,a}$ does not have any discrete spectrum to the right of its essential spectrum. That is,
 \begin{equation} \Re\sigma_{\rm d}(L_{*,a})\le\sup\Re\sigma_{\rm e}(L_{*,a}). \label{eq:dspec-bnd}\end{equation}
{\sc proposition} \ref{prop: supersonic ess spec} ensures that if
that if \eqref{eq: supersonic a condition 1} holds, 
then $\sup \Re \sigma_{\rm e} (L_{*,a}) \leq 0$ which implies spectral stability,  
$\sup \Re \sigma(L_{*,a}) \leq 0$ along with \eqref{eq:dspec-bnd}.

We now prove \eqref{eq:dspec-bnd}. By definition, $\lambda \in \sigma_\mr{d}(L_{*,a})$ if and only if there is an $L^2$ function $f(x)$ such that $(\lambda I - L_{*,a}) f = 0$. We rewrite this spectral problem  as an equivalent system of ODEs:
\begin{equation}
\label{eq: supersonic eigenfunction}
    \pd_x  f(x) =\Big[  a + \Sigma^{-1} \big(\lambda- A_*(x)  \big)  \Big] f(x)
\equiv \mathcal{A}(x,\lambda) f(x)\end{equation}
where $A_*(x)$ is given by \eqref{eq: A *} and 
$\Sigma$ is given by \eqref{eq: formula Sigma}.
We show that for all $\lambda\in\mathbbm C$ such that $\Re\lambda>\sup\Re \sigma_\mathrm{e}(L_{*,a})$, that $\lambda \not \in \sigma_\mathrm{d} (L_{*,a})$; 
this is the case if for such $\lambda$, 
all nontrivial solutions of \eqref{eq: supersonic eigenfunction} are unbounded as $x\to+\infty$.
We claim that the matrix $\mathcal{A}(x,\lambda)$ has the following properties: 
\begin{itemize}
\item[($\mathcal A$1)] $x\mapsto\mathcal A(x,\lambda) \in \mathcal C^1[0,\infty)$ is a 
$\mathcal C^1$ mapping into the space of complex square matrices which, 
for each $x$, 
varies analytically for all $\lambda \in \mathbbm C$.
\item[($\mathcal A$2)] 
\[ \sup_{\lambda\in \mathbb C}\|\mathcal A(x,\lambda)-\mathcal A_+(\lambda)\|_{\mathbb C^{n\times n}} \rightarrow 0\]
exponentially fast as $x \rightarrow + \infty$, where $\mathcal A_+(\lambda)=\lim_{x\to+\infty} \mathcal{A}(x,\lambda)$ 
is analytic for all $\lambda \in \mathbbm C$.
\item[($\mathcal A$3)] For all $\lambda \in \mathbbm C$, such that $\Re \lambda > \sup \Re \sigma_\mathrm{e} (L_{*,a})$, $\mathcal A_+(\lambda)$
has two eigenvalues with strictly positive real parts.
\end{itemize}
Assuming ($\mathcal A$3) holds, 
the set of all solutions of ODE are expressible 
as a linear combinations of
(a) either two linearly independent solutions with exponential growth rates equal to the (positive) real parts of eigenvalues of $\mathcal A_+(\lambda)$ or
(b) in the case of a non-diagonal Jordan normal form (due to a degenerate eigenvalue $\tilde\mu$) a solution with growth $\sim e^{\tilde\mu x}$ and a solution with growth $\sim x\ e^{\tilde\mu x}$. Hence, all solutions grow exponentially as $x\to+\infty$.

It suffices to verify ($\mathcal A1$)-$(\mathcal A3$) for any given $a \in \mathbbm R$, 
in particular for $a$ that satisfies \eqref{eq: supersonic a condition 1}.
Properties $(\mathcal A1$) and $(\mathcal A2$) are trivial.
We next show  $(\mathcal A3$).
Due to the equivalence of the equation $(\lambda I - L_+)f=0$ to  \eqref{eq: supersonic eigenfunction}, 
the eigenvalues, $\mu$, of $\mathcal{A}_+(\lambda)$ are roots
of the characteristic polynomial, arising by seeking solutions of  $(\lambda I - L_+)f=0$, 
with ansatz $f = f_0 e^{\mu x}$ where $f_0 \in\mathbbm C^2$ is a nonzero vector.
Recall that $L_+$, 
is the weight-conjugated  linearized operator $L_{*,a}$ evaluated at $x=+\infty$ using the nontrivial equilibrium
$\begin{bmatrix}
    \cos \vartheta& \sin\vartheta
\end{bmatrix}^\mathsf T$,
where $\vartheta = \pm\frac\pi2 - \theta$, see \eqref{eq: L theta a}.
Therefore,  $\mu=\mu(\lambda)$ satisfies 
\[
\begin{aligned}
    0 = & \det \left ( \begin{bmatrix}
    c & 1 \\ 1 & c 
    \end{bmatrix} (\mu - a)
    -K \begin{bmatrix}
    \sin 2 \vartheta   & 1 - \cos 2\vartheta \\
     -1 - \cos 2\vartheta & - \sin 2 \vartheta 
    \end{bmatrix} - \begin{bmatrix}
    \lambda & 0 \\ 0 &\lambda
    \end{bmatrix} \right) \\
    = & \big(c^2-1\big) (\mu-a)^2 - 2\big(\lambda c + K \cos 2\theta \big) (\mu-a) + \lambda^2 
\end{aligned}.   
\] 
These roots are given by
\begin{equation}
    \mu_\pm(\lambda) = a + \frac{\lambda c + K \cos 2 \vartheta 
    \pm \sqrt{\lambda^2 + 2 \lambda c K
    \cos 2 \vartheta + K^2 \sin^2 2\vartheta}} {c^2 - 1}.
\label{eq:mu_pm}
\end{equation}
Recall that $K>0$ and $c>1$.
\bigskip

Let 
\[
m(\lambda) = \min \big\{ \Re \mu_+(\lambda), \Re \mu_-(\lambda) \big\}
\]
As $\lambda$ varies over $\mathbbm C$ the \textit{set} of roots $\mu$ is always equal to $\{\mu_-(\lambda),\mu_+(\lambda)\}$. 
In a neighborhood which is small enough of any point where these roots vary analytically. 
To be precise, there is a pair of analytic functions $\mu_1(\lambda)$ and $\mu_2(\lambda)$ in some small neighborhood of $\lambda$, such that $\big\{\mu_1(\lambda), \mu_2(\lambda)\big\} = \big\{\mu_+(\lambda), \mu_-(\lambda)\big\}$.
This is true as long as $\lambda$ is not one of the branch points $\lambda_{1,2}$ at which $q(\lambda) = 0$, where $\sqrt{q(\lambda)}$ is the square-root expression appearing in \eqref{eq:mu_pm}, 
where 
\[q(\lambda) \equiv \lambda^2 + 2 \lambda c K \cos 2 \vartheta + K^2 \sin^2 2\vartheta \]
If $\lambda$ is not on the cut of $\sqrt{q(\lambda)}$,
we can simply choose $\mu_1 = \mu_-$ and $\mu_2 = \mu_+$, otherwise we just perturb the cut so that $\lambda$ is not on it.
Therefore, for $\lambda \neq \lambda_{1,2}$, $m(\lambda)$ is continuous.
Note that $\Re\mu_+(\lambda_j)=\Re\mu_-(\lambda_j)$ for $j=1,2$, and 
$\Re\mu_\pm(\lambda)\to \Re\mu_\pm(\lambda_j)$ as $\lambda\to\lambda_j$.
Hence
\[ m(\lambda)= \min\{\Re\mu_+(\lambda),\Re\mu_-(\lambda)\}\]
varies continuously on $\mathbb C$. 

We claim next that for all $\lambda$ satisfying $\Re\lambda>\sup\Re\sigma_{\rm e}(L_{*,a})$, we have that $m(\lambda)>0$.
Indeed,  for  $\lambda = M\in\mathbbm R$ such that $M\gg1$ :
\begin{equation}
\label{eq: supersonic mu large lambda}
    \mu_+(M) = \frac{M}{c+ 1} + \mathcal O(1),\quad \mu_-(M) = \frac{M}{c- 1} + \mathcal O(1),
\end{equation}
which implies that $m(M)>0$ for all large $M$, since $c>1$.

Now consider $\lambda$ varying over the region $\Re\lambda>\sup \Re \sigma_\mathrm{e} (L_{*,a})$. 
Suppose $m(\lambda)$ is not always positive in this region.
Since $m(\lambda)$ is continuous,
there must be a $\hat \lambda$ for which 
\begin{equation}
\label{eq: hat lambda}
    \Re \hat \lambda > \sup \Re \sigma_\mathrm{e}(L_{*,a})
\end{equation}
such that $m(\hat \lambda) = 0$.
Therefore either $\mu_+(\hat\lambda)$ or $\mu_-(\hat\lambda)$ is purely imaginary.
Without loss of generality we may assume that $\mu_+(\hat\lambda)=\ii\hat\xi$ for $\xi \in \mathbbm R$.
It follows there is a function $Y\sim e^{i\hat\xi x}$,
such that $(\hat\lambda I - L_+)Y=0$. Since $\sigma_{\rm e}(L_+)\subset \sigma_{\rm e}(L_{*,a})$, it follows that $\hat\lambda\in \sigma_{\rm e}(L_{*,a})$. 
However, this contradicts \eqref{eq: hat lambda}.
This contradiction implies that for all $\lambda$ such that \eqref{eq: hat lambda} holds,
we have $\Re\mu_\pm(\lambda)>0$.

Summarizing the result in this section, we have:
\begin{proposition}
\label{prop:no_disc-super}
Let $L_{*,a}$ denote the linearization about a supersonic pulse, with the weight parameter, $a$, satisfying \eqref{eq: supersonic a condition 1}.
Then, 
\begin{enumerate}
\item If $\Re\lambda>0$, then $\lambda$ is not in $\sigma_{\rm d}(L_{*,a})$.
\item$0\notin\sigma(L_{*,a})$. In particular, the translation mode: $e^{ax}\ \partial_xb_*$, which satisfies $L_{*,a} \left( e^{ax}\ \partial_xb_*\right)=0$, is not an $L^2(\mathbb{R})$ solution of  $L_{*,a}Y=0$.
\end{enumerate}
\end{proposition}
We need only verify Part 2. Note that if $b_*(x)$  satisfies \eqref{eq: PDE in moving frame}, then 
$L_*\pd_x b_*=0$ and hence $L_{*,a} \left( e^{ax}\ \partial_xb_*\right)=0$. 
We claim that $\pd_x b_*\notin L^2_a$ for $a$ satisfying \eqref{eq: supersonic a condition 1}. Note that
\[ \pd_x b_*(x) \sim \exp\left(-\frac{2 K \sqrt{1-(E-c)^2}}{c^2-1} \right),\]
from \eqref{eq: supersonic translation mode decay}.
  Further, by \eqref{eq: supersonic a condition 1}, we have
    \[
    a \geq\frac{K \sqrt{1-(E-c)^2}}{c-1} 
    > \frac{2 K \sqrt{1-(E-c)^2}}{(c-1)(c+1)} = \frac{2 K \sqrt{1-(E-c)^2}}{c^2-1}.
    \]
    Therefore
$e^{ax}  \pd _ x b_*(x) \to \infty$ as $x \to \infty$, and $e^{ax}  \pd _ x b_*(x)$ is not in $L^2$.

\subsection{Remark on instability of subsonic pulses and antikinks}
\label{sec: instability}
{\sc Proposition} \ref{prop: essential spectrum qualitative} implies, 
for a TWS to be spectrally stable in some exponential weighted space $L^2_w$, it is necessary that 
both the spectra of asymptotic operators $L_\pm$ of
$L_{*,w}$ is in the closed left-half plane.
Since we have restricted the weight $W(x) = e^{w(x)}$ to be of exponential type defined in \eqref{eq: weight general},
both of the asymptotic equilibria $b(\pm \infty)$ of $b_*$ need to be spectrally stable in some $L^2_{a_\pm}$ space, when observed in the reference frame moving with the same speed $c$ of $b_*$.
Conversely, if it is impossible to find, WLOG, $a_+$,
such either $b(\infty)$ is stable in $L^2_{a_+}$, as a result, for any $W(x)= e^{w(x)}$ of the exponential type, $\sup \Re \sigma_\mathrm{e} (L_{*,w}) \geq \sup \Re \sigma (L_+) > 0$ and $b_*$ is not stable in any $L^2_w$ with exponential type weight.
This is precisely what happens for the subsonic pulses, as well as antikinks.

In fact, since the subsonic pulses and antikinks are all subsonic, it is straightforward to verify that one of the asymptotic equilibria of a subsonic pulse, 
as well as the $b(-\infty)$ equilibrium of an antikink, 
cannot be rendered spectrally stable in any $L^2_a$ space since for these equilibria \eqref{eq: c cos 2 theta geq 0} is violated.

\section{Spectral stability of kinks}
\label{sec: kink stability}

Consider a kink profile $b_*(x)$ satisfying \eqref{eq: kink profile ODE}. 
As noted in {\sc remark} \ref{rem:pulse2pulse},
we may restrict our attention to kinks with speed $ 0 \leq c < 1$, $b_* = b_{c,0}$, 
where $b_{c,0}$ is the solution to \eqref{eq: kink profile ODE},
represented by the solid blue line in {\sc figure} \ref{fig: subsonic kink c > 0}.
The profile $b_*(x)$ tends to the equilibrium $\begin{bmatrix}
    0 & 0 
\end{bmatrix}^\mathsf T$
as $x \to -\infty$ and to a non-trivial equilibrium at $x\to+\infty$.
We have shown in section \ref{sec: trivial equilibrium} that the trivial equilibrium is 
spectrally stable in 
the \underline{unweighted} ($a=0$) space, $L^2(\mathbbm{R})$. 
We will first characterize the essential spectrum of kinks by applying {\sc proposition} \ref{prop: essential spectrum qualitative} again, 
before tackling the problem of locating the discrete spectrum.
\subsection{Essential spectrum for kinks}
Let $L_*$ denote the linearized operator about $b_*$.
Introduce a smooth spatial exponential weight $W(x) =e^{w(x)}$, where
\begin{equation}
\label{eq:w-def}
   \begin{cases}
w(x) = 0 & x\le -1\\
w(x) = ax & x\ge1
 \end{cases} 
\end{equation}

The linearized operator whose $L^2(\mathbb{R};dx)$ spectrum determines
 the spectrum of $L_*$ in $L^2(\mathbb{R};W(x)dx)$ is given by $L_{*,w} = \Sigma \big(\pd_x - w'(x)\big) + A_*(x)$; see \eqref{eq: conjugate operator}. 
\begin{proposition}
\label{prop: essential kink}
Assume that the weight $W(x)=e^{w(x)}$ is given by \eqref{eq:w-def}, where 
\begin{equation}
    \label{eq: essential kink condition}
        K \sqrt{\frac{1-c}{1+c}} \leq a \leq K \sqrt{\frac{1+c}{1-c}}.
    \end{equation}
    Here, recall $K = - \mathcal N'(1)>0$; see \eqref{eq: K}.
    Then, the essential spectrum of $L_{*,w}$ is contained in the closed left-half plane, i.e. $\Re \sigma_\mathrm{e} (L_{*,w}) \le 0$.
\end{proposition}
\begin{proof}
The ODE for a kink profile is given in \eqref{eq: kink profile ODE}.
The kink is a heteroclinic connection between the trivial equilibrium at $x = -\infty$
and the nontrivial equilibrium
$\begin{bmatrix}
    \cos \theta & \sin \theta 
\end{bmatrix}^\mathsf T
$ 
at $x = \infty$, 
with $\theta = \theta_{c,0} = -\frac12 \arcsin c$; see Section \ref{sec:theta}.
By {\sc proposition} \ref{prop: essential spectrum qualitative}, the supremum of essential spectrum $\sigma_\mathrm{e}(L_{*,w})$ is determined
by the spectra of the asymptotic operators $L_- = L_O$ and $L_+ = L_{\theta, a}$; specifically, 
\[
\sup \Re \sigma_\mathrm{e}(L_{*,w}) = \max \Big\{ \sup \Re \sigma(L_{O}), \sup \Re \sigma(L_{\theta,a})\Big\}
\]
The spectrum of $L_O$ is on the imaginary axis, 
by {\sc proposition} \ref{prop:stab-equil},
so $\sup \Re \sigma(L_{O}) = 0$.
Therefore $\sup \Re \sigma_\mathrm{e}(L_{*,w}) \leq 0$ if and only if $\sup \Re \sigma_\mathrm{e}(L_{\theta,a}) = 0$.
By {\sc proposition} \ref{prop: equilibria stability}, the spectrum of $L_{\theta, a}$ is in the closed left-half plane if and only if $a$ satisfies \eqref{eq: subsonic equilibria weight condition}. Hence, $\sup \Re \sigma_\mathrm{e}(L_{*,w}) \le 0$ if and only if 
\eqref{eq: subsonic equilibria weight condition} is satisfied. 
Since $\cos 2\theta = \sqrt{1-c^2}$, 
condition \eqref{eq: subsonic equilibria weight condition} is equivalent to \eqref{eq: essential kink condition}.

\end{proof}
{\sc Figure} \ref{fig: spectrum moving kink} shows a typical case of the essential spectrum of a moving kink. 
\begin{figure}
    \centering
    \textbf{Spectrum of linearized operator, $L_{*,w}$, for a kink}\par\medskip
    \includegraphics[width= \textwidth]{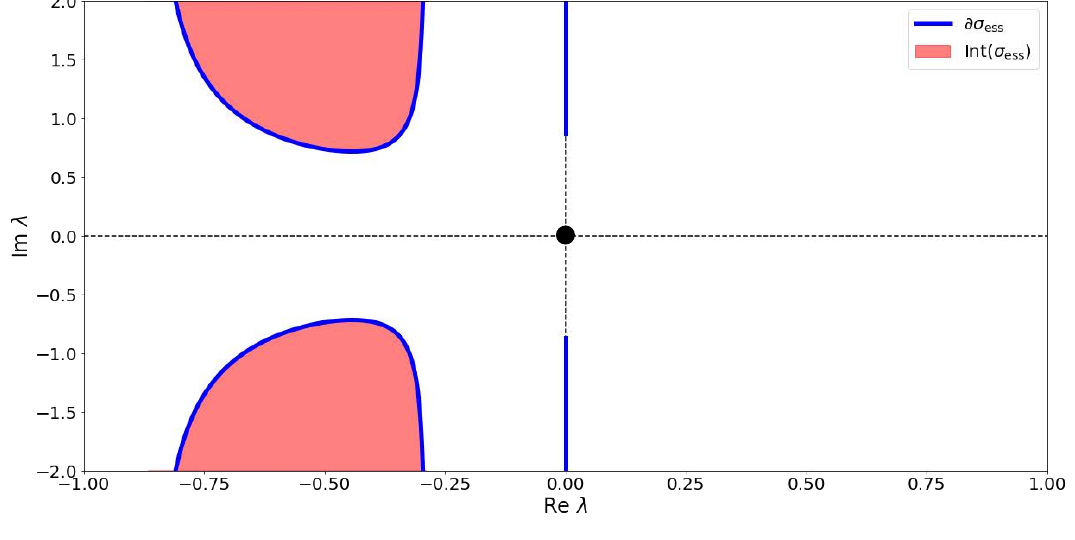}
    \caption{Spectrum of linearized operator, $L_{*,w}$, about a kink with speed $c\ne0$, including $ 0\in \sigma_\mr d \big(L_{*,w} \big) $ in black, showing linear spectral stability in the space $L^2_{w_c}$.  Parameter values: $K=2$, $c=0.5$, and $a=1.15$.}
    \label{fig: spectrum moving kink}
\end{figure}


\subsection{Neutral spectral stability of Non-moving ($c=0$) kinks}
\label{sec: nonmoving kink stability}
    The fact that non-moving kinks are spectrally stable in an appropriate weighted $L^2$ space is a special case of {\sc Theorem} \ref{thm: kink stability} in {\sc section} \ref{sec: moving kink stability} below.
    However, the current {\sc section} \ref{sec: nonmoving kink stability} provides a stronger result and simpler treatment for the $c = 0$ special case.
    Therefore we advise that the readers start with this section first, for pedagogical purposes.
    In particular, in the current section we do not assume that the nonlinearity $\mathcal N(r^2)$ satisfies the concavity condition \eqref{eq: concavity}.
    Moreover, the transformation \eqref{eq: B tilde transform 0} in the current section is a special case
    of the one provided in {\sc proposition} \ref{prop: B tilde B transform} 
    in {\sc section} \ref{sec: moving kink stability}, and is less complicated.
For non-moving kinks with speed $c = 0$ the only possible asymptotic weight as $x \to \infty$ is $e^{Kx}$ if it were to be spectrally stable in $L^2_w$, as a result of {\sc proposition} \ref{prop: essential kink} where the only possible $a$ is $a = K$.
There are four non-moving kinks, each connecting the trivial equilibrium to different nontrivial equilibria, 
but WLOG we study the one that asymptotes to $[1,0]$ as $x \to \infty$, by virtue of {\sc Theorem} \ref{thm: discrete symmetry of spectra}.
The polar angle of this equilibrium on the phase plane of \eqref{eq: TWS profile} is $\theta = 0$, 
setting $c = 0$ and $E = 0$ in \eqref{eq: theta c E}.
Now we note that, 
for this particular kink, 
$u_*(x) = r_0(x) \geq 0$ in \eqref{eq: kink profile ODE}, 
which now implies
\begin{equation}
\label{eq: kink amplitude ODE 0}
    r_0'(x) = r_0(x) \mathcal N \big( r_0(x)^2\big),\quad r_0(0) = \frac12
\end{equation}
So $r_0(x) \to 0$ exponentially as $x \to -\infty$, and $r_0(x) \to 1$ exponentially as $x \to \infty$.

\begin{theorem}[Neutral spectral stability for non-moving ($c=0$) kinks]
\label{thm: kink spectral stability c = 0}
Consider the non-moving kink ($c=0$) namely $ b_*= b_{0,0}$ satisfying \eqref{eq: kink profile ODE} with $c=0$, connecting the trivial fixed point and the fixed point $[1,0]$ on the phase plane.
Let $W(x) = e^{w(x)}\in\mathcal C^1$ such that 
\begin{equation*}
W(x) = e^{w(x)} = \begin{cases}
1 & \text{ for  $x \leq -1$} \\
e^{Kx} & \text{ for $x \geq 1$}
\end{cases}
\end{equation*}
Then, $b_*$ is neutrally spectrally stable in  $L^2_w$, {\it i.e.}
$\sigma\big(L_{*,w}\big)\subset \ii\mathbb{R}$. In particular, 
\begin{enumerate}[label=(\roman*)]
\item
The essential spectrum of $L_{*,w}$ is a subset of the imaginary axis:
\[
    \sigma_\mr{e} \big(L_* \big) = \ii \mathbbm R \setminus \ii \big( -\kappa, \kappa \big),
\]
where $\kappa = \min \big\{ 1,K\big\}>0$.
\item The discrete spectrum of $L_{*,w}$ is a subset of gap, $(-i\kappa, i\kappa)$ on the imaginary axis. Furthermore, $ 0\in \sigma_\mr d \big(L_{*,w} \big) $ always, and is a simple eigenvalue of $L_{*,w}$.
\end{enumerate}
\end{theorem}
The essential spectrum of a non-moving kink 
is a subset of the imaginary axis;
this is because the spectra of the left- and right-asymptotic operators of $L_{*,w}$, namely $L_- = L_O$ and $L_+ = L_{0, K}$, 
are all on the imaginary axis.
This can be seen from \eqref{eq: lambda O} and \eqref{eq: lambda theta a}. 
In fact:
\begin{equation}
    \label{eq: kink c=0 essential spectra}
    \sigma(L_O) = \ii \mathbbm R\setminus \ii (-1,1), \quad \sigma(L_{0,K}) = \ii \mathbbm R \setminus \ii (-K,K)
\end{equation}
Therefore to prove Theorem \ref{thm: kink spectral stability c = 0},
it suffices to locate the discrete spectrum, since the essential spectrum is the union of $\sigma(L_O)$ and $\sigma(L_{0,K})$ given by \eqref{eq: kink c=0 essential spectra}.
We consider the eigenvalue problem
\begin{equation}
\label{eq: ep0}
L_{*,w}\tilde{B} = \lambda\tilde{B} 
\end{equation}
and $ \lambda \in \mathbbm C \setminus \sigma_\mr{e} \big( L_{*,w} \big)$  and $0\ne\tilde {B} \in H^1$ satisfies \eqref{eq: ep0}. 
From \eqref{eq: conjugate operator} and \eqref{eq: A *} we have an equivalent formulation of \eqref{eq: ep0}:
\begin{equation}
\label{eq: ep0 explicit}
L_{*,w}\tilde{  B}(x) \equiv 
    \left(\begin{bmatrix}
    0 & 1 \\ 1 & 0 \end{bmatrix} \big(\pd_x - w'(x) \big)+ \begin{bmatrix}
    0 & \mathcal N_0  \\ - \mathcal N_0 - 2 \mathcal N_p & 0 \end{bmatrix}\right)\tilde{  B}(x)   = \lambda \tilde{  B}(x),
\end{equation}
where 
\begin{equation}
\label{eq: N 0 and N p}
    \mathcal N_0(x) := \mathcal N\big(r_0(x)^2 \big),\ \mathcal N_p (x) := \mathcal N'\big(s \big)\Big|_{s=r_0(x)^2}\ r_0(x)^2
\end{equation}
Note that $v_* = 0$ for kinks with speed $c=0$, see \eqref{eq: kink profile ODE} and set $c=0$ there.
Multiplying both sides of  \eqref{eq: ep0 explicit} by the Pauli matrix $\sigma_1$, we find that \eqref{eq: ep0 explicit} is equivalent to
\begin{equation}
   \partial_x\tilde{B}(x)= \begin{bmatrix}
        \mathcal N_0(x) + 2 \mathcal N_p (x) + w'(x)  & \lambda \\ 
        \lambda & - \mathcal N_0(x) + w'(x)
    \end{bmatrix}\tilde{  B} (x)
\label{eq: ep0 explicit1}\end{equation}
Introducing an integrating factor, we rewrite \eqref{eq: ep0 explicit1} as 
\[
\begin{aligned}
    & \pd_x \bigg( \exp \bigg[ -w(x)- \int_{-\infty} ^x  \mathcal N_p(y) \dd y \bigg]\tilde { B} (x) \bigg)\\
    =& 
    \begin{bmatrix}
        \mathcal N_0(x) + \mathcal N_p (x)  & \lambda \\ 
        \lambda & - \mathcal N_0(x) - \mathcal N_p(x)
    \end{bmatrix}
     \bigg( \exp \bigg[ -w(x)- \int_{-\infty} ^x  \mathcal N_p(y) \dd y \bigg] \tilde{ B}(x) \bigg)
\end{aligned}\]
So  $B(x) \in \mathcal C^1$ is a solution of 
\begin{equation}
\label{eq: reduced GEP c = 0}
 \partial_x B (x) = \begin{bmatrix}
        \mathcal N_0(x) + \mathcal N_p (x)  & \lambda \\ 
        \lambda & - \mathcal N_0(x) - \mathcal N_p(x)
    \end{bmatrix}  B(x)
\end{equation}
if and only if 
\begin{equation}
\label{eq: B tilde transform 0}
    \tilde { B}(x) = \exp \Big( w(x) + \int_{-\infty}^x \mathcal N_p(y) \dd y \Big)  B(x)
\end{equation}
is a $\mathcal C^1$ solution to  \eqref{eq: ep0}.
That the eigenvalue problem associated with \eqref{eq: ep0} is equivalent to the eigenvalue problem associated with equation \eqref{eq: B tilde transform 0}, is a consequence of the following lemma, which is proved in {\sc appendix} \ref{app: proof of lemma: c = 0 both L 2}:
\begin{lemma}
\label{lemma: c = 0 both L 2}
    Let $\lambda \in \mathbbm C$, and let $ B(x)$ and $\tilde { B}(x)$ be related by \eqref{eq: B tilde transform 0}. Then, $ B\in L^2$ if and only if $\tilde { B} \in L^2$. 
\end{lemma}
It then follows immediately that:
\begin{corollary}
\label{cor: equivalence of eps}
The pair $(\lambda, \tilde B)$, with $0\ne \tilde B\in H^1$ solves the eigenvalue problem \eqref{eq: ep0}, if and only if $(\lambda, B)$ where $ 0 \neq B \in H^1$ solves the eigenvalue problem \eqref{eq: reduced GEP c = 0}.
\end{corollary}
Now, it is convenient to write out \eqref{eq: reduced GEP c = 0}:
\begin{equation}
\label{eq:writout text}
\begin{aligned}
U_x & = \big[\mathcal N_0(x) + \mathcal N_p (x)\big] U + \lambda V
\\
V_x & = \lambda U - \big[\mathcal N_0(x) + \mathcal N_p (x)\big] V
\end{aligned}
\end{equation}
We are now able to reduce the problem of finding eigenpairs of $L_{*,w}$ to finding bound state and bound state energies of some linear Schr\"{o}dinger operator on the real line, as shown in the following lemma.
For proof of this lemma see {\sc appendix} \ref{app: proof of lemma: schroedinger}.
\begin{lemma} 
\label{lemma: schroedinger}
Assume that the pair $(\lambda, B(x))$ where 
\[\lambda\in\mathbb C\quad{\rm and}\quad  0\ne B(x) \equiv \begin{bmatrix}
     U(x) &V(x)
\end{bmatrix}^\mathsf T \in L^2
\]
solves \eqref{eq: reduced GEP c = 0}. 
Then,
\begin{itemize}
    \item[(i)] $(\mathcal{E}=-\lambda^2, U(x) )$ with $0\ne U\in L^2$ and $ -\lambda^2\in\mathbbm R$,
is an eigenpair for the eigenvalue problem:
\[ \left(-\pd_x^2 + \mathcal V(x)\right)U(x) = \mathcal{E}U(x),\quad U(x) \in L^2,\]
with the real-valued potential
\begin{equation*}
    \mathcal V(x) =  \big[\mathcal N_0(x) + \mathcal N_p (x)\big]_x + \big[\mathcal N_0(x) + \mathcal N_p (x)\big]^2.
\end{equation*}
\item[(ii)] Moreover, 
$\mathcal E = 0$ is the ground state of operator $- \pd_x^2 + \mathcal V(x)$.
\end{itemize}
\end{lemma}
We are ready to prove {\sc Theorem} \ref{thm: kink spectral stability c = 0}
\begin{proof}[Proof of {\sc Theorem} \ref{thm: kink spectral stability c = 0}]
Part (i) follows Part (i) of {\sc Theorem} \ref{thm: essential spectrum of L infty} and Part (ii) of {\sc Theorem} \ref{thm: border of essential spectrum}, 
with $\sigma(L_\pm)$ given by \eqref{eq: kink c=0 essential spectra}

It follows from {\sc lemma} \ref{lemma: schroedinger},
by self-adjointness of $-\pd_x^2 + \mathcal V$ that $-\lambda^2\in\mathbb{R}$
and therefore,
\[\text{ $\lambda$ is either real or purely imaginary.}\]
We may constrain $\lambda$ further. Indeed, note that  $\mathcal N_p(x) = \mathcal N'\big(r_0^2 (x) \big) r_0^2(x)$ approaches $0$ as $x \to -\infty$  (since $r_0(-\infty) = 0$) and approaches $-K = \mathcal N'(1)$ as $x \to +\infty$ (since $r_0(\infty) = 1$). 
Hence, $\mathcal V(-\infty) = 1$, and $\mathcal V(\infty) = K^2>0$. It follows that $-\lambda^2=\mathcal{E}<\min\{1,K^2\}$. It follows that
\begin{equation}
\label{statement lambda}
    \text{ $\lambda$ is either real or purely imaginary with $\big|\Im \lambda\big| < \min\{1, K\} = \kappa$.}
\end{equation}
Now we claim there are no real \textcolor{blue}{non-zero} eigenvalues of $L_{*,w}$. Once proved, this implies that  there are 
no  eigenvalues with nonzero real part,
and the non-moving kink is neutrally spectrally stable in $L^2_w$.\\
Suppose there exists $0\ne\lambda \in \mathbbm R$ and $\tilde B \in H^1$ that solves \eqref{eq: ep0}, 
then from {\sc corollary} \ref{cor: equivalence of eps} and {\sc lemma} \ref{lemma: schroedinger} 
there would be an eigenpair $\mathcal E = -\lambda^2 < 0$ and $U\in L^2$ of the operator $-\pd_x^2 + \mathcal V(x)$. 
In fact, this is impossible since $\mathcal E = 0$ is the ground state energy of $-\pd_x^2 + \mathcal V(x)$. Indeed, from \eqref{eq:writout text} $U(x)$ is a multiple of $\exp(\int^x\mathcal{N}_0(s)+\mathcal{N}_p(s) ds)$,
which does not change sign.

Therefore there is no real eigenvalues of $L_{*,w}$, and proof of Theorem \ref{thm: kink spectral stability c = 0} is concluded.
\end{proof}

\begin{remark}
We remark that $\sigma\big(L_{*,w} \big)$ has no embedded eigenvalues in the essential spectrum, see statement \eqref{statement lambda}, 
and that $\sigma_\mr d \big(L_{*,w}\big)$
and is a finite subset of the imaginary interval $\big( -\ii\kappa , \ii\kappa \big)$. Indeed, 
these facts hold for any potential, $\mathcal{V}$, satisfying  $\int (1 + |x|)\mathcal V(x) < \infty$
 \textcolor{red}{\bf citation of Fadeev criterion}
\end{remark}

\subsection{Spectral stability of moving kinks $(0\leq c<1)$}
\label{sec: moving kink stability}
Introduce the smooth weight $W(x) = e^{w(x)}$, where  
\begin{equation}
\label{eq: W c weight}
    W_c(x) = e^{w_c(x)} = \begin{cases}
1 & \text{ for  $x \leq -1$} \\
e^{ \frac{K}{\sqrt{1-c^2}} x } & \text{ for $x \geq 1$}
\end{cases}
\end{equation}
We furthermore require the nonlinearity satisfies the following \textbf{concavity} condition:
\begin{equation}
    \label{eq: concavity}
    \mathcal N''(s) \leq 0,\quad \text{for $s \in [0,1]$}
\end{equation}
\begin{theorem}
\label{thm: kink stability}
Assume the nonlinearity satisfies the concavity condition \eqref{eq: concavity}.
Then, 
\begin{enumerate}
    \item Any kink $b_*= b_{c,0}$  (speed $c\in[0,1)$) is spectrally stable in $L^2_{w_c}$ given by \eqref{eq: W c weight}.
    That is, the spectrum of weighted operator $L_{*,w_c}$  is a subset of the closed left-half complex plane. 
\item In particular, $0 \in \sigma_\mr{d} \big(L_{*,w}\big)$
and the corresponding zero energy eigenspace is spanned by $\partial_x b_{c,0}(x)$, 
arising from translation invariance of \eqref{eq: PDE in moving frame}.
\end{enumerate}
\end{theorem}
\begin{remark}
\label{remark: translation mode kink}
    The neutral mode by translation invariance which corresponds to the neutral eigenvalue $0$, is always in the spectrum, since $e^{\frac{Kx}{\sqrt{1-c^2}}} \pd_x b_*(x) \to 0$ as $x\to \infty$.  
    In fact, the rate $a = \frac{K}{\sqrt{1-c^2}}$ in {\sc Theorem} \ref{thm: kink stability} satisfies 
\[
a \leq \frac{2K}{\sqrt{1-c^2}}
\]
where $\frac{2K}{\sqrt{1-c^2}}$ 
comes from the decay rate of the kink as $x \to \infty$, namely \eqref{eq: kink translation mode decay}.
This is in contrast to the case of supersonic pulses,
see {\sc proposition} \ref{prop:no_disc-super} and the discussion which follows it.
\end{remark}
\begin{remark}
We can improve the result a bit by merely requiring that the parameter $a$ in $e^{w(x)} = e^{ax} $ for $x > 1$ to satisfy
\[
    K \sqrt{\frac{1-c}{1+c}} \leq a \leq K \sqrt{\frac{1+c}{1-c}}
\]
However for simplicity we made the restriction $a = \frac{K}{\sqrt{1-c^2}}$ above.
This restriction does not affect our understanding of the ``big picture''.
Moreover, with this weight, when $c=0$,
the results in {\sc Theorem} \ref{thm: kink spectral stability c = 0} are recovered.
\end{remark}

Recall for $c = 0$ we transformed eigenvalue problem \eqref{eq: ep0} to \eqref{eq: reduced GEP c = 0}, through transformation \eqref{eq: B tilde transform 0}. 
We  proceed in an analogous manner.
We begin with the eigenvalue problem:
\begin{equation}
\label{eq: ep c}
L_{*,w} \tilde { B} (x) 
=  \lambda \tilde { B} (x),\quad B\in L^2
\end{equation}
Explicitly, $L_{*,w}$ defined in \eqref{eq: conjugate operator} is the (weight-conjugated) linearized operator about a kink of speed $0<c<1$:
\begin{align*}
L_{*,w}= \Sigma \Big( \pd_x - w_c'(x) \Big) + A_*(x),
\end{align*}
where $\Sigma= c \sigma_0 + \sigma_1$ and 
\[A_*(x) = \begin{bmatrix}
    {\mathcal N'\big(r_*( x )^2 \big)} r_*(x)^2 \sin 2\theta
    & \mathcal N_* \big(r_*(x)^2 \big) +  {2 \mathcal N' \big(r_*( x )^2 \big)} r_*(x)^2 \sin^2 \theta \\
     -\mathcal N \big(r_*( x )^2 \big)- {2\mathcal N' \big(r_*( x )^2 \big)}r_*(x)^2 \cos^2 \theta & -{\mathcal N' \big(r_*( x )^2 \big)}r_*(x)^2 \sin 2\theta 
    \end{bmatrix}\]
The expression for $A_*(x)$ is the general expression \eqref{eq: A *} where set 
set $(u_*,v_*) = b_*$ satisfying \eqref{eq: kink profile ODE} with speed $c$;
the kink profile which tends to $(0,0)$ as $x\to-\infty$ and to $(\cos(\theta),\sin(\theta))$ as $x\to+\infty$. 
Here, $\theta = \theta_{c,E} = \frac{1}{2}\arcsin (E-c)$; see Section \ref{sec:theta}.

Eigenvalue problem \eqref{eq: ep c} is rather complicated,
in view of the expression of $A_*$ above. 
However, it is possible to reduce it to a simpler form:
\begin{proposition}
\label{prop: B tilde B transform}
Assume eigenvalue problem \eqref{eq: ep c} is satisfied by $\lambda \in \mathbbm C, \tilde B \in L^2$.
Let
 $B=\begin{bmatrix}
    U & V
\end{bmatrix}^\mathsf T$ satisfying 
\begin{equation}
\label{eq: transform B beta text}
    B (X) =
    e^{-c \Lambda X }
    e^{ w_c\big(\sqrt{1-c^2} X \big) + \int_{-\infty}^X \mathcal N_p(Y) \dd Y }  \beta (X)
\end{equation}
where the intermediate $\beta(X)$ is related to $\tilde B$ by
\begin{equation}
\label{eq: B tilde beta transformation text}
\tilde B \big(\sqrt{1-c^2} X \big) = \big(\sigma_0 \cos \theta + \sigma_1 \sin \theta\big)  \beta (X)
\end{equation} 
where $\sigma_j,\ j=0,1,2,3$ are Pauli matrices. See \eqref{eq: pauli}.
Then, 
\begin{equation}
\label{eq: UV}
    \begin{aligned}
     U_X & = \big(\mathcal N_0 + \mathcal N_p \big)  U + \big( \Lambda -2c\mathcal N_p\big) V \\
     V_X & = \Lambda U - \big(\mathcal N_0 + \mathcal N_p\big) V
    \end{aligned},
\end{equation}
where  $\mathcal{N}_0 = \mathcal N\big(r_0(x)^2\big)$ and $\mathcal{N}_p = \mathcal N'(r_0(x)^2)r_0(x)^2$, defined in \eqref{eq: N 0 and N p}, where $r_0(x)$ is the amplitude of the non-moving kink profile satisfying \eqref{eq: kink amplitude ODE 0}; and $\Lambda = \frac{\lambda}{\sqrt{1-c^2}}$.
\end{proposition}
The proof of {\sc proposition} \ref{prop: B tilde B transform} is given in {\sc appendix} \ref{app: proof of prop: B tilde B transform}. 
Note that setting $c = 0$, 
transformation \eqref{eq: B tilde transform 0} is recovered from \eqref{eq: B tilde beta transformation text} and \eqref{eq: transform B beta text}. 
\begin{remark}
\label{rem: intuition}
    We remark on the intuition behind the transform \eqref{eq: B tilde beta transformation} above.
Consider the ODE \eqref{eq: kink profile ODE} which the profiles of kinks and antikinks satisfy 
(without the restriction on the second line of \eqref{eq: kink profile ODE}).
For $c = 0$, the orbits of possible kinks lie on the $u$-axis, and those of possible antikinks lie on the $v$-axis.
In general, however, the kink $b_{c,0}$ and $\mathcal P b_{c,0}$ lie on the line parallel to vector $\begin{bmatrix}
    \cos \theta & \sin\theta
\end{bmatrix}$, 
and the antikinks $\mathcal {CP} b_{c,0}$ and $\mathcal C b_{c,0}$ lie on the line parallel to the vector
$\begin{bmatrix}
    \sin \theta & \cos\theta
\end{bmatrix}$, see {\sc figure} \ref{fig: subsonic kink c > 0};
note that for $c > 0$, $\theta = -\frac12 \arcsin c <0$.
We would like to express the perturbation in terms of coordinates along these two vectors, hence the transform \eqref{eq: B tilde beta transformation}.
\end{remark}

The following result states that to preclude \textit{unstable} discrete spectrum, it suffice to prove that there are no eigenpairs $(B,\Lambda)$ of 
\eqref{eq: UV}. 
\begin{lemma}
\label{lem: c neq 0 kink unbounded}
Assume \eqref{eq: ep c} is satisfied by $\Re \lambda  > 0$ and $\tilde B\in L^2$.
Then $B$ defined through \eqref{eq: transform B beta text} satisfying \eqref{eq: UV}
with $\Lambda = \frac{\lambda}{\sqrt{1-c^2}}$ is in $L^2$ as well. 
\end{lemma}
The proof of {\sc lemma} \ref{lem: c neq 0 kink unbounded} is given in {\sc appendix} \ref{app: proof of lem: c neq 0 kink unbounded}.


We now show there is no $B\in L^2$ that satisfies \eqref{eq: UV} with some $\Re \Lambda > 0$. We proceed with shorthand $F(X) = \mathcal N_0(X) + \mathcal N_p(X)$ and $f = 2c \mathcal N_p(X)$. 
Acting $\pd_X$ on both sides of the second equation of \eqref{eq: UV}, then use the first equation for $U_X$, we obtain a closed equation for $V$
\begin{equation}
    \label{eq: gep V}
    V_{XX}
    =
\big( F^2 - F_X \big) V + \Lambda\big(\Lambda - f \big) V
\end{equation}

\begin{lemma}
\label{lemma: no V L 2} Assume $\Re \Lambda >0$, 
and that the nonlinearity $\mathcal N(r^2)$ satisfies the concavity condition
\eqref{eq: concavity}:
\[
\mathcal N''(s) \leq 0
\]
For $s \in [0,1]$, along with hypotheses ($\mathcal N$1) to ($\mathcal N$4) in {\sc section} \ref{sec: nonlinearity}.
Then, the only  $V \in L^2$ that satisfies \eqref{eq: gep V} is $V\equiv0$.
\end{lemma}
\begin{proof}
Suppose $0\ne V\in L^2(\mathbb{R})$. Then, $V$ is also smooth. Take $L^2$-inner product of $V$ with \eqref{eq: gep V} and obtain
\begin{equation}
    \label{eq: Lambda quadratic}
    \big\|V\|_{L^2}^2 \ \Lambda^2 + \big\langle V, - fV \big\rangle_{L^2} \Lambda + 
    \big\langle V, \big(F^2 - F_X\big) V \big\rangle_{L^2}
    +\big\| V_X \big\|^2_{L^2} = 0 ,
\end{equation}
a quadratic equation in $\Lambda$. We claim that the coefficients of this quadratic are all non-negative.
Note that the quadratic coefficient of \eqref{eq: Lambda quadratic}, namely $\big\|V \big\|^2_{L^2} > 0$, and 
\[
- f = - 2c \mathcal N_p = - 2 c \mathcal N' \big(r_0(X)^2 \big) r_0(X)^2 \geq 0
\]
since $\mathcal N$ is monotonically decreasing;
so the linear coefficient $\big\langle V, -f V \big\rangle_{L^2} \geq 0$ as well. 
Furthermore, $F_X \leq 0$. 
In fact, since $r_0(X)$ is monotonically increasing in $X$, and
\[
    F  = \mathcal N_0 + \mathcal N_p = \mathcal N\big(r_0(X)^2\big) + \mathcal N'\big(r_0(X)^2\big) r_0(X)^2
\]
is monotonically \textit{decreasing} in $r_0(X)^2$.
$\mathcal N$ is so by definition, and as a result of the concavity condition \eqref{eq: concavity}, there is also $\mathcal N' \leq 0$ monotonically decreases in $r_0(X)^2$.
Therefore $F_X \leq 0$ and $F^2-F_X \geq 0$ whence $\big\langle V, \big(F^2 - F_X\big) V \big\rangle_{L^2} \geq 0$.
We conclude that $\Lambda$ is a root of a quadratic polynomial:
\[
c_2 \Lambda^2 + c_1 \Lambda + c_0 =0
\]
where $c_2>0$, $c_1, c_0\geq0$.
It is easily checked that either $\Lambda$ is one of two negative real roots or its real part is equal to $-c_1/2c_0 \leq 0$. In all cases, $\Re\Lambda\le0$, contradicting the assumption that $\Re\Lambda>0$.
As a result there is no $\Lambda$ with nonnegative real part such that \eqref{eq: Lambda quadratic} holds. 
This argument carries over for any $V \in L^2$ and we will be done.
\end{proof}

Now we summarize the discussion above and give the proof of {\sc Theorem} \ref{thm: kink stability}:
\begin{proof}[Proof of {\sc Theorem} \ref{thm: kink stability}]
The essential spectrum of $L_{*,w_c}$ is in the closed left-half plane since $a = \frac{K}{\sqrt{1-c^2}}$,
satisfying the condition \eqref{eq: essential kink condition} 
in {\sc proposition} \ref{prop: essential kink}. 

\textit{Assume} there was an eigenpair $\lambda, \tilde B$ of $L_{*,w_c}$,
where $\Re \lambda >0$ and $\tilde B\in L^2$.
Then by {\sc lemma} \ref{lem: c neq 0 kink unbounded}, 
there would be a pair $\Lambda = \frac{\lambda}{\sqrt{1-c^2}}, B = \begin{bmatrix}
    U & V
\end{bmatrix}^\mathsf T \in L^2$ 
which would solve the generalized eigenvalue problem \eqref{eq: UV}.
In particular we would have
$\Re \Lambda > 0$ and $V \in L^2$.
Such $V$ would in fact be smooth since the coefficients
in \eqref{eq: UV} are bounded and smooth.
As a result the pair $\Lambda, V \in L^2$ 
would solve a generalized eigenvalue problem expressed by a second-order variant coefficient ODE,
namely \eqref{eq: gep V}.
However such $V$ does not exist by {\sc lemma} \ref{lemma: no V L 2}. 

Therefore \textit{the assumption does not hold}; namely $L_{*,w}$ does not have any eigenvalue with a positive real part.

Thus $\sigma\big(L_{*,w_c}\big)$ is contained in the closed left-half plane, and by {\sc definition}
\ref{def: spectral stability}, moving kinks with nonlinearity concave on $r^2 \in [0,1]$ are spectrally stable 
in the weighted space $L^2_{w_c}$ where $w_c(x)$ is given by \eqref{eq: W c weight}.
\end{proof}

\section{Conclusions}
There is a large family of heteroclinic traveling wave solutions of \eqref{eq: PDE in lab frame}, in both supersonic ($|c| > 1$) and subsonic ($|c| < 1$) regimes; {\sc section} \ref{sec: tws}.
The bounded  heteroclinic traveling waves are either heteroclinic connections  between the zero solution and an equilibrium on the  unit circle (kinks, antikinks) or  heteroclinic connections between distinct equilbria on the unit circle (pulses). 
For saturable nonlinearities, supersonic pulses are nonlinearly convectively stable against perturbations, which decay rapidly as $x\to+\infty$. This follows 
from an a priori bound on the growth of general solutions of the IVP and finite propagation speed 
property for the semilinear hyperbolic system \eqref{eq: PDE in lab frame} ({\sc Theorem} \ref{thm: convective}). For general nonlinearities,  supersonic pulses are
linear spectrally stable ({\sc Theorem} \ref{thm: spectral stability of supersonic pulses}) in an appropriate weighted $L^2$ space; the spectra of their linearized operators lie in the left half plane.
Kinks are linearly and neutrally spectrally stable in suitably weighted $L^2$ spaces: (a) for the case of non-moving kinks $(c=0)$ ({\sc Theorems} \ref{thm: kink spectral stability c = 0}) and (b)
for arbitrary kinks $(|c|<1)$ under a concavity condition on the nonlinearity ({\sc Theorem} \ref{thm: kink stability}); {\sc section} \ref{sec: kink stability};   the spectra of their linearized operators lie on the imaginary axis. 
Antikinks and subsonic pulses are exponentially unstable; {\sc section} \ref{sec: instability}.
Open questions and future directions are discussed in {\sc section} \ref{sec:openq}.

\appendix
\section{A complete list of the bounded traveling wave solutions}
\label{app: classification}
The following list of \textit{bounded} TWS types is exhaustive. 
\begin{enumerate}
    \item Equilibria. 
    These corresponds to spatially constant time-independent solutions of
    \eqref{eq: PDE in lab frame}:
    \[
        \big\{[0,0]\big\} \cup \Big\{ \big[\cos \theta, \sin \theta\big]:\ \theta\in (-\pi, \pi] \Big\}
    \]
    which are TWSs of arbitrary speeds 
    $c \in\mathbbm R$.
    \item Kinks and antikinks. 
    For each $ 0 \leq c <1$, there are four kink-like solutions; 
    in particular, they are $b_{c,0}$, $\mathcal P b_{c,0}$ which are kinks, 
    and $\mathcal {CP}  b_{c,0}, \mathcal {C}  b_{c,0}$ which are antikinks; for each $-1<c \leq 0$, there are four kink-like solutions. In particular, they are $\mathcal {CPT} b_{|c|,0} = \mathcal {TC}  b_{|c|,0}$, $\mathcal {CT} b_{|c|,0}$ which are kinks, and $\mathcal T b_{|c|, 0}$, $\mathcal {PT}  b_{|c|,0}$ which are antikinks. 
    See {\sc figure} \ref{fig: kinks}.
    
    {\bf TWSs of all the following types are invariant under $\mathcal C$. }
    \item Subsonic pulses.
    In each of the four following cases there are two subsonic pulse solutions. 
    For each $0 \leq c < 1$ and each $0<E<1+c$ they are $ b_{c,E}$ and $\mathcal P  b_{c,E}$;
    for each $0 \leq c < 1$ and each $-1+c <E<0$ they are $\mathcal T  b_{-c,|E|}$ and $\mathcal{PT}  b_{-c,|E|}$; 
    for each $-1<c<0$ and each $0<E<1+c$ they are $\mathcal T  b_{|c|, E}$ and $\mathcal {PT}  b_{|c|, E}$;
    for each $-1<c<0$ and each $-1+c<E<0$ they are $ b_{c, |E|}$ and $ \mathcal P  b_{c, |E|}$.
    
\item Supersonic pulses.
    For each of the following cases there are four supersonic pulse solutions, 
    except for the marginal cases which will be indicated. 
    For each $c>1$ and each $0<-1+c\leq E\leq 1+c$, these are $ b_{c,E}$, $\mathcal P b_{c,E}$ which respectively degenerates to 
    $\pm [1,1]/\sqrt{2}$ at $E = 1+c$, 
    and $\mathcal T  b_{-c,-E}$ and $\mathcal {PT}  b_{-c,-E}$ which respectively degenerates to $\pm[-1,1]/\sqrt{2}$ at $E=-1+c$. 
    For each $c<-1$ and each $-1+c \leq E \leq 1+c <0$, 
    they are $\mathcal T b_{|c|,|E|}$ and $\mathcal {PT} b_{|c|,|E|}$ which respectively degenerates to $\pm[-1,1]/\sqrt{2}$ at $E=-1-c$, 
    and $ b_{c,E}$ and $\mathcal P  b_{c,E}$ which respectively degenerates to $\pm[1,1]/\sqrt{2}$ at $E = -1+c$.
    See {\sc figure} \ref{fig: supersonic}.
%
    \item Periodic wave trains. 
    For each of the following cases there is one periodic wave train solution. In particular, 
    for each $c > 1$ and $0< E<1-c$, it is counterclockwise and ``small''; 
    for each $c>1$ and each $E>1+c$, it is clockwise and ``large''; 
    for each $c< -1$ and $-1+c<E<0$ it is clockwise and ``small''; 
    for each $c<-1$ and each $E<-1-c$ it is counterclockwise and ``large''.
\end{enumerate}

\section{Proof of {\sc proposition} \ref{prop: nonlinearity}}
\label{app: nonlinearity}
Recall 
\begin{equation}
\label{eq: N * formula}
N_*\big (x; B \big) = \begin{bmatrix}
    \mathcal N\big(r_*^2\big) V+
    \delta \mathcal N_*\big (x; B \big) v_* \\
    - \mathcal N\big(r_*^2\big) U -\delta \mathcal N_*\big(x;  B \big) u_*
    \end{bmatrix}=\begin{bmatrix}
    \mathtt{I}_1+\mathtt{II}_1  \\
    \mathtt{I}_2 +\mathtt{II}_2
    \end{bmatrix}
\end{equation}
where $\mathcal N \big(r^2 \big)$ is bounded, real-valued and sufficiently smooth, 
satisfying $\mathcal N(0)=1$, $\mathcal N(1) = 0$ and $r^2=1$ is the only zero of $\mathcal N\big(r^2 \big)$. 
Moreover, there is $\mathcal N_\infty <0$ such that $\big|\mathcal N\big(r^2\big) - \mathcal N_\infty \big| = \mathcal O\big(r^{-\alpha})$
as $r \to \infty$, and for each $k = 1,2,\cdots$ there is $C >0 $ such that
$\big| \mathcal N^{(k)}\big(r^2 \big)\big| = \mathcal O\big(r^{-\alpha-k}\big)
$
as $r \to \infty$.
Moreover, $b_*$, the TWS we work with, as well as its components and all their derivatives with respect to $x$, 
asymptotes to their respective limits exponentially as $x \to \pm\infty$, and are all bounded, in particular we only need to work out the estimate details for terms $\mathtt{I}_1$ and $\mathtt{II}_1$ in \eqref{eq: N * formula}, since those for $\mathtt{I}_2$ and $\mathtt{II}_2$ are almost identical.\\

\subsection{Self-mapping properties}
\subsubsection{
Estimates for $\mathtt I_1$ on $ H^1$
}
If $V \in H^1$ we have $
    \Big\| N\big( r_*^2 \big) V \Big\|_{L^2}  \leq C\big\| V \big\|_{L^2}$
, also ($V$ being weakly differentiable), for each $x$:
\[
\Bigg| \Big[ \mathcal N \big( r_*(x)^2 \big) V(x) \Big]_x \Bigg| \leq   \Big| \mathcal N(r_*^2) V' + 2 \mathcal N'(r_*^2)r_* r_*'V \big| 
\leq C \Big[ \big| V'(x) \big| + \big| V(x) \big| \Big]
\]
Thus 
\[
\begin{aligned}
    N\big(r_*(x)^2\big) V(x) & \in H^1\big(\mathbbm R_x \big)\\
    \Big\| N\big(r_*(x)^2\big) V(x) \Big\|_{H^1\big(\mathbbm R_x \big)} &\leq C \big\|V\|_{H^1} 
\end{aligned}
\]
Where $C$ only depends on $b_*$ and the profile of $\mathcal N_*$.

\subsubsection{Estimates for $\mathtt {II}_1$ in $L^2$}
The following calculus lemma will be repeatedly used in the sequel. 
\begin{lemma}
\label{lem: estimate F}
For $\ell, k$ positive integers, let
     $F\big( x,  y \big) \in \mathcal C^1 \Big( \mathbbm R_x \times \mathbbm R^\ell_{y}, \mathbbm R^k \Big).
$ 
Assume 
\begin{enumerate}[label=(\roman*)]
    \item $ F(x, 0)= 0$ for all $x\in \mathbbm R$, and $ F$ is bounded.
    \item There is a nonnegative continuous function $f\geq 0$ on $[0,\infty)$ such that
\[
    \left\| \frac{\pd  F\big(x,  y \big) }{\pd y} \right\|_{M^{\ell\times k}\big( \mathbbm R \big) } \leq f \big( |y| \big)
\]
uniformly for all $x \in \mathbbm R$.
\end{enumerate}

Then, there is a constant $C>0$ such that, uniformly in $x\in\mathbbm R$, we have
\[
    \Big|  F\big(x,  \eta \big) \Big| \leq C \big|  \eta \big|.
\]
\end{lemma}
\begin{proof}
[Proof of {\sc lemma} \ref{lem: estimate F}]
For any $M > 0$ and for all $|\eta| \geq M$,
\[
    \Big|  F\big(x,  \eta \big) \Big| \leq \big\| F\big\|_{L^\infty}  \leq \frac{\big\| F\big\|_{L^\infty} }{M} \big|  \eta\big|.
\]
For $| \eta| \leq M$, 
\[
    F\big(x,  \eta \big)  = \int_0^1 \pd_t  F(x,t\eta) dt= \int_0^1  \frac{\pd  F\big(x, t\eta \big) }{\pd y} dt\ \   \eta.
    \]
Hence,
\[ |F(x,\eta)| \le  \max_{0\le r\le M} f(r)\ |\eta|.\]
Therefore, uniformly in $x$ and for all $\eta \in \mathbbm R^\ell$, we have
\[
    \Big|  F\big(x,  \eta \big) \Big|
    \leq
    \max \bigg\{ \frac{\big\| F\big\|_{L^\infty} }{M},
\max_{r \in [0,M]} f(r) \bigg\} \big|  \eta \big|.
\]
\end{proof}

Recall
\[
\delta \mathcal N_*(x;U,V)  =\mathcal N\Big( \big( u_*(x) + U \big)^2 + \big(v_*(x) + V\big)^2 \Big) - \mathcal N\big( r_*(x)^2\big) 
\]
Now we apply {\sc lemma} \ref{lem: estimate F} to $  B \mapsto \delta \mathcal N_*\big( x; B\big ) v_*(x) $ which is continuously differentiable, bounded, and
vanishes for $ B = 0$. Moreover,
\[
    \begin{aligned}
    & \left| \frac{\pd \delta \mathcal N_*\big( x; B\big ) v_*(x)}
    {\pd  B}\right| = \left| 2 v_* \mathcal N' \Big( \big( u_* + U \big)^2 + \big(v_* + V\big)^2 \Big)
    \begin{bmatrix}
    u_*+ U \\ v_*+ V 
    \end{bmatrix}\right| \\
    \leq & C \sqrt{ \Big[ u_*(x) + U \Big]^2 + \Big[ v_*(x)+ V \Big]^2 } \leq  C \sqrt{ \big\|  b_* \big\|_{L^\infty} ^2 + 2 \big\|  b_* \big\|_{L^\infty} \big|  B\big| + \big| B\big|^2} \\
    \leq & C \sqrt{ 1 + \big|  B\big| ^2 }.
    \end{aligned}
\]
is uniformly bounded by an increasing and continuous function of $| B|$. 
Therefore {\sc lemma} \ref{lem: estimate F} can be applied to $ F \big(x,  B \big) = \delta \mathcal N_*\big (x; B\big) v_*(x)$, thus pointwise
\[ 
\big| \delta \mathcal N_* \big( x; B \big) v_*(x) \big| \leq C \big|  B(x) \big| \]
For $ B \in  L^2$:
\[
\begin{aligned}
    \delta \mathcal N_* \big(x; B \big) v_*(x) & \in L^2 \\
    \big\| \delta \mathcal N_* \big(x; B \big) v_*(x) \big\|_{L^2} & \leq C \big\|  B \big\|_{L^2}
\end{aligned}
\]
for $C$ only depending on $ b_*$ and $\mathcal N$ and we have:
\[
    \Big\|  N_* \big(x ;  B \big) \Big\|_{L^2 } \leq C \big\|  B \big\|_{L^2}
\]

\subsubsection{Estimates for $\mathtt {II}_1$ on $H^1$}
We only need to prove the derivative in $x$ of $\delta \mathcal N_* \big(x; B \big) v_*(x)$ namely the following function is in $L^2$:
\[\Big[ \delta \mathcal N_*\big (x; B(x) \big) v_*(x)\Big]_x = \left( \delta \mathcal N_*\big(x;  B(x) \big)\right)_x v_*(x) + \delta \mathcal N_*\big( x;  B(x) \big) v_*'(x) 
\]
Apply {\sc lemma} \ref{lem: estimate F} to $\delta \mathcal N_*(x; B) v_*'(x)$ to obtain
\[ 
\big| \delta \mathcal N_*(x;U,V) v_*'(x) \big| \leq C \big|  B(x) \big| \] thus
\[
\begin{aligned}
    \delta \mathcal N_*(x;U,V) v_*'(x) & \in L^2 \\
    \big\| \delta \mathcal N_*(x;U,V) v_*'(x) \big\|_{L^2 } & \leq C \big\|  B \big\|_{L^2}.
\end{aligned}
\]
For a fixed $x\in\mathbbm R$, 
{
\[
    \begin{aligned}
    & \left(\frac{\dd}{\dd x} \delta \mathcal N_*(x; B)\right) v_*(x) \\
    = &2 \bigg [ \mathcal N'\Big( \big(u_*+U\big)^2 + \big(v_* + V \big)^2 \Big) \Big[  \big(u_* + U \big) \big( u_*' + U'  \big) 
    + \big(v_* + V\big) \big( v_*' + V' \big) \Big] \\
    & - \mathcal N'\big(u_*^2+v_*^2\big) \big(u_* u_*' + v_* v_*' \big) \bigg]
    v_*\\
    = & 2 \bigg[ 
    \mathcal N' \Big( \big|  b_*(x) +  B(x) \big|^2 \Big) \big(  b_*(x) +  B(x) \big) \cdot \Big(  b_*'(x) +  B' (x) \Big) - \mathcal N' \Big ( \big| b_* (x) \big|^2 \Big)  b_*(x) \cdot  b_*'(x) 
    \bigg] v_*(x)
    \end{aligned}
\]
}
We only need to estimate the expression inside $\big[ \cdots \big]$, since $|v_*|$ is bounded:
{
\[
    \begin{aligned}
    &  \bigg| 
    \mathcal N' \Big( \big|  b_*(x) +  B(x) \big|^2 \Big) \big(  b_*(x) +  B(x) \big) \cdot \Big(  b_*'(x) +  B' (x) \Big) - \mathcal N' \Big ( \big| b_* (x) \big|^2 \Big)  b_*(x) \cdot  b_*'(x) 
    \bigg| \\
    \leq &\bigg| 
    \mathcal N' \Big( \big|  b_*(x) +  B(x) \big|^2 \Big) \big(  b_*(x) +  B(x) \big) \cdot  B'(x) \bigg| 
    \\
    &+\bigg| \Big[ \mathcal N' \Big( \big|  b_*(x) +  B(x) \big|^2 \Big) \big(  b_*(x) +  B(x) \big) - \mathcal N' \Big ( \big| b_* (x) \big|^2 \Big)  b_*(x) \Big] \cdot  b_*'(x) 
    \bigg|\\
    \leq & C \big| B'(x)\big| + C \big|  B(x) \big|
    \end{aligned} 
\]
} the first term can be bounded because $\mathcal N'(r^2)$ decays fast enough as $r \to \infty$ 
so $\mathcal N' \big(r^2 \big) r$ is bounded since it is also continuous. 
Then we apply to the second term {\sc lemma} \ref{lem: estimate F} and note that $ b_*'(x)$ is bounded. 
Therefore
\begin{equation*}
\begin{aligned}
    \Bigg| \left(\frac{\dd}{\dd x} \delta \mathcal N_*(x; B)\right) v_*(x) \Bigg| & \leq C\Big( \big| B(x) \big| + \big| B'(x) \big| \Big)\\
    \Bigg\| \left(\frac{\dd}{\dd x} \delta \mathcal N_*(x; B)\right) v_*(x) \Bigg\|_{L^2 \big( \mathbbm R_x \big) } 
    & \leq C \big\|  B\big\|_{H^1} 
\end{aligned}
\end{equation*}
Thus we have:
\[
    \Big\| \mathcal N_* \big(x ;  B(x) \big) \Big\|_{H^1 \big( \mathbbm R_x\big) } \leq C \big\|  B \big\|_{H^1}
\]
for a constant $C>0$ depending only on $ b_*$ and $\mathcal N$.

\subsection{Lipschitz properties on $L^2$}

\subsubsection{$\mathtt I_1$ is trivially globally Lipschitz}
$\mathtt I_1 = \mathcal N\big(r_*(x)^2 \big)V(x) $ is trivially globally Lipschitz in $ B(x)$ since it is linear in $V$.

\subsubsection{Lipschitz property of $\mathtt {II}_1$}
Let $\tilde B = \big[ \tilde U, \tilde V \big] \in \mathcal X = L^2$.
Note that \[
    \Big| \nabla \mathcal N \big (| B|^2 \big) \Big| = 2 \Big| \mathcal N' \big( | B|^2 \big) \Big| \big|  B \big| 
\]
is bounded:
\[
    \begin{aligned}
    & \Big| \delta \mathcal N_*\big(x; B(x) \big) v_*(x) - \delta \mathcal N_*\big(x; \tilde { B}(x) \big) v_*(x) \Big| \\
    \leq & \bigg| \mathcal N\Big( \big|  b_*(x) +  B(x) \big|^2 \Big) - \mathcal N\Big( \big|  b_*(x) + \tilde{ B}(x) \big|^2 \Big) \bigg| \big\| v_* \big\|_{L^\infty} \\
    \leq &
     \big|  B(x) - \tilde{ B}(x) \big| 
     \sup_{s \in [0,1] } \Big | \nabla \mathcal N \Big (| b_*(x) + (1-s)  B (x) + s \tilde{ B} (x) |^2 \Big) \Big| \\
     \leq & C\big|  B(x) - \tilde{ B}(x) \big| 
    \end{aligned}
\] 
Integrating both sides we have
\[
     \left \| \delta \mathcal N_*\big(x; B(x) \big) v_*(x) - \delta \mathcal N_*\big(x;\tilde{ B} (x) \big) v_*(x)\right\|_{L^2 \big( \mathbbm R_x \big) } \leq  C \big\|  B - \tilde{ B} \big\|_{L^2}
\]

Doing the same estimate for $\mathtt{I}_2$ and for $\mathtt {II}_2$ and we have obtained the global Lipschitz property:
\[
\Big\|  N_* \big( x ;  B(x) \big) -  N_* \big( x ; \tilde{ B} (x) \big) \Big\|_{L^2 \big(\mathbbm R_x \big) } \leq C \big\|  B - \tilde{ B} \big\|_{L^2}\]

\subsection{Lipschitz properties on $H^1$}
\subsubsection{$\mathtt I_1$ is trivially globally Lipschitz on $H^1$} $\mathtt I_1 = \mathcal N\big(r_*(x)^2 \big)V(x) $ is trivially globally Lipschitz since it is linear in $V$ for $ B\in H^1$ and $\mathcal N\big(r_*(x)^2\big)$ is smooth in $x$.
\subsubsection{Lipschitz property of $\mathtt{II}_1$}
Now we estimate 
\[
    \begin{aligned}
    &  \bigg \| \bigg[ \Big( \delta \mathcal N_*\big(x; B(x) \big) - \delta \mathcal N_*\big(x;\tilde { B}(x) \big) \Big) v_*(x)\bigg]_x \bigg\|_{L^2(\mathbbm R_x)} 
    \\
    \leq & \left \| \left( \delta \mathcal N_*\big(x; B(x) \big) - \delta \mathcal N_*\big(x;\tilde { B}(x) \big)\right)_x v_*(x)\right\|_{L^2(\mathbbm R_x)} \\
    & + \left \|\left( \delta \mathcal N_*\big(x; B(x) \big) - \delta \mathcal N_*\big(x;\tilde { B}(x) \big)\right) v_*'(x)\right\|_{L^2(\mathbbm R_x)}
    \end{aligned}
\]
the second term can be estimated similarly as in is the $L^2$ Lipschitz estimate above:
\[
    \left \|\left( \delta \mathcal N_*\big(x; B(x) \big) - \delta \mathcal N_*\big(x;\tilde { B}(x) \big) \right) v_*'(x)\right\|_{L^2(\mathbbm R_x)} \leq C \big\| B-\tilde { B}\big\|_{L^2}
\]
To estimate the first term, note that {\footnotesize
\begin{equation}
    \begin{aligned}
    & \Bigg| \frac12 \left( \delta \mathcal N_*\big(x; B(x) \big) - \delta \mathcal N_*\big(x;\tilde{ B}(x) \big)\right)_x \Bigg|
    = \Bigg| \frac12 \left[\mathcal N\Big( \big|  b_*(x) +  B(x) \big|^2 \Big) - \mathcal N\Big( \big|  b_*(x) + \tilde{ B}(x) \big|^2 \Big)\right]_x \Bigg| \\
    = & \Bigg|
    \mathcal N' \Big( \big|  b_*(x) +  B(x) \big|^2 \Big) \big(  b_*(x) +  B(x) \big) \cdot \Big(  b_*'(x) +  B' (x) \Big) - \mathcal N' \Big( \big|  b_*(x) + \tilde{ B}(x) \big|^2 \Big) \big(  b_*(x) + \tilde{ B} (x) \big) \cdot \Big(  b_*'(x) + \tilde{ B}' (x) \Big) \Bigg|
    \\ 
    \leq & \Bigg| \mathcal N' \Big( \big|  b_*(x) +  B(x) \big|^2 \Big) \big(  b_*(x) +  B(x) \big) \Bigg|  \Big|  B' (x) - \tilde{ B}'(x) \Big|
    \\
     & \Bigg| \mathcal N' \Big( \big|  b_*(x) +  B(x) \big|^2 \Big) \big(  b_*(x) +  B(x) \big) - \mathcal N' \Big( \big|  b_*(x) + \tilde{ B}(x) \big|^2 \Big) \big(  b_*(x) + \tilde{ B} (x) \big) \Bigg| \Big|  b_*'(x) + \tilde{ B}' (x) \Big|\\
    =: & \mathtt{III} + \mathtt{IV}
    \end{aligned}
\end{equation}
}
Term $\mathtt{III}$ is bounded pointwise by 
\begin{equation*}
    C \Big|  B' (x) - \tilde{ B}'(x) \Big|
\end{equation*}
since $\mathcal N'\big(r^2 \big) r$ is bounded.
For term $\mathtt {IV}$, set $\eta = B - \tilde B$ and define
{\footnotesize
\[
     F \big(x ,  \eta) := \mathcal N' \Big( \big|  b_*(x) + \tilde{ B}(x) +  \eta \big|^2 \Big) \big(  b_*(x) + \tilde{ B}(x) +  \eta \big) - \mathcal N' \Big( \big|  b_*(x) + \tilde{ B}(x) \big|^2 \Big) \big(  b_*(x) + \tilde{ B} (x) \big)
\]
} Then $ F$ is a bounded function of $\eta$, and continuously differentiable, vanishes for $ \eta =  0$. Consider the gradient of $F(x, \eta)$ with respect to $ \eta$ for a fixed $x\in \mathbbm R$ (suppressing explicit dependence on $x$): 
\[
    \begin{aligned}
    \frac{\pd  F}{\pd  \eta} =&  2 \mathcal N'' \Big ( \big|  b_* + \tilde{ B} +  \eta \big| \Big) \big ( b_* + \tilde{ B} +  \eta \big) \big ( b_* + \tilde{ B} +  \eta \big) ^\mathsf T + \\
    + & \mathcal N' \Big( \big ( b_* + \tilde{ B} +  \eta \big)  \Big)  \sigma_0
    \end{aligned}
\]
Both of the terms are bounded uniformly for all $ \eta$. Therefore applying {\sc lemma} \ref{lem: estimate F} on $ F\big(x ,  \eta \big) $, and note that
\[
\mathtt{IV} = \Big|  F \big( x ,  B(x) - \tilde{ B}(x) \big) \Big| \Big|  b_*'(x) + \tilde{ B}'(x) \Big| 
\]
there is
\[
\begin{aligned}
    \mathtt{IV} \leq & C \big |  B(x) - \tilde{ B}(x) \big|
\Big|  b_*'(x) + \tilde{ B}'(x) \Big|  \\
\leq & C \Big ( \big |  B(x) - \tilde{ B}(x) \big|
\big|  b_*'(x) \big| + \big |  B(x) - \tilde{ B}(x) \big|
\big| \tilde { B}' (x) \big| \Big)
\end{aligned}
\]
Where $C$ here is a constant only depending on the profile of $\mathcal N$.

Thus integrating $\mathtt{III}$ and $\mathtt{IV}$, 
\[
    \begin{aligned}
    & \bigg\| \Big( \delta \mathcal N_*\big(x; B(x) \big) - \delta \mathcal N_*\big(x;\tilde { B}(x) \big) \Big) v_*(x)\bigg\|_{H^1} \\
    \leq & C \Big\|  B  - \tilde{ B} \Big\|_{L^2} + C  \Big\|  B'  - \tilde{ B}' \Big\|_{L^2} + C \big\|\tilde{ B}' \big\|_{L^2} \Big\|  B  - \tilde{ B}\Big\|_{L^\infty} \\
    \leq & C \Big( 1 + \big\| \tilde { B} \big\|_{H^1} \Big) 
    \Big\|  B  - \tilde{ B} \Big\|_{H^1}
    \end{aligned}
\]
The last inequality is due to Sobolev embedding.
\section{Proofs of some technical results}
\subsection{Proof of {\sc lemma} \ref{lemma: real part bounded}}
\label{app: proof of lemma: real part bounded}
If $\alpha + \beta = 0$, $\alpha \beta \leq 0$ and $f(k) = \sqrt{\alpha \beta - k^2}$ and is purely imaginary.
Therefore $\Re f(k) = 0$ and we are done.
Moreover, $f(0) = \sqrt{\alpha\beta}$, $\Re f(0) = 0$ for $\alpha \beta \leq 0$ or 
$|\alpha\beta| \leq \Big|\frac{\alpha + \beta}{2}\Big|$ for $\alpha \beta > 0$.
So we can study the case when $\alpha + \beta \neq0$ and $k \neq 0$.

Consider the \textit{conformal} mapping $g(z) = z \mapsto z^2$ from half plane $\Re z > 0$ to $\mathbbm C \setminus \mathbbm R_{\leq0}$.
Now the line $\Re z = \Big| \frac{\alpha + \beta}{2} \Big|$ cuts $\Re z >0$ into two path-connected components:
\[
\Omega_1 = \Big\{ 0 < \Re z \leq \Big| \frac{\alpha + \beta} 2 \Big| \Big\},\quad
\Omega_2 = \Big\{\Re z \geq \Big| \frac{\alpha + \beta} 2 \Big| \Big\}
\]
and $\Omega_1 \cap \Omega_2 = \Big\{ \Re z = \Big|\frac{\alpha+\beta}{2} \Big| \Big\}$.
Under $g(z)$, $\Re z = \Big| \frac{\alpha + \beta} 2 \Big|$ transforms to the parabola $\Pi$ given by 
\[
    x = p(y) = -\frac{y^2}{(\alpha+\beta)^2} + \Big(\frac{\alpha+\beta}{2}\Big)^2
\]
where $x,y = \Re (z^2), \Im (z^2)$. The set $\Omega_1$ is transformed into $g\big(\Omega_1\big)$ by $g(z)$, 
which is the ``left'' path-connected component of the cut complex plane $\mathbbm C \setminus \mathbbm R_{\leq 0}$ of which $0$ is an element. 
Now the image of $f(k)$ is a subset of the open right-half plane since $h(z) = z \mapsto \sqrt{z}$ on $\mathbbm C \setminus \mathbbm R_{\leq 0}$ is the inverse of $g(z)$. 
$\Re f(k) \leq \Big| \frac{\alpha + \beta} 2 \Big|$ is equivalent to that the image of $f(k)$, $k \neq0$, is in $\Omega_1$; 
this is further equivalent to the image of $f(k)^2$, $k \neq 0$, is in $g\big(\Omega_1\big)$. 
In fact,
\[
    f(k)^2 = \alpha\beta - k^2 - \ii (\alpha+\beta) k + \alpha\beta
\]
whose image is
\[
    x = q(y) = -\frac{y^2}{(\alpha+\beta)^2} + \alpha\beta, \quad (y \neq 0)
\]
and since $q(y) \leq p (y)$, the image of $f(k)^2$ sits to the left (inclusive) of $\Pi$, 
namely $\mathrm{ran} f(k)^2 \in g (\omega_1)$ and equivalently $\mathrm {ran} f(k) \in \Omega_1$, and this is further equivalent to 
\[
\Re f(k) \leq \Big| \frac{\alpha + \beta}{2} \Big|
\]
Now we prove that $\sup \Re f(k) = \Big| \frac{\alpha + \beta}{2} \Big|$.
In fact, WLOG assume $\alpha + \beta < 0$ and 
\[
\Re f(k) =\Re \ii k \sqrt{- \alpha \beta / k^2 -  (\alpha + \beta) / \ii k + 1} = - (\alpha + \beta ) / 2 = \Big|\frac{\alpha+\beta}{2}\Big| + o(1)
\]
which is exactly what needs to be proved.
\subsection{Sketch of proof of {\sc proposition} \ref{prop: essential spectrum qualitative}}
\label{app: proof of prop: essential spectrum qualitative}
Now the asymptotic operator $L_\infty = \Sigma \big( \pd_x - a\big) + A(x) $ by \eqref{eq: L infty}. $L_\infty$ is constant-coefficient for both $x < 0$ and $x > 0$. 
We can characterize its essential spectrum. We adapt Theorem 3.1.11 and Remark 3.1.14 of Kapitula and Promislow (2013) \cite{KaPr13}.
\begin{theorem}
[Essential spectrum of $L_\infty$]
\label{thm: essential spectrum of L infty}
$\lambda \in \sigma_\mathrm{e}(L_\infty)$ if and only if either of the following two statements is true:
\begin{enumerate}[label=(\roman*)]
    \item $\lambda \in \sigma\big(L_+\big) \cup \sigma \big(L_- \big)$
    \item The index of $\lambda - L_\infty$, defined by
    \begin{equation}
    \label{eq: index}
    \begin{aligned}
        \mathrm{ind}\big (\lambda I - L_\infty\big) := & \dim \mathbbm E^\mathrm{u}\big(\lambda I - L_-\big) - \dim \mathbbm E^\mathrm{u}\big(\lambda I - L_+\big)
    \end{aligned}
\end{equation}
  is not equal to zero. 
    \end{enumerate}
\end{theorem}
In the theorem above,
$\mathrm{dim} 
\mathbbm E^\mathrm{u} \big(\lambda I - L_-\big)$ is the number of independent vectors $\eta$ such that the following generalized eigenvalue problem is solved with some $\Re k > 0$:
\begin{equation*}
    L_- \eta e^{\ii k x} = \lambda \eta e^{\ii k x} 
\end{equation*}
similarly for $\mathbbm E^\mathrm u \big(\lambda I - L_+ \big)$. The following theorem adapted from the Theorem 3.1.13 of \cite{KaPr13} characterizes the border of $\sigma_\mr{e}\big(L_\infty\big)$:
\begin{theorem}
[Characterization of $\pd \sigma_\mr{e}\big(L_\infty\big)$]
\label{thm: border of essential spectrum}
\begin{enumerate}[label=(\roman*)]
    \item The border of the essential spectrum of $L_\infty$ is contained in the union of the essential spectra of $L_\pm$. Namely
\begin{equation*}
    \pd \sigma_\mr{e} \big( L_\infty \big) \subset \sigma\big(L_+\big) \cup \sigma \big(L_-\big)
\end{equation*}
\item  The set
\begin{equation*}
    \mathbbm C \setminus \Big[ \sigma\big(L_+\big) \cup \sigma \big(L_-\big)\Big]
\end{equation*} consists of connected components that are either entirely contained in $\sigma_\mr{e}\big(L_\infty\big)$ or does not intersect with it.
\end{enumerate}
\end{theorem}

We are interested in the essential spectrum of $L_{*,w}$. In fact, 
\begin{theorem}
\label{prop: essential spectra identical}
The essential spectra of an $L_{*,w}$ and $L_\infty$ are identical.
\end{theorem}
{\sc Theorem} \ref{prop: essential spectra identical} follows as an immediate corollary of the following {\sc Theorem} \ref{thm: Weyl} of Weyl
on the invariance of the essential spectrum of a linear operator under {\it relatively compact perturbations} \cite{EdEv18}, and {\sc proposition} \ref{prop: compact perturbation}:
\begin{theorem}[Weyl]
\label{thm: Weyl}
Let $L : \mathcal D(\mathcal L) \subset X\rightarrow Y$ be a closed operator between Banach spaces $X$ and $Y$, and $P$ relatively compact to $L$ then
\begin{equation*}
    \sigma_\mathrm{e}(L) = \sigma_\mathrm{e}(L+P)
\end{equation*}
\end{theorem}
We say $P$ is \textbf{relatively compact} with respect to $L$, if it is a compact operator on $\mathcal D(L)$ equipped with the \textbf{graph norm} to $Y$. 
For our case, the domain of $L_{*,w}$ is $H^1$, and the graph norm of $L_{*,w}$ is equivalent to $H^1$ norm since $L_{*,w} = \Sigma \pd_x + \text{some bounded matrix function}$, with $\Sigma$ being a invertible constant matrix.
To apply {\sc Theorem} \ref{thm: Weyl} we need the following proposition
\begin{proposition}
\label{prop: compact perturbation}
$L_{*,w}$ is a relatively compact perturbation of $L_\infty$. 
\end{proposition}
\begin{proof}
Note that $L_{*,w} - L_\infty = A_*(x) - A_\infty(x)$ where 
\[
A_\infty(z) =\big( A_- + w'(-\infty)) \mathbbm{1}_{(-\infty,0]}(x) + ( A_+ + w'(\infty) )\mathbbm{1}_{(-,\infty)}(x)
\]
So it can be identified with a piecewise smooth matrix function, 
with a type-1 discontinuity at $x = 0$.
The proof has two steps. First we prove that for any $N >0$, 
the operator defined by 
\[
\begin{aligned}
    & \delta L_N(x) = \left ( A_*(x) - A_\infty(x)\right) \mathbbm{1}_{[-N,N]}(x)
    \\= & \left ( A_*(x) - A_\infty(x)\right) \mathbbm{1}_{[-N,0]}(x) 
    +\left ( A_*(x) - A_\infty(x)\right) \mathbbm{1}_{(0,N]}(x) \\
    = & \delta L_N^- (x) + \delta L_N^+ (x)
\end{aligned}
\]
is compact relative to $L_\infty$. Let $(f_n)_n$ be an arbitrary sequence bounded in $H^1$ and as a result of the equivalence of $\|\cdot\|_{H^1}$ and the graph norm, 
it is also bounded in the graph norm. 
Note that 
\[\left( \delta L_N f_n\right)_n = \left( \delta L_N^- f_n\right)_n +
\left( \delta L_N^+ f_n\right)_n
\]
and $\delta L_N^\pm$ can be identified with a bounded smooth matrix function on $[0,\infty)$ and $(-\infty, 0]$, respectively.
So both sequences $\delta L_N^\pm f_n$ are bounded in $H^1$, therefore both admit convergent subsequences in $L^2$, as a result of the Rellich-Kondrachev compactness theorem\cite{Hahe2010kolmogorov};
so is their sum $\delta L_N f_n$. 
As a result $\delta L_N$ compact from $H^1$ to $L^2$.

Then we prove $(L_N)_N$ is convergent in norm as bounded operators from $H^1$ to $L^2$.
Let $N$ be a positive integer and we have \[ \|f\|_{L^\infty} \leq C \] from Sobolev embedding  with $C$ independent of $f$ as long as we fix $\|f\|_{H^1} = 1$.
Since $A_*(x) - A^\infty(x) \to 0$ exponentially fast as $x \to \pm\infty$, we have 
\[
\begin{aligned}
 & \| (\delta L_N - L_{*,w} + L_\infty) f \|_{L^2} \\
 \leq & \left\| (A_*(x) - A_\infty(x))f(x)\mathbbm{1}_{|x| > N} (x) \right\|_{L^2} \\
 \leq & C \left\|(A_*(x) - A_\infty(x))  \mathbbm{1}_{|x| > N} (x) \right\|_{L^2} \| f\|_{L^\infty} \\
 \leq & C \left\|(A_*(x) - A_\infty(x))  \mathbbm{1}_{|x| > N} (x) \right\|_{L^2} \\
\leq & C e^{- \mu N} 
\end{aligned}
\]
for some $C, \mu > 0$ independent of $N$, $ b$ and vanishes as $N \rightarrow \infty$.
Thus $(\delta L_N)_N$ converges to $L_{*,w} -L_\infty$ 
in operator norm of $\mathcal B \left( H^1 , L^2\right)$, the space of bounded linear operators from $H^1$ to $L^2$.
Since $L_{*,w} - L_\infty$ is compact in this space, 
it is compact relative to $L_\infty$ and the proof is complete.
\end{proof}
\subsection{Proof of {\sc lemma} \ref{lemma: c = 0 both L 2}}
\label{app: proof of lemma: c = 0 both L 2}
It is equivalent to proving that the weight in \eqref{eq: B tilde transform 0} is bounded for all $x\in \mathbbm R$.
Equivalently, the exponent
\begin{equation}
\label{eq:expon}
     w(x) + \int_{-\infty}^x \mathcal N_p(y) \dd y  
\end{equation}
is bounded uniformly for $x \in \mathbbm R$. 
Note that $w(x)=0$ for $x\leq -1$; also $r_0(x)\to 0$ exponentially as $y\to -\infty$, so $\mathcal N_p(x)= \mathcal N'\big(r_0(x)^2 \big) r_0(x)^2 \to 0$ exponentially as $x\to-\infty$ as well since $\mathcal N'(r^2)$ is bounded for all $r^2 \geq 0$. 
As a result, the expression in \eqref{eq:expon} is bounded for all large and negative $x$. 
For $x \geq 1$, there is $w(x)=Kx$.
Also, $\mathcal N_p(x)= \mathcal N'\big(r_0(x)^2 \big) r_0(x)^2 \to -K$ exponentially fast as $x \to \infty$, since $\mathcal N''(r^2)$ is also bounded near $r^2 = 1$ and $r_0(x)^2 \to 1$ exponentially fast as $x \to \infty$. 
Thus, for some $X_0>0$ and all $x\ge X_0$, we have:
\[w(x) + \int_{-\infty}^x \mathcal N_p(y) \dd y = Kx + \int_{-\infty}^{X_0} N_p(y) dy + \int_{X_0}^x \left(-K+\mathcal{O}(e^{-\gamma X})\right) dX,\]
for some $\gamma>0$ which is bounded for all $x\ge X_0$.
as $x \to \infty$.
\subsection{Proof of {\sc lemma} \ref{lemma: schroedinger}}
\label{app: proof of lemma: schroedinger}
Written out, \eqref{eq: reduced GEP c = 0} is 
\begin{equation}
\begin{aligned}
U_x & = \big[\mathcal N_0(x) + \mathcal N_p (x)\big] U + \lambda V
\\
V_x & = \lambda U - \big[\mathcal N_0(x) + \mathcal N_p (x)\big] V
\end{aligned}
\label{eq:writout}
\end{equation}
It is easy to see in fact $U$ is twice continuously differentiable.
Now, differentiate the first equation with respect to $x$ and using the second to eliminate $V_x$:
\[
U_{xx}  = \big[\mathcal N_0(x) + \mathcal N_p (x)\big]_x U + \lambda ^ 2 U + \big[\mathcal N_0(x) + \mathcal N_p (x)\big]^2 U
\]
or \[
    \big[-\pd_x^2 + \mathcal V(x) \big] U (x) = \mathcal E U(x),\ \mathcal E = -\lambda^2  \] where
is a time-independent Schr\"{o}dinger equation with a real-valued potential
\[
    \mathcal V(x) =  \big[\mathcal N_0(x) + \mathcal N_p (x)\big]^2 + \big[\mathcal N_0(x) + \mathcal N_p (x)\big]_x.
\]
Let $\mathcal L = -\pd_x^2 + \mathcal V(x)$.
We observe that $U\ne0$ if $\lambda \neq 0$; otherwise if $\lambda\ne0$ and $U=0$, it would follow from \eqref{eq:writout} that $V=0$, 
which contradicts the assumption that $B\ne0$.

Now we prove that $\mathcal E = 0$ is an eigenvalue of $-\pd_x^2 + \mathcal V(x)$.
From the definition of $\mathcal N_0$ and $\mathcal N_p$ in \eqref{eq: N 0 and N p}, 
\[
\mathcal N_0(-\infty) + \mathcal N_p(-\infty) = 1 > 0 > \mathcal N_0(\infty) + \mathcal N_p(\infty) = -K\]
so there is an $x'\in \mathbbm R$ at which hold that $\mathcal N_0(x') + \mathcal N_p(x') = 0$ and that its derivative is negative.
At $x'$, there is 
\[
\begin{aligned}
    \mathcal V(x') = &   \big[\mathcal N_0(x') + \mathcal N_p (x')\big]^2 + \Big( \pd_x \big[\mathcal N_0(x) + \mathcal N_p (x)\big]\Big)_{x=x'} \\
    = & 0 + \Big( \pd_x \big[\mathcal N_0(x) + \mathcal N_p (x)\big]\Big)_{x=x'} < 0
\end{aligned}\]
Therefore $\min_{x \in \mathbbm R} \mathcal V (x) <0$.

Now $\mathcal E = - \lambda^2 =  0$ is an eigenvalue of $\mathcal L$. 
In fact, in this case, $\lambda = 0$ and from \eqref{eq:writout} we see $V(x) \equiv 0$ since as $x \to -\infty$,  
\[
- \mathcal N_0(-\infty) - \mathcal N_p(-\infty) = -1\] 
so $V$ cannot be bounded as $x \to -\infty$, unless it vanishes for all $x \in \mathbbm R$, 
and there is a unique $U(x)$ given by (up to a constant):
\begin{equation}
\label{eq: U solution}
U(x) = U_0(x) = \exp \bigg( \int^x \mathcal N_0(x) + \mathcal N_p(x) \bigg)
\end{equation}
that solves the first of the (now decoupled) equations in \eqref{eq:writout}. 
In fact, direct differentiation yields:
\[
U_{0,xx}(x) = \Big[\big(\mathcal N_0 + \mathcal N_p\big)U_0 \Big]_x = \big(\mathcal N_0 + \mathcal N_p\big)_x U_0 + \big(\mathcal N_0 + \mathcal N_p\big)^2 U_0
\]
Now for any $x\in \mathbbm R$ there is $U(x) \neq 0$, otherwise, there would be $U \equiv 0$. 
So $U(x)$ has \textit{no nodes} and thus it is the \underline{ground state} of $\mathcal L$,
and $\mathcal E = 0$ is the (nondegenerate) ground state energy.
\subsection{Proof of {\sc proposition} \ref{prop: B tilde B transform}}
\label{app: proof of prop: B tilde B transform}
We will express formulas with Pauli matrices in this proof. 
See \eqref{eq: pauli} and \eqref{eq: pauli commutation}.
Throughout this proof, $\theta = \theta_{c,0} = -\frac{1}{2}\arcsin c$.

We conduct an intermediate transform on the unknown function $\tilde B$:
\begin{equation}
\label{eq: B tilde beta transformation}
\tilde B \big(\sqrt{1-c^2} X \big) = \big(\sigma_0 \cos \theta + \sigma_1 \sin \theta\big)  \beta (X)
\end{equation} 
for $X \in \mathbbm R$, which is exactly \eqref{eq: B tilde beta transformation text}.

We must take care when plugging \eqref{eq: B tilde beta transformation} into \eqref{eq: ep c}. 
Consider a generic differential equation of the form $\pd_x f(x) = g(x)$
which holds for all $x \in \mathbbm R$. Consequently we have 
\[
    \frac{1}{\sqrt{1-c^2}}\pd_X f \big(\sqrt{1-c^2}X \big) = g \big(\sqrt{1-c^2}X \big)
\]
for all $X \in \mathbbm R$. 
So write $L_{*,w} \big(\pd_x, x\big) $, \eqref{eq: ep c} is equivalent to 
\begin{equation}
\label{eq: PDE scaling transform}
L_{*,w} \bigg(\frac{1}{\sqrt{1-c^2}} \pd_X, \sqrt{1-c^2} X\bigg) \tilde B \big(\sqrt{1-c^2}X\big) = \lambda \tilde B \big(\sqrt{1-c^2}X\big)
\end{equation}
For convenience we further define
\begin{equation}
v(X) =\sqrt{1-c^2} w_c'(x)|_{x = \sqrt{1-c^2} X} = \frac{\dd w_c\big(\sqrt{1-c^2} X \big) }{\dd X} ,\quad \Lambda = \frac{\lambda}{\sqrt{1-c^2}}
\label{eq: definition v(X) and Lambda}
\end{equation}
For $b_* (x) = \begin{bmatrix}
    u_*(x) & v_*(x)
\end{bmatrix} = b_{c,0}(x)$ 
satisfying \eqref{eq: kink profile ODE}, since $u_*(x) = r_*(x) \cos \theta$ for all $x \in \mathbbm R$ and $r_*(0) =1/2$, we have 
\begin{equation}
\label{eq: kink profile ODE c}
    r_*'(x) = \frac{r_*(x) \mathcal N \big( r_*(x)^2 \big) }{\sqrt{1-c^2}}
\end{equation}
Note that for $c = 0$, $r_0 (x) = u_0(x)$ satisfies
\begin{equation}
    \label{eq: kink profile ODE 0}
    r_0'(x) = r_0(x) \mathcal N \big( r_0(x)^2 \big)
\end{equation}
and 
$r_0 \Big( \frac{x}{\sqrt{1-c^2}} \Big)$ satisfies
\begin{equation}
\label{eq: kink profile ODE 0 temp}
    \frac{\dd }{\dd x} r_0 \Big( \frac{x}{\sqrt{1-c^2}} \Big) = \frac{1}{\sqrt{1-c^2}} r_0'\Big( \frac{x}{\sqrt{1-c^2}} \Big)
\end{equation}
Combining \eqref{eq: kink profile ODE 0} and \eqref{eq: kink profile ODE 0 temp}, 
we see $x \mapsto r_0 \Big( \frac{x}{\sqrt{1-c^2}} \Big)$ satisfies \eqref{eq: kink profile ODE c} and have value $r_0(0/\sqrt{1-c^2} ) = r_0(0) = 1/2$ at $x = 0$.
Therefore for all $x\in \mathbbm R$, 
\[
    r_*(x) =  r_0 \Big( \frac{x}{\sqrt{1-c^2}} \Big) 
\]
and equivalently for all $X \in \mathbbm R$, 
\begin{equation}
    \label{eq: kink scaling}
    r_*\big(\sqrt{1-c^2} X \big) =  r_0 ( X )
\end{equation}
Now we transform \eqref{eq: ep c} according to \eqref{eq: PDE scaling transform}. 
In the first identity below we used \eqref{eq: kink scaling}:
\begin{equation}
\label{eq: L * w expression}
\begin{aligned}
    & L_{*,w} \bigg(\frac{1}{\sqrt{1-c^2}} \pd_X, \sqrt{1-c^2} X\bigg) = \big( c \sigma_0 + \sigma_1\big) \bigg( \frac{1}{\sqrt{1-c^2}} \pd_X - w_c'(x)|_{x = \sqrt{1-c^2}X} \bigg)  \\
    & +
    \begin{bmatrix}
    {\mathcal N'\big(r_*( x )^2 \big)} r_*(x)^2 \sin 2\theta
    & \mathcal N_* \big(r_*(x)^2 \big) +  {2 \mathcal N' \big(r_*( x )^2 \big)} r_*(x)^2 \sin^2 \theta \\
     -\mathcal N \big(r_*( x )^2 \big)- {2\mathcal N' \big(r_*( x )^2 \big)}r_*(x)^2 \cos^2 \theta & -{\mathcal N' \big(r_*( x )^2 \big)}r_*(x)^2 \sin 2\theta 
    \end{bmatrix}_{x = \sqrt{1-c^2}X}\\
    =&  \frac{\big( c \sigma_0 + \sigma_1\big)}{\sqrt{1-c^2}} \big( \pd_X - v(X) \big)  \\
    & +\begin{bmatrix}
    {\mathcal N'\big(r_0( X )^2 \big)} r_0(X)^2 \sin 2\theta
    & \mathcal N \big(r_0(X)^2 \big) +  {2 \mathcal N' \big(r_0( X )^2 \big)} r_0(X)^2 \sin^2 \theta \\
     -\mathcal N \big(r_0( X )^2 \big)- {2\mathcal N' \big(r_0( X )^2 \big)}r_0(X)^2 \cos^2 \theta & -{\mathcal N' \big(r_0( X )^2 \big)}r_0(X)^2 \sin 2\theta 
    \end{bmatrix}
    \\
    = &\frac{\big( c \sigma_0 + \sigma_1\big)}{\sqrt{1-c^2}} \big( \pd_X - v(X) \big)  
    - \sigma_1 \mathcal N_p ( X ) \cos 2 \theta + \ii \sigma_2  \big[ \mathcal N_0 (X) + \mathcal N_p( X ) \big]
    + \sigma_3 \mathcal N_p (X) \sin 2\theta \\
    = & \frac{\big( c \sigma_0 + \sigma_1\big)}{\sqrt{1-c^2}} \big( \pd_X - v \big)  
    - \sigma_1 \mathcal N_p \sqrt{1-c^2} + \ii \sigma_2  \big( \mathcal N_0 + \mathcal N_p \big)
    - \sigma_3 \mathcal N_p c 
\end{aligned}
\end{equation}
In the last line we have used $\theta = -\frac12 \arcsin c$. 
Using \eqref{eq: L * w expression} above, plug the transformation \eqref{eq: B tilde beta transformation} relating $\tilde B$ with $\beta$ into \eqref{eq: PDE scaling transform}, \eqref{eq: ep c} is equivalent to 
\begin{equation}
\label{eq: transformed c neq 0 GEP writeout}
\begin{aligned}
& \frac{c \sigma_0  + \sigma_1 }{\sqrt{1-c^2}}
\big(\pd_X - v \big) \big(\sigma_0 \cos \theta + \sigma_1 \sin \theta\big) \beta 
- \sigma_1 \mathcal N_p \sqrt{1-c^2} 
\big(\sigma_0 \cos \theta + \sigma_1 \sin \theta\big) \beta  
\\
&+ \ii \sigma_2 \big(\mathcal N_0 + \mathcal N_p \big) 
\big(\sigma_0 \cos \theta + \sigma_1 \sin \theta\big)
\beta - \sigma_3 c \mathcal N_p \big(\sigma_0 \cos \theta + \sigma_1 \sin \theta\big) \beta
\\
= &\frac{c \sigma_0  + \sigma_1 }{\sqrt{1-c^2}}
\big(\pd_X - v \big) \big(\sigma_0 \cos \theta + \sigma_1 \sin \theta\big) \beta 
- \sigma_1 \mathcal N_p \sqrt{1-c^2} 
\big(\sigma_0 \cos \theta + \sigma_1 \sin \theta\big) \beta 
\\
&+ \ii \big(\mathcal N_0 + \mathcal N_p \big) 
\big(\sigma_0 \cos \theta - \sigma_1 \sin \theta\big) \sigma_2
\beta -  c \mathcal N_p \big(\sigma_0 \cos \theta - \sigma_1 \sin \theta\big) \sigma_3 \beta
\\
= & \lambda\big(\sigma_0 \cos \theta + \sigma_1 \sin \theta\big) \beta 
\end{aligned}
\end{equation}
In the last identity we have used the fact that $\sigma_i$ and $\sigma_j$ anti-commute for $i, j \in \{1,2,3\}$ and $i \neq j$. See \eqref{eq: pauli commutation}.
We would like to obtain a differential equation in which there is no coefficient in front of $\pd_X$.
For this, we multiply on both sides (on the left) by
\begin{equation}
\label{eq: inverses}
    \sqrt{1-c^2} \big( \sigma_0 \cos \theta + \sigma_1 \sin \theta \big)^{-1} \big( c \sigma_0 + \sigma_1\big)^{-1} 
\end{equation}
Note that $ \sigma_0 \cos \theta + \sigma_1 \sin \theta $ and $ c \sigma_0 + \sigma_1 $ commute,
therefore so does their inverses. 
Explicitly, we have:
\[
\big( c \sigma_0 + \sigma_1 \big)^{-1} = \frac{-c \sigma_0 + \sigma_1}{1-c^2}
\]
and, note that $\cos^2 \theta - \sin^2 \theta = \cos 2\theta = \sqrt{1-c^2}$, there is also
\[
\big(\sigma_0 \cos \theta + \sigma_1 \sin \theta \big)^{-1} = \frac{\sigma_0 \cos \theta - \sigma_1 \sin \theta }{\sqrt{1-c^2}}
\]
Now multiplying to the left by \eqref{eq: inverses} on both sides of the last identity of
\eqref{eq: transformed c neq 0 GEP writeout}:
\begin{equation}
\label{eq: transformed c neq 0 GEP}
    \begin{aligned}
    & \big(\pd_X - v \big) \beta - \big(-c\sigma_0 + \sigma_1\big)
    \mathcal N_p
    \sigma_1 \beta \\
    & + \ii \frac{-c\sigma_0 + \sigma_1}{1-c^2} \big( \sigma_0 \cos \theta - \sigma_1 \sin\theta \big)^2 
    \big(\mathcal N_0 + \mathcal N_p\big) \sigma_2 \beta \\
    & - \frac{-c\sigma_0 + \sigma_1}{1-c^2} \big( \sigma_0 \cos \theta - \sigma_1 \sin\theta \big)^2
    c \mr N_p \sigma_3 \beta
    \\
    = & \frac{-c \sigma_0 + \sigma_1}{\sqrt{1-c^2}} \lambda \beta = \big(-c \sigma_0 + \sigma_1\big) \Lambda \beta
    \end{aligned}
\end{equation}
The second term of LHS of \eqref{eq: transformed c neq 0 GEP} is 
\begin{equation}
\label{eq: transformed c neq 0 GEP term 2}
- \big(-c\sigma_0 + \sigma_1\big)
    \mathcal N_p
    \sigma_1 \beta = \big( - \sigma_0 + c \sigma_1 \big) \mathcal N_p \beta 
\end{equation}
the third term:
\begin{equation}
\label{eq: transformed c neq 0 GEP term 3}
\begin{aligned}
 & \ii \frac{-c\sigma_0 + \sigma_1}{1-c^2} \big( \sigma_0 \cos \theta - \sigma_1 \sin\theta \big)^2 \big(\mathcal N_0 + \mathcal N_p\big) \sigma_2 \beta \\
 = & \ii \frac{-c\sigma_0 + \sigma_1}{1-c^2} \big(\sigma_0 - \sigma_1 \sin 2\theta\big)
 \big(\mathcal N_0 + \mathcal N_p \big) \sigma_2 \beta \\
 = &\ii \frac{-c\sigma_0 + \sigma_1}{1-c^2} \big(\sigma_0 + \sigma_1 c\big) \big(\mathcal N_0 + \mathcal N_p \big) \sigma_2 \beta
 = \ii \sigma_1 \big(\mathcal N_0 + \mathcal N_p \big) \sigma_2 \beta \\
 = & -  \big(\mathcal N_0 + \mathcal N_p \big) \sigma_3 \beta
 \end{aligned}
\end{equation}
where again $\sin 2\theta = -c$.
The fourth term is 
\begin{equation}
\label{eq: transformed c neq 0 GEP term 4}
\begin{aligned}
    & - \frac{-c\sigma_0 + \sigma_1}{1-c^2} 
    \big( \sigma_0 \cos \theta - \sigma_1 \sin\theta \big)^2
    c \mathcal N_p \sigma_3 \beta
    \\
    = & \frac{c\sigma_0 - \sigma_1}{1-c^2} \big(\sigma_0 + \sigma_1 c\big) c \mathcal N_p \sigma_3 \beta
    \\
    = & - \sigma_1 c \mathcal N_p \sigma_3 \beta
    = \ii c \mathcal N_p \sigma_2
\end{aligned}
\end{equation}
Combining \eqref{eq: transformed c neq 0 GEP}, \eqref{eq: transformed c neq 0 GEP term 2}, \eqref{eq: transformed c neq 0 GEP term 3} and \eqref{eq: transformed c neq 0 GEP term 4}:
\[ \big(\pd_X - v \big) \beta + \big( - \sigma_0 + c \sigma_1 \big) \mathcal N_p \beta 
-  \big(\mathcal N_0 + \mathcal N_p \big) \sigma_3 \beta +\ii c \mathcal N_p \sigma_2
= \big(-c \sigma_0 + \sigma_1\big) \Lambda \beta
\]
which is
\begin{equation}
\pd_X  \beta = \sigma_0 \big( v - c \Lambda + \mathcal N_p\big) 
\beta + \Lambda \sigma_1 \beta
- c \mathcal N_p \big( \sigma_1 + \ii \sigma_2 \big)  \beta
+ \big(\mathcal N_0 + \mathcal N_p \big) \sigma_3  \beta 
  \label{eq:beta-eq}  
\end{equation}
Further, we let 
\begin{equation}
\label{eq: transform B beta}
\begin{aligned}
    B (X) & := e^{-c\Lambda X }
    \exp\bigg( \int_{-\infty}^X v(Y) + \mathcal N_p(Y) \dd Y \bigg)  \beta (X)  \\
    & = e^{-c\Lambda X }
    \exp\bigg( w_c\big(\sqrt{1-c^2} X \big) + \int_{-\infty}^X \mathcal N_p(Y) \dd Y \bigg)  \beta (X)
\end{aligned}
\end{equation}
which is exactly \eqref{eq: transform B beta text}.
Plugging \eqref{eq: transform B beta} into \eqref{eq:beta-eq} we get rid of the first first term on the RHS of \eqref{eq:beta-eq}: 
\[
    \pd_X  B = \Lambda \sigma_1 B
- c \mathcal N_p \big( \sigma_1 + \ii \sigma_2 \big)  B
+ \big(\mathcal N_0 + \mathcal N_p \big) \sigma_3  B 
\]
Equivalently, $B=\begin{bmatrix}
    U & V
\end{bmatrix}^\mathsf T$ satisfies
\begin{equation}
\label{eq: reduced c neq 0 U V system}
    \begin{aligned}
     U_X & = \big(\mathcal N_0 + \mathcal N_p \big)  U + \big( \Lambda -2c\mathcal N_p\big) V \\
     V_X & = \Lambda U - \big(\mathcal N_0 + \mathcal N_p\big) V
    \end{aligned}.
\end{equation}
which is exactly \eqref{eq: UV}, 
and the proof of {\sc proposition} \ref{prop: B tilde B transform} is done.
\subsection{Proof of {\sc lemma} \ref{lem: c neq 0 kink unbounded}}
\label{app: proof of lem: c neq 0 kink unbounded}
By assumption $L_{*,w} \tilde B(x) = \lambda \tilde B(x)$.
Since the coefficients of $L_{*,w}$ are all bounded and smooth with respect to $x$, 
$B(x)$ is actually smooth. 
Moreover, $\lambda$ with $\Re \lambda > 0$ is not in the essential spectrum of $L_{*,w}$ since $a = \frac{K}{\sqrt{1-c^2}}$ satisfies the condition in {\sc proposition} \ref{prop: essential kink}.
Hence there are constants $\mu_\pm$ with $\Re \mu \neq 0$ such that $\big|\tilde B(x)\big| = C e^{\mu x} (1 + o(1))$ as $x \to \infty$; similarly for $x \to -\infty$. 
As a result, $\tilde B \in L^2$ if and only if $\tilde B (x) = o(1) $ as $|x| \to \infty$.
From the transform \eqref{eq: B tilde beta transformation} relating $\tilde B$ and $\beta$, $\tilde B \in L^2$ if and only if $\beta \in L^2$, if and only if $\beta (X) = o(1) $ as $|X| \to \infty$.

Now assume $\beta \in L^2$ and rewrite \eqref{eq: transform B beta} as
\[
    B(X) = e^{g(X)}  \beta(X)
\]
where 
\[
g(X) = - c \Lambda X + w_c \big(\sqrt{1-c^2} X \big) + \int_{-\infty}^X \mathcal N_p(Y) \dd Y
\]
For $X \geq \frac{1}{\sqrt{1-c^2}}$, $w_c\big(\sqrt{1-c^2}X\big) = K$, see \eqref{eq: definition v(X) and Lambda};
and $\mathcal N_p(X) = \mathcal N'\big(r_0(X)^2\big) r_0(X)^2$ approaches $-K$ exponentially fast as $X \to \infty$. 
So $g(X) = -c\Lambda X + \mathcal O(1)$
as $X \to \infty$.
Since $\Re \Lambda >0$, 
if $\beta\in L^2$, there must be $B(X) \to 0$ (exponentially fast) as $X \to \infty$.

On the other hand $g(X)  = -c \Lambda X + o(1)$ as $X \to -\infty$ since $w_c\big(\sqrt{1-c^2}X\big)=0$ for $X \leq \frac{1}{\sqrt{1-c^2}}$, 
and $\mathcal N_p(X) \to 0$ exponentially fast. 
Since \eqref{eq: reduced c neq 0 U V system} is also exponentially asymptotically constant, as $X \to -\infty$, $B(X)$ also behaves exponentially.
There must be $B(X) \sim e^{\mu X}$
where $\mu$ satisfies the following, obtained by taking limits of the coefficients of \eqref{eq: reduced c neq 0 U V system} (note that $\mathcal N_0(-\infty) = 1$):
\[
    \det \begin{bmatrix}
    1 - \mu & \Lambda \\ \Lambda & - 1 - \mu \end{bmatrix} = 0
\]
namely $\mu = \pm \sqrt{1 + \Lambda^2}$. Since $\beta(X) = o(1)$ as $X \to -\infty$, 
\begin{equation}
\label{eq: B X asymp}
    B(X) = o \big(e^{-c\Lambda X}\big), \quad \text{as $X \to -\infty$}
\end{equation} 
This forces $B(X) \sim e^{\sqrt{1+\Lambda^2} X}$ since otherwise there must be $B(X) \sim e^{-\sqrt{1+\Lambda^2} X}$.
Due to the following relation
\begin{equation}
    \label{eq: B X order comparison}
    e^{-c\lambda X} = o \Big( e^{- \sqrt{\Lambda^2 + 1}X } \Big),\quad \text{as $X \to -\infty$}
\end{equation}
so condition \eqref{eq: B X asymp} is violated.
To prove \eqref{eq: B X order comparison}, 
note that $\Re \big(-\sqrt{1 + \Lambda^2} \big)  < \Re ( - \Lambda)$ 
since $\Re \Lambda > 0$, by the following elementary fact:
\begin{equation}
\label{eq: sun guanhao}
    \Re \sqrt{z^2 + 1 } > \Re z, \quad \text{for $\Re z > 0$}
\end{equation}
and thus 
\[
    \Re \big(c \Lambda - \sqrt{1+\Lambda^2} ) < \Re \big(c \Lambda - \Lambda ) < 0
\]
with $0\leq c < 1$, so $\Re (-c \Lambda) >  \Re \big(-\sqrt{\Lambda^2+1} \big)$, and $\Re \big( -c\Lambda X\big)  < \Re \big(-\sqrt{\Lambda^2+1} X \big)$ for $X < 0$. 
As a result \eqref{eq: B X order comparison} holds.

So $B(X) \sim  e^{\sqrt{1+\Lambda^2} X} \to 0$ as $X \to -\infty$.
Therefore if $\beta\in L^2$, there must be $B(X) \to 0$ exponentially fast as $|X| \to \infty$, which implies $B \in L^2$.

Now we prove inequality \eqref{eq: sun guanhao} \footnote{The authors thank Dr. SUN Guanhao of UCSD for pointing out this is not a trivial fact and for providing the following proof. }
\begin{proof}[Proof of \eqref{eq: sun guanhao} ]
Note that for any $z \in\mathbbm C$,
\begin{equation}
\label{eq: real part of z expression}
    \Re z = \frac{1}{2} \bigg( z + \frac{z\bar z}{z} \bigg).
\end{equation}
Write $z = r e^{\ii \theta}$. Since $\Re z > 0 $, where $-\pi/2 <\theta < \pi/2$. 
Equation \eqref{eq: sun guanhao} becomes
 $r\cos\theta<\Re \sqrt{r^2e^{2i\theta}+1}$. 
Dividing by $r$, we find that \eqref{eq: sun guanhao} is equivalent to  \[
\Re \sqrt{r^{-2} + e^{\ii 2 \theta} } > \cos \theta,\quad -\pi/2 <\theta < \pi/2. \] 
Note that the real part of $\sqrt{\cdot}$ is always nonnegative, it is equivalent to prove 
\[2 \Big[\Re \sqrt{s + e^{\ii 2 \theta} } \Big]^2 - 2 \cos^2 \theta > 0, \quad s \equiv r^{-2} > 0, \quad -\pi/2 <\theta < \pi/2.\]
Using \eqref{eq: real part of z expression}, the LHS satisfies
\[
\begin{aligned}
& \frac{1}{2}\left[ \Big(s + e^{\ii 2 \theta} \Big)^{1/2} + \frac{\Big(s^2 + 1 + 2 s \cos 2\theta \Big)^{1/2} }{\Big( s+ e^{\ii2\theta}\Big)^{1/2}} \right]^2  - 2 \cos^2 \theta 
\\ =& 2 \times \frac{1}{4}\Big[  s+e^{\ii 2 \theta} + 2 \big( s^2+1+2s\cos 2\theta\big)^{1/2} + s + e^{-\ii2\theta} \Big] - 2 \cos^2 \theta \\
=& s + \cos 2\theta  + \big( s^2 +1+2s \cos 2\theta\big) ^{1/2}  - \cos 2\theta - 1 \\
 = &  s + \big( s^2 +1+2s \cos 2\theta\big) ^{1/2} -1 > s + |s-1|-1 \geq 0
\end{aligned}
\]
with $s>0$ and $\cos 2\theta  \neq -1$ since the latter requires $\theta = \pm \pi/2$, contradicting the requirement that $\Re z = r\cos \theta > 0$. 
Thus we have concluded the proof of \eqref{eq: sun guanhao}.
\end{proof}

\bibliographystyle{plain}
\bibliography{references}
\end{document}